\documentclass[acmsmall,nonacm,screen,]{acmart}
\usepackage{macros}
\extendedtrue


\title{Bidirectional Elaborators \`a la Carte}

\author{Andrew Slattery}
\email{aws46@cam.ac.uk}
\orcid{0009-0008-4933-821X}
\affiliation{%
  \department{Computer Laboratory}
  \institution{University of Cambridge}
  \city{Cambridge}
  \country{United Kingdom}
}

\author{Jonathan Sterling}
\email{js2878@cl.cam.ac.uk}
\orcid{0000-0002-0585-5564}
\affiliation{%
  \department{Computer Laboratory}
  \institution{University of Cambridge}
  \city{Cambridge}
  \country{United Kingdom}
}

\begin{document}

\begin{abstract}
Surface syntax in proof assistants like Rocq, Lean, Agda, and Idris is highly implicit, lacking many details that are needed for user-written code to denote precisely defined mathematical objects. \emph{Elaboration} is an algorithm that accounts for these details by translating surface syntax to an explicit enough core syntax. The reliability and predictability of elaboration relies on several critical properties of the core type system, including \emph{decidability of judgemental equality} and the \emph{injectivity of type constructors}; these dependencies are witnessed in a concrete system by explicit calls to conversion checking and weak-head reduction subroutines.

We introduce a dependently typed monadic domain specific language for the executable specification of correct-by-construction elaboration algorithms that is abstracted from any particular representation of normal forms or algorithm for conversion checking. In particular, we represent a bidirectionally typed surface language for Martin-L\"of type theory by shallow embedding in this DSL so that the translation of surface terms into core terms amounts to elementary equational calculation. This translation is correct by construction in the sense that it cannot produce ill-typed terms, and is automatically stable under judgemental equality of core terms and even under substitution; from the latter property, we obtain a new denotational interpretation of the \emph{suspension} of elaboration problems. Finally, a concrete elaboration algorithm is extracted by algebraic means from a presheaf model of the DSL built out of the \emph{bi-initial natural model} of Martin-L\"of type theory. \end{abstract}

\maketitle

\section{Introduction}

Present-day proof assistants based on dependent type theory (\emph{e.g.}\ Rocq, Lean, Agda, and Idris) are all based on \emph{elaboration}: rather than merely ``checking'' that user-written code is well-formed, elaboration is a process that transforms user-written code to a more explicit core syntax \emph{or} fails with an error.
The difference between surface syntax and its core representation is that the latter must contain all the information needed to precisely describe and distinguish a well-defined mathematical object, whereas surface syntax often omits annotations and coercions that would be burdensome for a programmer to write. The design of surface and core syntax are dialectically linked, conditioned both by ergonomic and mathematical concerns:
\begin{enumerate}
  \item \emph{Ergonomics:} it is desirable for a programmer to avoid typing in bureaucratic annotations, arguments, and coercions when these can be reliably and predictably inserted by a machine.
  \item \emph{Mathematics:} whatever information is made implicit in the surface syntax must be unambiguously and effectively reconstructible: this depends on mathematical metatheorems about the core language, \emph{e.g.} normalisation, decidability, strengthening, \emph{etc.}
\end{enumerate}

Experts know well what kinds of decisions in core language design might impede the implementability of a reliable surface language, and (conversely) what constraints on core languages may be required in order to enable  various desirable surface language features. This folklore has not, however, prevented the promulgation of high-profile surface languages such as Lean's in which (for example) the relation of \emph{definitional equality} fails to be transitive;\footnote{Lean's language reference~\citep{lean4-reference-typesystem} comments that its definitional equality is ``reflexive and symmetric, but not transitive''.} for a user, such a failure might result in needing more than one proof step to verify $x=y$ even when $x$ and $y$ are related by a sequence of definitional equations. Scenarios like this challenge the objectivity of proof assistants, which rests on a carefully negotiated division of labour between operator and machine.

Although questionable designs like this may be avoidable with a bit of care, we believe that the bigger problem is an air of ``black magic'' between core language metatheory and the design of suitable surface languages. We aim to address this gap by introducing a flexible mathematical framework in which a surface language is described modularly in terms of \textbf{\emph{elaboration combinators}} in a monadic domain specific language (DSL), similar to the idea of parser combinators~\citep{hill:1996,hutton-meijer:1998}. Then, simple-minded calculational reasoning in the metalanguage may be used to reason about the behaviour of surface language code under elaboration.

Our DSL's metalanguage is \emph{extensional} Martin-L\"of type theory~\citep{hofmann:1997},\footnote{Everything in this paper can be re-done (with a bit of work) in \emph{intensional} type theory. Doing so would detract from our informal exposition, but it will be necessary for the computerised implementation of our framework in future work.} whose dependent types are used to ensure that elaboration combinators are correct by construction in the following sense:

\begin{enumerate}
  \item Elaboration always terminates and can only produce well-typed terms.
  \item Elaboration cannot distinguish definitionally equal core language terms, and therefore respects the full $\beta/\eta$-congruence of the core language.
  \item Elaboration is stable under core language substitution; elaboration problems can therefore be solved in any order without changing the end result.
\end{enumerate}

Our viewpoint on the desirable properties of elaboration is similar to that of Kov\'acs, who has recently argued for \emph{stable elaboration} as a design principle~\citep{kovacs2ltt-principles}. We admit that no existing production system satisfies all of these desiderata, but our experience building and using interactive proof assistants over many years inspires our conviction that these are both desirable \emph{and} achievable---with a bit of mathematics and a bit of engineering. This paper aims to provide the mathematics.

\subsection{Background}

\subsubsection{Algebraic core languages and metatheory}\label{sec:intro:algebraic-core-language}

In its history, the \emph{syntax} of type theory has been treated in a variety of ways that span from the very concrete (inference rules over terms generated by a grammar) to the very abstract (anything satisfying the universal property of the initial model). We prefer to ground the discussion in the structural r\^ole of syntax in the \emph{applications} of type theory---which are mainly in the area of proof assistants and programming languages. To be useful in these areas, an approach to syntax must be compatible with the statement and verification of key metatheoretic lemmas including the \emph{injectivity of type constructors} and \emph{decidability of judgemental equality}, both of which follow from \emph{normalisation}.

For many years, metatheoretic results of these kinds were primarily established with respect to extremely concrete descriptions of syntax---at great cost. In the past decade, however, the category theoretic \emph{glueing method}, long well-understood for simple types~\citep{altenkirch-hofmann-streicher:1995,fiore:2002}, has finally been adapted to the far more difficult case of normalisation for dependent types, first by \citet{coquand:2019} and later refined by \citet{sterling:2021:thesis,sterling:2025:grothendieck} and \citet{bocquet:2026}. With these streamlined methods in hand, the metatheory of dependent types can be established directly on \emph{any} representation of the initial model of the type theory viewed as a \textbf{\emph{second-order generalised algebraic theory}}~\citep{uemura:2021:thesis} or \emph{SOGAT}.

\paragraph{Second-order generalised algebraic theories}

A second-order generalised algebraic theory is described by a signature that interleaves sorts, operations, and equations in a restricted version of extensional type theory. Dependent function sorts for operator arguments are available only when the domain sort is designated as \emph{representable}.
For example, in the signature for Martin-L\"of type theory we have two generating sorts
\[
  \Tp:\SortIdent{Sort}\text{,}
  \qquad
  \Tm\colon \Tp\to\SortIdent{Sort}_{\mathit{repr}}
  \text{.}
\]
The sort of types $\Tp:\SortIdent{Sort}$ is not representable because bound variables cannot range over types, whereas the dependent sort of terms $\alpha:\Tp\vdash \Tm\prn{\alpha}:\SortIdent{Sort}_{\mathit{repr}}$ is representable because bound variables \emph{may} range over terms of a given type. The \emph{contexts} of a type theory arise only indirectly in a SOGAT: they are given by telescopes of representable sorts. %

Every SOGAT gives rise to a (2,1)-category of models, which in the case of Martin-L\"of type theory restricts to the usual notion of CwF~\citep{dybjer:1996} or natural model~\citep{awodey:2018:natural-models}. The \emph{bi-initial model} of a given SOGAT plays the r\^ole of syntax, and it is with respect to this thing that the normalisation results of \citet{sterling:2021:thesis}, \citet{bocquet:2026}, and others apply. These kinds of results should not be thought of as side-stepping traditional syntax; as \citet{uemura:2021:thesis} has shown, traditional syntax gives one representation of the bi-initial model of \emph{any} SOGAT, and therefore the abstract normalisation result applies directly to traditional syntax even though it (thankfully) does not refer to it. We suggest that it is an advantage to have a metatheory that is robust to \emph{any} change in representation of syntax.

\paragraph{The bi-initial model and its metatheory}

A model of a SOGAT consists of a category of contexts $\CX{\mathcal{M}}$ together with an algebra for the signature of sorts and operations and equations in the presheaf category $\Psh{\CX{\mathcal{M}}}$ such that each family of representable sorts is taken to a representable family~\citep{awodey:2018:natural-models,uemura:2021:thesis}. The role of presheaves here can be understood as follows in the case of our running example: the sort $\Tp:\SortIdent{Sort}$ should be interpreted by a family of sets $\prn[\big]{\Tp_{\mathcal{M}}\prn{\Gamma}}_{\prn{\Gamma\in\CX{\mathcal{M}}}}$ of types in context, equipped with a substitution action $\Tp_{\mathcal{M}}\prn{\Gamma}\to\Tp_{\mathcal{M}}\prn{\Delta}$ for any context morphism $\gamma\colon\Delta\to\Gamma$. Once suitable associativity and unit laws for the substitution action are imposed, we have exactly the notion of a presheaf.

In a model of Martin-L\"of type theory, the interpretation of a type constructor amounts to a natural transformation of presheaves. Written in the extensional type theoretic vernacular of $\Psh{\CX{\mathcal{M}}}$ viewed as a locally cartesian closed category, for example, we have an arrow
\[
  \TpPi_{\mathcal{M}}\colon \prn{\alpha:\Tp_{\mathcal{M}}}\times\prn{\Tm_{\mathcal{M}}\prn{\alpha}\to\Tp_{\mathcal{M}}}\to \Tp_{\mathcal{M}}
\]
and the normalisation result implies that when $\mathcal{M}\coloneq\InitML$ is the bi-initial model of Martin-L\"of type theory, the natural transformation $\TpPi_{\InitML}$ depicted above is a \textbf{\emph{monomorphism}}. This is the semantic version of the injectivity of type constructors.

\subsubsection{Bidirectional surface languages}

The ``official'' syntax for Martin-L\"of type theory contains many apparently redundant type annotations. For example, the application of a $\lambda$-abstraction of type $\TpPi\prn{\alpha,x\mathpunct{.}\beta\prn{x}}$ would involve multiple references to $\alpha$ and $\beta$:
\[
  \TmApp\prn{\alpha,x\mathpunct{.}\beta\prn{x},\TmLam\prn{\alpha,x\mathpunct{.}\beta\prn{x}, x\mathpunct{.}\Var{body}\prn{x}},\Var{arg}}
  \tag{fully annotated}
\]

When the injectivity of type constructors (\S~\ref{sec:intro:algebraic-core-language}) holds, an alternative syntax with fewer annotations can be justified as presenting the bi-initial model, allowing us to write
\[
  \TmApp\prn{\TmLam\prn{x\mathpunct{.}\Var{body}\prn{x}},\Var{arg}}
  \text{.}
  \tag{totally unannotated}
\]

Leveraging metatheory to remove all annotations seems splendid at first, but unfortunately it renders the typing problem undecidable and therefore prevents a reasonable implementation. Consequently, the design of surface languages is conditioned by the need for annotation schemes that are less verbose than the default but not so sparse as to impede the system's practical implementation. This search has led to the idea of \emph{bidirectional type checking}~\citep{coquand:1996,pierce-turner:2000}, which is based on the observation that if only neutral and normal forms are allowed, then the totally unannotated system \emph{does} have decidable type checking by means of a bidirectional information flow protocol~\citep{mcbride:ni}:
\begin{enumerate}
  \item Types flow \emph{out} of contexts \emph{through} neutrals: we may \textbf{synthesise} the type of a neutral $\Gamma\vdash E$.
  \item Types flow \emph{in} from the goal \emph{into} normals:
  we may \textbf{check} whether a normal form $\Gamma\vdash M:\alpha$ is well-typed.
\end{enumerate}

Of course, we probably do want to be able to write $\beta$-redexes in a suitable surface language. The \emph{bidirectional} way to achieve this is to force these redexes to come with an appropriate type annotation, \emph{e.g.}\ by extending the synthesisable terms by a constructor $\prn{\alpha\ni M}$ embedding an annotated normal. This leads to a middle ground in which some but not all annotations are present:
\[
  \TmApp\prn[\big]{
    \TpPi\prn{\alpha,x\mathpunct{.}\beta\prn{x}}
    \ni
    \TmLam\prn{x\mathpunct{.}\Var{body}\prn{x}},
    \Var{arg}
  }
  \tag{bidirectionally annotated}
\]

Following \citet{dagand:2013}, we prefer to think of bidirectionality as a property of a distinct surface language with separate sorts for checkable and synthesisable terms. Then, a \emph{bidirectional elaborator} translates the surface terms to core language terms, inserting all the omitted annotations.

\subsection{Key ideas and contributions}

\subsubsection{Bidirectional elaboration in a partiality monad}

In many prior works~\citep{dunfield-krishnaswami:2021,mcbride:ni,mihejevs-hedges:2025}, bidirectional type checking and elaboration have been viewed as an instance of \emph{well-moded logic programming}. This perspective is good for producing attractive displays of inference rules, but it leaves underspecified many important matters that we consider of the essence, including the proper scoping of hypothetical judgements as well as the suspension and resumption of temporarily unsolved elaboration problems.
For this reason, we pursue a \emph{functional} interpretation of bidirectional elaboration in a partiality monad, where input and output modes are mapped to actual inputs and outputs of monadic functions. This is similar to the proposal of \citet{atkey:2015} (see \S~\ref{sec:related-work} for discussion), but unlike \emph{op.\ cit.}\ we will not use the \emph{Maybe} monad.

\paragraph{A generic model of the core language}
The metalanguage for our monadic interpretation of bidirectional elaboration will be \emph{extensional type theory} in which we additionally assume a sequence of constants and equations corresponding to the structure of a model of the target type theory, including  constants
$\Tp:\SET$ and $\Tm\colon \Tp\to\SET$.
Such constants can be found in the extensional type theoretic vernacular of $\Psh{\CX{\InitML}}$, which is where we shall ultimately interpret our metalanguage.

\paragraph{Semantic domains for surface language sorts}

Relative to a monad $\PMC$ determining a notion of computation~\citep{moggi:1991} appropriate for elaboration, we define three semantic domains $\TypScript$, $\SynScript$, and $\ChkScript$ corresponding to surface types, synthesis-mode, and checking-mode surface  terms respectively:
\[
  \TypScript \coloneq \PMC\prn{\Tp}
  \qquad
  \SynScript\coloneq\PMC\prn[\big]{
    \prn{\alpha:\Tp}\times \Tm\prn{\alpha}
  }
  \qquad
  \ChkScript\coloneq
  \prn{\alpha:\Tp}\to \PMC\prn{\Tm\prn{\alpha}}
\]

In other words,
\begin{enumerate}
  \item a surface language type denotes a \emph{computation} of a core language type;
  \item a surface language synthesis-mode term denotes a \emph{computation} of a core language term equipped with its type;
  \item a surface language checking-mode term denotes a dependent mapping from core language types $\alpha$ to \emph{computations} of core language terms of type $\alpha$.
\end{enumerate}

Then, the \emph{constructors} of the surface language will be implemented as \textbf{combinators}.

\paragraph{A partiality monad for elaboration}

What monad $\PMC$ shall we use? Experience both on the theoretical and the practical side suggests that a \emph{partiality monad} is appropriate, and we are keen to clarify that the role of partiality is to introduce \emph{conditions} for successful elaborations rather than to introduce any form of recursion. In particular, we define
\[
  \PMC\prn{X} \coloneq \prn{\varphi:\PROP}\times \prn{\varphi\to X}
\]
where $\PROP$ is a universe of proof-irrelevant propositions.\footnote{In the actual development, we shall replace $\PROP$ with a more restrictive subuniverse in order to automatically obtain the effective computability of elaboration.} An element of $\PMC\prn{X}$ is a \emph{partial} element of $X$ whose support is given by the first component. With this in hand, we can begin to introduce some of the surface language constructors as combinators.

\begin{example}[Conversion]
  We may implement the \emph{conversion rule} or \emph{mode switch} as follows:

  \iblock{\small
    \mrow{\ElabConv\colon \SynScript\to\ChkScript}
    \mhang{
      \ElabConv~E~\alpha \coloneq \Do
    }{
      \commentrow{// synthesise the type of $E$.}
      \mrow{\prn{\Var{typ},\Var{tm}}\gets E}
      \commentrow{// wait for $\Var{typ}$ to be equal to $\alpha$.}
      \mrow{\_\gets \Await~{\Var{typ}=_{\Tp}\alpha}}
      \commentrow{// by equality reflection in the metalanguage we have $\Var{tm}:\Tm\prn{\alpha}$, so we may return the latter.}
      \mrow{\Return~\Var{tm}}
    }
  }

  In the above, `$\Await~{\Var{typ}=_{\Tp}\alpha}$' is the unique element of $\PMC\prn{\Var{typ}=_{\Tp}\alpha}$ whose support is the proposition $\Var{typ}=_{\Tp}\alpha$; in fact, it is the \emph{greatest} element of the partial order $\PMC\prn{\Var{typ}=_{\Tp}\alpha}$. The $\ElabConv$ combinator captures \emph{precisely} the functionality of the conversion rule in bidirectional elaboration: to check a synthesis-mode term $E$ at a specific type, you synthesise the type of the former and proceed further only if it is equal to the type of the latter. In the scope of a proof that the two types are equal, you are free to return the core term denoted by $E$.
\end{example}

Below we show one further example to illustrate the way that our approach leverages the metatheory of the core language.

\begin{example}[Function application]
  Now assume that the constant $\TpPi\colon \prn{\alpha:\Tp}\times\prn{\Tm\prn{\alpha}\to\Tp}\to \Tp$ is \emph{injective}. Then we may implement the elaboration combinator for function application:

  \iblock{\small
    \mrow{\ElabApp : \SynScript\to\ChkScript\to\SynScript}
    \mhang{
      \ElabApp~E~M\coloneq\Do
    }{
      \commentrow{// we synthesise $E$, hoping to get a function type.}
      \mrow{
        \prn{\Var{typ},\Var{fun}}\gets E
      }
      \commentrow{// because $\TpPi$ is injective, the preimage $\TpPi^{-1}\prn{\Var{typ}}$ is a proposition!}
      \mrow{
        \prn{\Var{dom},\Var{fam}}\gets \Await~\TpPi^{-1}\prn{\Var{typ}}
      }
      \commentrow{// in this scope, we have $\Var{typ}= \TpPi\prn{\Var{dom},\Var{fam}}$; we check $M$ against $\Var{dom}$.}
      \mrow{
        \Var{arg}\gets M\prn{\Var{dom}}
      }
      \commentrow{// in this scope, we have $\Var{arg}:\Tm\prn{\Var{dom}}$; we apply the function and return the result with its type.}
      \mrow{
        \Return~\prn{
          \Var{fam}\prn{\Var{arg}},
          \TmApp\prn{\Var{dom},\Var{fam},\Var{fun}, \Var{arg}}
        }
      }
    }
  }
\end{example}

\subsubsection{Category-theoretic interpretation of scope}

Hypotheses are classified by an auxiliary sort $\HypScript\coloneq \prn{\alpha:\Tp}\times\Tm\prn{\alpha}$. In other words, a hypothesis denotes a typed core term and not a computation thereof, in keeping with the call-by-value interpretation of elaboration scripts. Binders in the surface language are represented using the metalanguage's function space; for example, the elaboration combinator for the surface language dependent function type is specified as follows:

\iblock{
  \mrow{\ElabPi : \TypScript\to\prn{\HypScript\to\TypScript}\to\TypScript}
}

To implement $\ElabPi~A~B$ we would ideally use Kleisli extension to read the values of $A$ and $B$ and apply $\TpPi$ to the results. But that doesn't quite make sense: we have $B\colon\prn{\alpha:\Tp}\times\Tm\prn{\alpha}\to\PMC\prn{\Tp}$, which is not the kind of thing that we can bind. Here we come to one of our main contributions: an interpretation of scope as \textbf{\emph{higher-arity lax monoidal structure}}. In particular, for any $K:\SET$ and family of types $X\colon K\to \SET$ we can define a ``scoping'' combinator
that takes $f\colon \prn{k:K}\to\PMC\prn{X\prn{k}}$ to the evident partial element that is defined when \emph{all} the $f\prn{k}$ are defined:
\[
  \boldsymbol{\vartheta}_{K,X}\colon
  \prn[\big]{\prn{k:K}\to \PMC\prn{X\prn{k}}}
  \to
  \PMC\prn[\big]{
    \prn{k:K}\to X\prn{k}
  }
\]

The above turns out to generalise the lax cartesian monoidal structure~\citep{kock:1972} of the monad $\PMC$, as when we specialise $K\coloneq\mathbf{2}$ we obtain (up to isomorphism) precisely the double strength
$
  \PMC\prn{X\prn{0}}\times\PMC\prn{X\prn{1}}\to
  \PMC\prn{X\prn{0}\times X\prn{1}}
$.
Writing ``$\Scope~k:K~\In~f\prn{k}$'' for $\boldsymbol{\vartheta}_{K,X}\prn{f}$, we can now implement the elaboration combinator for the surface language dependent function type constructor:

\iblock{\small
  \mhang{
    \ElabPi~A~B\coloneq\Do
  }{
    \commentrow{// first we bind the domain}
    \mrow{\Var{dom}\gets A}
    \commentrow{// now we have $\Var{dom}:\Tp$; we then use the scoping combinator to bind the entire family.}
    \mrow{\Var{fam}\gets \Scope~x:\Tm\prn{\Var{dom}}~\In~B\prn{\Var{dom},x}}
    \commentrow{// now we have $\Var{fam}\colon \Tm\prn{\Var{dom}}\to\Tp$.}
    \mrow{\Return~\TpPi\prn{\Var{dom},\Var{fam}}}
  }
}

\subsubsection{Equational reasoning about elaboration}

Instead of treating the surface language as a mere notation for the core language, we are giving a direct and compositional denotational semantics into a monad that governs the (fallible) computation of core language terms. The use of denotational semantics means that surface language terms themselves have a useful \emph{equational theory}.
The simplest use of equational reasoning is to unfold surface language terms (construed as elaboration \emph{scripts}, \emph{i.e.} instances of elaboration combinators) to the core language terms they compute by expanding definitions and using the monad laws. These calculations, as we show in our case study (\S~\ref{sec:case-study}) are elementary and easy to carry out.

\paragraph{Multilinearity of elaboration}

A more sophisticated use of equational reasoning is to prove theorems about individual elaboration combinators. For example, each elaboration combinator is built from $\TypScript$, $\SynScript$, and $\ChkScript$; as these are all \emph{algebras} for the partiality monad, we can ask whether the elaboration combinators are \emph{multilinear} in these arguments; the answer, we have shown in \S~\ref{sec:multilinearity-of-elaboration}, is \emph{yes} (for the binder-forming combinators, so long as we grant a mild \emph{strengthening} property that our intended semantics satisfies).
Multilinearity in the context of elaboration combinators intuitively means that the well-formedness conditions of a complete surface language term come from those of its subterms; it implies multistrictness, \emph{i.e.}\ the preservation of undefined elements in all arguments. Multilinearity is one possible sanity condition that might be used to distinguish reasonable elaboration combinators from unreasonable ones (like $\mathtt{first} : \ChkScript\to\ChkScript\to\ChkScript$, which returns
its first argument and discards its second; this is not multilinear and would not be likely to appear in a competently designed surface language).

\subsubsection{Inequational reasoning and subject reduction}

Each monadic type $\PMC\prn{X}$ carries a standard \emph{partial ordering}: we have $\prn{\varphi,u}\preccurlyeq_{\PMC} \prn{\psi,v}$ when $\varphi\to \psi$ holds and, assuming $\varphi=\psi=\top$ we have $u=v$. This is an ``information ordering'' on partial elements, just like in domain theory. %
Taking the partial ordering seriously, we get new ways to reason about surface code by means of its denotation in the partiality monad. For example, \citet{mcbride:ni} has suggested a rewriting system for surface terms including ``subject reductions'' for surface $\beta$-redexes like the following:
\[
  \TmApp\prn[\big]{
    \TpPi\prn{\alpha,x\mathpunct{.}\beta\prn{x}}
    \ni
    \TmLam\prn{x\mathpunct{.}\Var{body}\prn{x}},
    \Var{arg}
  }
  \longmapsto
  \Var{body}\prn{
    \alpha\ni\Var{arg}
  }
  \tag{\citet{mcbride:ni}}
\]

Mc~Bride's reduction is reflected in our denotational model of bidirectional surface syntax as a (strict) \emph{inequation} in the poset $\SynScript \equiv \PMC\prn{\prn{\alpha:\Tp}\times\Tm\prn{\alpha}}$, as we discuss in \S~\ref{sec:subject-reduction}:
\[
  \TmApp\prn[\big]{
    \TpPi\prn{\alpha,x\mathpunct{.}\beta\prn{x}}
    \ni
    \TmLam\prn{x\mathpunct{.}\Var{body}\prn{x}},
    \Var{arg}
  }
  \preccurlyeq_{
    \SynScript
  }
  \prn[\big]{
    \beta\prn{\alpha\ni\Var{arg}}
    \ni
    \Var{body}\prn{
      \alpha\ni\Var{arg}
    }
  }
  \tag{Theorem~\ref{thm:subject-reduction}}
\]

The reason for the directionality of the inequation is that the left-hand side places constraints on $\beta$ and $\Var{body}$ that are not placed on the right: on the right, these only need be well-typed when instantiated with $\prn{\alpha\ni\Var{arg}}$. We suggest that the subject reduction laws for surface calculi should be viewed \emph{not} as reflecting a rewriting semantics for the core theory but instead as a guide for the safe optimisation and refactoring of surface expressions: an inequality $u\preccurlyeq v$ means that the proof becomes more complete when $u$ is replaced by $v$, without changing the core term it produces.

\subsubsection{Presheaf semantics for stable elaboration with suspension}

One of the major desiderata for \emph{stable elaboration} elucidated by \citet{kovacs2ltt-principles} is that ``if we reorder elaboration actions and metavariable solutions, the \emph{eventual} output must not change, up to conversion''. From a formal point of view, our partiality monad satisfies this criterion because it is a \emph{commutative} monad: elaboration actions that place constraints on metavariables can be reordered at will without changing the results.

This is at first slightly confusing. For example, if an elaboration action requires that $?\textsc{meta}=_{\Tp}\mathbb{N}$ in order to succeed, wouldn't this mean that it would fail if it was executed before solving $?\textsc{meta} \coloneq \mathbb{N}$, but succeed if executed afterwards? The key idea of our approach is, to the contrary, that the denotation of an elaboration action contains the \emph{condition} under which it may go ahead, but this condition does not need to be boolean. Therefore, under either ordering, the result is a partial element with support $?\textsc{meta}=_{\Tp}\mathbb{N}$, which may become true in various scopes. So rather than asserting a condition (and failing otherwise), we can think of an action like `$\Await~?\textsc{meta}=_{\Tp}\mathbb{N}$' as simply \textbf{suspending} the elaboration problem until such time as the metavariable is solved.

These intuitions are ultimately vindicated by the interpretation of our metalanguage into \emph{presheaves} on the category of contexts of the object language. In the presheaf model, a proposition denotes a \emph{sieve}~\citep{maclane-moerdijk:1992}: rather than a boolean truth value, it is a \emph{predicate} that isolates the substitutions under which it becomes true. Then the force of `$\Await~\varphi$' is to narrow this sieve to the contexts where $\varphi$ holds; a partial element therefore \emph{awaits} a substitution that will bring it into its support.

Although there is much more to investigate and we admittedly have little to say in this paper about metavariables, we believe that our presheaf semantics provides a workable starting point for understanding the suspension and resumption of elaboration.

\subsection{Discussion of related work}\label{sec:related-work}

We take inspiration liberally from prior works, especially those of \citet{atkey:2015} and \citet{uemura:2021:thesis}, as well as many conversations with Conor Titania Mc~Bride. We wish to acknowledge in particular the following prior works:

\begin{enumerate}
  \item From Atkey's \emph{Algebraic Approach to Type Checking and Elaboration}~\citep{atkey:2015} we take the idea that surface syntax (for simple type systems) might denote monadic combinators in a functional programming language. At the end of his talk, Atkey asked ``Does this work for more complex type systems?'' Our results suggest that the answer is `yes', assuming one upgrades the setup appropriately. In addition to our successful treatment of full dependent types, our work differs from Atkey in a few ways: (1) context is implicitly governed by our novel $\Scope$ combinator, (2) our account of elaboration is correct by construction, (3) we account for the suspension of elaboration problems by using a partiality monad on presheaves instead of the maybe monad.
  \item From Uemura's general initiality proof for SOGATs~\citep{uemura:2021:thesis}, we take the idea of interpreting surface syntax using a partiality monad on presheaves. Whereas Uemura treats the interpretation of fully explicit raw syntax, we have shown how to interpret \emph{bidirectional} raw syntax by leveraging the metatheory of the core language; we further carefully arrange for our elaboration function to be effectively computable, employing ideas from synthetic topology.
  \item A proof-relevant version of the partial map classifier is employed by \citet{oconnor:2019} to support the management of subgoals in \textsc{Trip}, an Agda DSL for Hoare-logic verification of imperative programs. We note that the `monad' of \emph{op.\ cit.} does not satisfy the monad laws, but it would be a polynomial pseudomonad if formulated in the appropriate environment (see \citet{awodey-newstead:2018,aberle-spivak:2025} for related discussion).
  \item The asynchronicity suggested by our $\Await$-notation is inspired by Conor Titania Mc~Bride's (paraphrased) assertion that a reliable elaborator has more in common with an \emph{operating system} than a straightline type checker---a point of view that is explored in terms of interactive processes in TypOS~\citep{allais-et-al:2022:typos}. Our use of presheaves is also consonant with Mc~Bride's suggestion that a correct solution to the ``Good Bye, Lenin!'' problem (by which they refer to resuming suspended elaboration problems under an updated information environment) must involve a \emph{Kripke-style} interpretation of the function space.

\end{enumerate}

We are not the only ones to consider elaborating ``directly'' to semantics; \emph{SynthLean} (\citet{synthlean:cpp}) is a certifying elaborator in Lean that targets an arbitrary natural model of Martin-L\"of type theory. The SynthLean elaborator is, as the authors say, ``certifying rather than certified''---which means that, unlike ours, it is not guaranteed to terminate or to be defined on all well-typed expressions. We believe that our methods would be compatible with the goals of SynthLean/HoTTLean.

\paragraph{Semantic \emph{vs.}\ syntactic frameworks}
To discuss the remaining related work, we find it helpful to think of elaboration as a more sophisticated form of \emph{parsing}. Proposals for general accounts of parsing tend to fall into two categories:

\begin{enumerate}
  \item \emph{Semantic frameworks.} Parser combinator libraries like Parsec~\citep{leijen-meijer:2001} provide an API for compositionally producing and reasoning about parsers, using ideas from functional programming.
  \item \emph{Syntactic frameworks.} Parser generators like Yacc~\citep{johnson:1975:yacc}, Menhir~\citep{pottier:2006}, \emph{etc.} transform suitably constrained grammars into working parsers, using ideas from formal language theory.
\end{enumerate}

Similarly, we have proposed an API for compositionally producing and reasoning about elaborators without introducing any kind of formal grammar (\emph{i.e.}\ presentation by inference rules).  On the syntactic side, both \citet{mcbride:ni} and \citet{felicissimo:2025} have proposed syntactic frameworks that might be characterised as ``typechecker generators'': you give the framework a set of rules written in a suitably constrained format, and it gives you a typechecker.

We are very interested in \emph{both} sides of this space, and we think that there is potential for combining them. For example, one weakness of the cited syntactic approaches has been the failure to account for $\eta$-rules, which are table stakes for dependently typed systems in the present day. Our semantic approach effortlessly handles $\eta$-rules, and we wonder if it might one day serve as an alternative ``backend'' to the cited syntactic frameworks.

\section{Elementary structure of Martin-L\"of type theory, internally}\label{sec:elementary-model}

The goal of this paper is to describe a semantic version of the elaboration function which takes values in \emph{actual constituents} of the initial model of Martin-L\"of type theory---which are, by definition, intrinsically typed and identified up to judgemental equality. Such a treatment would be sound by construction without requiring any further proof. In this section, we describe the axiomatic form in which we shall treat the target model of the type theory.

Models of Martin-L\"of type theory are typically described in terms of Dybjer's \emph{categories with families}~\citep{dybjer:1996}, or equivalent structures like Awodey's \emph{natural models}~\citep{awodey:2018:natural-models}. In either case, one starts with a category of contexts and then the elementary structure (judgements, type connectives, \emph{etc.}) of the type theory is expressed in terms of presheaves and natural transformations over the category of contexts. In recent years, several authors~\cite{gratzer-sterling:2020,awodey:2025,bocquet:2026} have noted that the elementary structure can be studied independently of its instantiation over a category of contexts as, indeed, it can be expressed in any locally cartesian closed category. This point of view, in fact, can be traced back to Martin-L\"of's \emph{logical framework} presentation~\citep{nordstrom-petersson-smith:2001,martin-lof:1986} from the late 1980s.

At this stage, we do not yet have need of contexts and will therefore stick to the elementary axiomatisation of type theory's judgements and type connectives in the logical framework style. What follows may be read in the internal language of \emph{any} locally cartesian closed category (\emph{e.g.}\ presheaves on the category of contexts of the initial model of Martin-L\"of  type theory). %
We first describe the bare judgemental structure of Martin-L\"of type theory: types and terms.

\begin{definition}
  A \emph{judgemental structure} is defined to be a type $\Tp:\SET$ together with a family of types $\Tm\colon \Tp\to\SET$. We will at times write $\widetilde{\Tm}$ for the dependent pair type $\prn{\alpha:\Tp}\times\Tm\prn{\alpha}$.
\end{definition}

Of course, a judgemental structure is nothing more than a \emph{container} in the sense of \citet{abbott-altenkirch-ghani:2005}, which determines a polynomial \emph{extension} functor $\Fam\colon\SET\to\SET$:
\[
  \Fam\prn{X} \coloneq
  \prn{\alpha:\Tp}\times
  \prn{
    \Tm\prn{\alpha}\to X
  }
\]

Next, we describe individual connectives over a judgemental structure.

\begin{definition}[Dependent function types]
  A \emph{dependent function type structure} is given by constants

  \iblock{
    \mrow{
      \TpPi \colon
      \Fam\prn{\Tp}
      \to \Tp
    }
    \mrow{
      \TmLam \colon
        \prn[\big]{
          \alpha:\Tp,
          \beta\colon\Tm\prn{\alpha}\to\Tp,
          b\colon \prn{x:\Tm\prn{\alpha}} \to \Tm\prn{\beta\prn{x}}
        }\to
        \Tm\prn{
          \TpPi\prn{\alpha,\beta}
        }
    }
    \mrow{
      \TmApp \colon
      \prn[\big]{
        \alpha:\Tp,
        \beta\colon \Tm\prn{\alpha}\to\Tp,
        f :\Tm\prn{\TpPi\prn{\alpha,\beta}},
        x : \Tm\prn{\alpha}
      }
      \to
      \Tm\prn{\beta\prn{x}}
    }
  }

  \noindent
  satisfying the usual $\beta$- and $\eta$-laws:
  \begin{align*}
    \prn[\big]{
      \ldots,
      b\colon \prn{x:\Tm\prn{\alpha}} \to \Tm\prn{\beta\prn{x}},
      x:\Tm\prn{\alpha}
    }&\to
    \TmApp\prn{\alpha,\beta,\TmLam\prn{\alpha,\beta,b},x} =
    b\prn{x}
    \\
    \prn{
      \ldots,
      f : \Tm\prn{\TpPi\prn{\alpha,\beta}}
    }&\to
    f = \TmLam\prn{\alpha,\beta, \lambda x\mathpunct{.} \TmApp\prn{\alpha,\beta,f,x}}
  \end{align*}
\end{definition}

\begin{definition}[Dependent pair types]
  A \emph{dependent pair type structure} is given by constants

  \iblock{
    \mrow{
      \TpSg \colon
      \Fam\prn{\Tp}
      \to \Tp
    }
    \mrow{
      \TmPair \colon
        \prn[\big]{
          \alpha:\Tp,
          \beta\colon\Tm\prn{\alpha}\to\Tp,
          a : \Tm\prn{\alpha},
          b:\Tm\prn{\beta\prn{a}}
        }\to
        \Tm\prn{
          \TpSg\prn{\alpha,\beta}
        }
    }
    \mrow{
      \TmFst \colon
      \prn[\big]{
        \alpha:\Tp,
        \beta\colon \Tm\prn{\alpha}\to\Tp,
        u : \Tm\prn{\TpSg\prn{\alpha,\beta}}
      }
      \to
      \Tm\prn{\alpha}
    }
    \mrow{
      \TmSnd \colon
      \prn[\big]{
        \alpha:\Tp,
        \beta\colon \Tm\prn{\alpha}\to\Tp,
        u : \Tm\prn{\TpSg\prn{\alpha,\beta}}
      }
      \to
      \Tm\prn{\beta\prn{\TmFst\prn{\alpha,\beta,u}}}
    }
  }

  \noindent
  satisfying $\beta$- and $\eta$-laws:
  \begin{align*}
    \prn{\ldots,a : \Tm\prn{\alpha},b:\Tm\prn{\beta\prn{a}}}
    &\to
    \TmFst\prn{\alpha,\beta,\TmPair\prn{\alpha,\beta,a,b}} = a
    \\
    \prn{\ldots,a : \Tm\prn{\alpha},b:\Tm\prn{\beta\prn{a}}}
    &\to
    \TmSnd\prn{\alpha,\beta,\TmPair\prn{\alpha,\beta,a,b}} = b
    \\
    \prn{\ldots,u:\Tm\prn{\TpSg\prn{\alpha,\beta}}}
    &\to
    u = \TmPair\prn{
      \alpha,\beta,
      \TmFst\prn{\alpha,\beta,u},
      \TmSnd\prn{\alpha,\beta,u}
    }
  \end{align*}
\end{definition}

\begin{definition}[Intensional identity types]
  An \emph{intensional identity type structure} is given by constants

  \iblock{
    \mrow{
      \TpId \colon \prn{\alpha:\Tp,u:\Tm\prn{\alpha},v:\Tm\prn{\alpha}}
      \to \Tp
    }
    \mrow{
      \TmRefl \colon \prn{\alpha:\Tp,u:\Tm\prn{\alpha}}\to
      \Tm\prn{\TpId\prn{\alpha,u,u}}
    }
    \mrow{
      \TmJ\colon
        \begin{pmatrix*}[l]
          \alpha:\Tp,
          u:\Tm\prn{\alpha},v:\Tm\prn{\alpha},
          p:\Tm\prn{\TpId\prn{\alpha,u,v}},\\
          \beta:\prn{
            x:\Tm\prn{\alpha},
            y:\Tm\prn{\alpha},
            z:\Tm\prn{\TpId\prn{\alpha,x,y}}
          }\to\Tp,\\
          b : \prn{
            x:\Tm\prn{\alpha}
          }\to
          \Tm\prn{\beta\prn{x,x,\TmRefl\prn{\alpha,x}}}
        \end{pmatrix*}
        \to \Tm\prn{\beta\prn{u,v,p}}
    }
  }

  \noindent
  satisfying the usual computation rule
  \[
    \prn{
      \ldots,
      b : \prn{x:\Tm\prn{\alpha}}\to \Tm\prn{\beta\prn{x,x,\TmRefl\prn{\alpha,x}}}
    }\to
    \TmJ\prn{\alpha,u,u,\TmRefl\prn{\alpha,u},\beta,b}
    = b\prn{u}\text{.}
  \]
\end{definition}

We also include an ``answer type'' in order to have a base case.\footnote{When formulating a type theory about which one will prove metatheoretic results, it is important to have a base type: otherwise, the initial model of the type theory will actually contain no types! The \emph{unit type} could be used, but if the unit type is the only base type, it trivialises the normalisation theorem. For this reason, we include a base type with two constructors; we could include an eliminator, but this is not needed to illustrate our methods.}

\begin{definition}[Answer type structure]
  An \emph{answer type structure} is given by constants $\TpAns:\Tp$ and $\TmYes,\TmNo:\Tm\prn{\TpAns}$.
\end{definition}

\section{A partiality monad for synthetic elaboration}\label{sec:partiality-monad}

We intend to interpret correct-by-construction elaboration algorithms in the partiality monad of a topos equipped with a dominance~\citep{rosolini:1986}. We recall some basic definitions.

\subsection{Dominance and partiality (relative) monad; the \texorpdfstring{$\Do$}{do}-notation}

A partial element of a type $X$ consists of a \emph{support} $\varphi:\PROP$ under which it is defined, followed by a function $\varphi\to X$. Because we ultimately intend our elaborators to be executable, we don't want to allow all propositions to be used as supports; on the other hand, it would be too strong to require the support to be decidable (a Boolean), as we would not be able to interpret something like the conversion rule that is defined just when two given types are equal.\footnote{Indeed, although it is externally the case that type equality is decidable, it will not be internally decidable. Later on, we shall choose a notion of proposition that refers to external decidability.} At this point, we will just axiomatise whatever closure conditions we need relative to a \emph{class} of propositions, which we later on will instantiate in an appropriate way. For such a class to give rise to a partiality monad, we need it to satisfy Rosolini's \emph{dominance} property~\cite{rosolini:1986}.

\begin{definition}[{\citet{rosolini:1986}}]\label{def:dominance}
  A \emph{dominance} is a univalent universe $\OProp\subseteq\PROP$ of propositions containing $\top$, such that for $\varphi\in \OProp$ and $\psi\in \OProp^{\varphi}$, the dependent pair type $\prn{p:\varphi}\times \psi\prn{p}$ lies in $\OProp$.
  We shall call the dominance \emph{strict} if we additionally have $\bot\in\OProp$.
\end{definition}

Now let $\OProp$ be an arbitrary fixed dominance.

\begin{definition}[Open and closed propositions and subsets]
  Elements of $\OProp$ are called \emph{open} propositions. A subset $U\subseteq X$ is called open when its characteristic map $\chi\colon X\to\PROP$ factors (necessarily uniquely) through $\OProp\subseteq \PROP$. %
  A proposition $\varphi$ is called \emph{closed} when its negation is open; similarly, a subset is closed when its complement is open.
\end{definition}

\begin{definition}[Discrete and Hausdorff types]
  A type $X$ is called \emph{discrete} when its equality predicate $X\subseteq X\times X$ is open; dually, $X$ is called \emph{Hausdorff} when its equality predicate is closed.
\end{definition}

\begin{definition}
  An \emph{open partial element} of a type $X$ is an open proposition $\varphi\in\OProp$ together with a family $x\colon \varphi\to X$ of elements of $X$. We shall write $\PMC\prn{X}$ for the \emph{open partial element classifier} of $X$:
  \[
  \PMC\prn{X} \coloneq \prn{\varphi:\OProp}\times X^{\varphi}
  \]
\end{definition}

When it is clear from context, we will simply say ``partial'' instead of ``open partial''.

\begin{notation}
  Given $u:\PMC\prn{X}$ we shall write $u\IsDef \coloneq \pi_1\prn{u}$ for the proposition that $u$ is defined; when we have $p:u\IsDef$, we shall write $u_p$ for the value $\pi_2\prn{u}\prn{p}:X$ of $u$.
\end{notation}

\citet{rosolini:1986} observes that the partial element classifier associated to a dominance forms a monad, but we will be interested in a more refined \emph{relative 2-monad} structure~\cite{altenkirch-chapman-uustalu:2015,fiore-et-al:2016} on $\PMC$ that incorporates the standard partial ordering of partial elements.

\begin{definition}[Standard partial ordering of partial elements]\label{def:pmc-poset}
  Given a type $X$,  we equip the partial element classifier $\PMC\prn{X}$ with its standard partial ordering
  $
  u \preccurlyeq_{\PMC\prn{X}} v\Longleftrightarrow
  u\IsDef\to \prn{v\IsDef \land u=v}
  $.
\end{definition}

Consequently, we may view the partial element classifiers as giving a functor $\PMC\colon \SET\to\POSET$.

\begin{definition}[Discretely-ordered inclusion]
  Every type $X$ may be regarded as a poset $\DiscIncl\prn{X}\coloneq\prn{X,{=_X}}$ ordered by equality. Any map $X\to A$ into a poset is automatically monotone as a function of $\DiscIncl\prn{X}$, so $\mathbf{Pos}\prn{\DiscIncl\prn{X},A}\cong A^X$. Moreover $\DiscIncl$ preserves products: $\DiscIncl\prn[\big]{\prod_{i:I}X_i}\cong\prod_{i:I}\DiscIncl\prn{X_i}$.
\end{definition}

\begin{construction}[Relative 2-monad structure on partial element classifiers]\label{con:L-relative-monad}
  The partial element classifier functor $\PMC\colon\SET\to\POSET$ is in fact a \emph{relative 2-monad}~\cite{altenkirch-chapman-uustalu:2015,fiore-et-al:2016} on the discretely-ordered inclusion functor $\DiscIncl\colon\SET\to\POSET$. The unit is defined as follows:

  \iblock{
    \mrow{
      \eta_X\colon \DiscIncl\prn{X}\to \PMC\prn{X}
    }
    \mrow{
      \eta_X\prn{x}\IsDef \coloneq \top
    }
    \mrow{
      \eta_X\prn{x}_{p} \coloneq x
    }
  }

  Given $f\colon \DiscIncl\prn{X}\to \PMC\prn{Y}$, the Kleisli extension $f^\dagger \colon \PMC\prn{X}\to\PMC\prn{Y}$ is defined as follows:

  \iblock{
    \mrow{
      f^\dagger\prn{u}\IsDef \coloneq \prn{p:u\IsDef}\times f\prn{u_p}\IsDef
    }
    \mrow{
      f^\dagger\prn{u}_{\prn{p,q}} \coloneq f\prn{u_p}_q
    }
  }
\end{construction}

Beyond the usual relative monad laws, the two-dimensional structure asks that the extension operator be \emph{locally monotone}: if $f\preccurlyeq g$ pointwise then $f^\dagger\preccurlyeq g^\dagger$, which is routine.

We will at times follow the notation of functional programming languages like Haskell:

\begin{notation}[The $\Do$-notation]\label{notation:do-notation}
  We will sometimes write ``$\Return~x$'' for $\eta_X\prn{x} : \PMC\prn{X}$.
  When it is convenient, we shall write ``$x\gets u; f\prn{x}$'' for an applied Kleisli extension $f^\dagger\prn{u}$. A block of Kleisli bindings is prefixed by the keyword ``$\Do$''.
\end{notation}

\subsection{Open-subsingleton subsets and partial elements; the \texorpdfstring{$\Await$}{await}-notation}

\begin{definition}
  A subset $U\subseteq X$ is said to be \emph{open-subsingleton} when the existential quantification $\exists x:X\mathpunct{.}x\in U$ is open and for any $x,y\in U$ we have $x=y$.
\end{definition}

\begin{construction}[The generic partial element]
  Every open-subsingleton subset $U\subseteq X$ gives rise to a \emph{generic} partial element $\boldsymbol{\gamma}_U:\PMC\prn{U}$ defined like so:
  \begin{align*}
    \boldsymbol{\gamma}_U\IsDef &\coloneq \exists x:X\mathpunct{.} x\in U
    \\
    \prn{\boldsymbol\gamma_U}_p &\coloneq \textit{(the unique element of $U$ under these assumptions)}
  \end{align*}
\end{construction}

\begin{notation}[The $\Await$-notation]
  For any open-subsingleton subset $U\subseteq X$, we shall write ``$\Await~U$'' in the monadic metalanguage for the generic partial element $\boldsymbol{\gamma}_U:\PMC\prn{U}$. For an open proposition $\varphi\in \OProp$ we shall write ``$\Await~\varphi$'' instead of the more verbose ``$\Await~\brc{x:\mathbf{1}\mid \varphi}$''.
\end{notation}

\begin{example}[Equational assertions in discrete types]\label{ex:equational-assertion}
  When $X$ is a discrete type, then for any $u,v:X$ the equality proposition $\prn{u=_Xv}$ is open. Therefore, we may write
  ``$
  \Await~{u=_Xv}
  $''
  for the trivial computation that is defined if and only if $u$ and $v$ are equal.
\end{example}

\begin{example}[Non-equality assertions in Hausdorff types]
  Dually to Example~\ref{ex:equational-assertion}, when $X$ is Hausdorff we may write ``$\Await~\prn{u\not=_X v}$'' for the trivial computation that is defined if and only if $u$ and $v$ are \emph{unequal}.
\end{example}

\subsection{Compact types and universal quantification; the $\Scope$-notation}

The partial element classifier (relative) 2-monad is commutative and has a unique strength, corresponding uniquely~\cite{kock:1972} to a \emph{lax (cartesian) monoidal structure} consisting of an appropriately coherent datum
\[
  \boldsymbol{\vartheta}_{\Card{n},X}
  \colon \prn[\big]{\textstyle\prod_{k<\Card{n}} \PMC\prn{X_k}}
  \to
  \PMC\prn[\big]{
    \textstyle\prod_{k<\Card{n}} X_k
  }
  \tag{lax monoidal structure}
\]
by which any finite number of partial elements may be combined to form a partial tuple of elements.

One of our main insights is that this lax monoidal structure exists not only at finite arities $\Card{n}$ but, in fact, at any arity $K$ that is \emph{compact} in the sense of synthetic topology~\citep[\S~7]{taylor:2000}:

\begin{definition}[Compact types]
  A type $X$ is called \emph{compact} when for every open predicate $\varphi\colon X\to \OProp$, the universal quantification $\forall x:X\mathpunct{.}\varphi\prn{x}$ is open.
\end{definition}

We now come to our generalisation of the lax monoidal structure to compact arities.

\begin{construction}[Lax monoidal structure at compact arity]\label{con:lax-monoidal-compact}
  Let $X\colon K\to \SET$ be a family of types indexed in some compact type $K$ for the dominance $\OProp$. Then we may define the \emph{lax monoidal structure of $\PMC$ at arity $K$} to be the following family of maps:

  \iblock{
    \mrow{
      \boldsymbol{\vartheta}_{K,X} \colon \prn[\big]{
        \prn{k:K}\to \PMC\prn{X\prn{k}}
      }
      \to
      \PMC\prn[\big]{
        \prn{k:K}\to X\prn{k}
      }
    }
    \mrow{
      \boldsymbol{\vartheta}_{K,X}\prn{f}\IsDef \coloneq \forall k:K\mathpunct{.} f\prn{k}\IsDef
    }
    \mrow{
      \boldsymbol{\vartheta}_{K,X}\prn{f}_p\coloneq \lambda k\mathpunct{.}f\prn{k}_{p\prn{k}}
    }
  }

  The above is natural in $X$.
\end{construction}

Clearly every finite type is compact, as the universal quantifier may be reduced to iterated conjunction, which we have from $\OProp$ being a dominance. In this case, the lax monoidal structure obtained from Construction~\ref{con:lax-monoidal-compact} coincides with the standard one.
We shall, however, eventually use Construction~\ref{con:lax-monoidal-compact} to interpret \emph{hypothetical judgements} and \emph{binders} in the denotational semantics of elaboration by assuming certain non-finite objects arising from syntax to be compact.

\begin{notation}[The $\Scope$-notation]
  Let $f\colon \prn{k:K}\to \PMC\prn{X\prn{k}}$ be a dependent function over a compact type $K$. We shall write ``$\Scope~k:K~\In~f\prn{k}$'' for $\boldsymbol{\vartheta}_{K,X}\prn{f} : \PMC\prn[\big]{\prn{k:K}\to X\prn{k}}$.
\end{notation}

\begin{lemma}[Compactness is closed under compact sums]\label{lem:compact-sum}
  If $K$ is compact and $J_k$ is compact for each $k:K$, then $\widetilde{J}\coloneq\prn{k:K}\times J_k$ is compact.
\end{lemma}
\ifextended
\begin{proof}
  For an open predicate $\varphi$ on $\widetilde{J}$ we have $\forall\prn{k,j}:\widetilde{J}\mathpunct{.}\varphi\prn{k,j}\Longleftrightarrow\forall k:K\mathpunct{.}\forall j:J_k\mathpunct{.}\varphi\prn{k,j}$; the inner quantifier is open since $J_k$ is compact, the outer since $K$ is compact.
\end{proof}
\fi

In what follows, let $K$ be compact and let $J\colon K\to\SET$ be a family of compact types; we will write $\widetilde{J}\coloneq \prn{k:K}\times J_k$, which is also compact by Lemma~\ref{lem:compact-sum}. We check several important coherence laws for the lax monoidal structure below:

\begin{lemma}[Compatibility with extension]\label{lem:scope-extension}
  For a family of maps $u_k\colon X_k\to\PMC\prn{Z_k}$ with $k:K$, we have
  $
  \boldsymbol{\vartheta}_{K,Z}\circ\textstyle\prod_k u_k^{\dagger}
  =\prn[\big]{\boldsymbol{\vartheta}_{K,Z}\circ\textstyle\prod_k u_k}^{\dagger}\circ\boldsymbol{\vartheta}_{K,X}
  \colon\textstyle\prod_k\PMC\prn{X_k}\to\PMC\prn[\big]{\textstyle\prod_k Z_k}
  $
  or equivalently
    $
    \Scope~k:K~\In~\prn{
      x_k \gets w_k;
      u_k\prn{x_k}
    }
    =
    x \gets \prn{\Scope~k:K~\In~w_k};
    \Scope~k:K~\In~u_k\prn{x_k}\text{.}
  $
\end{lemma}
\ifextended
\begin{proof}
  Write $w_k=\prn{\psi_k,v_k}:\PMC\prn{X_k}$. Unfolding, the left-hand side applied to $\prn{w_k}_k$ is defined exactly when $\forall k\mathpunct{.}\prn{p:\psi_k}\times u_k\prn{v_k\prn{p}}\IsDef$, while the right-hand side is defined exactly when $\prn{P:\forall k\mathpunct{.}\psi_k}\times\forall k\mathpunct{.}u_k\prn{v_k\prn{P\prn{k}}}\IsDef$. These open propositions are equal by the distributivity of products over sums; under this identification both sides take the value $\lambda k\mathpunct{.}u_k\prn{v_k\prn{p_k}}_{q_k}$.
\end{proof}
\fi

Lemma~\ref{lem:scope-extension} says that $\boldsymbol{\vartheta}$ is (at finite arity) the symmetric lax monoidal structure of the commutative monad $\PMC$. The lax-idempotency of $\PMC$---visible in the standard order on partial elements---is the structural reason for this commutativity~\citep{slattery:2023,slattery:2024:thesis}.

Commutativity says that the \emph{order} of two bindings is immaterial; a second, independent fact is that the \emph{repetition} of a binding is immaterial. Together with its (unique) strength, these are exactly the defining properties of an \emph{equational lifting monad}~\citep{bucalo-fuhrmann-simpson:2003}.

\begin{lemma}[Duplication]\label{lem:duplication}
  For $u:\PMC\prn{X}$ and $k\colon X\times X\to\PMC\prn{Z}$ we have
  $
  x\gets u;\ y\gets u;\ k\prn{x,y}
  =
  x\gets u;\ k\prn{x,x}\text{.}
  $
  Equivalently, $\boldsymbol{\vartheta}_{\mathbf{2}}\prn{u,u}=\PMC\prn{\Delta}\prn{u}$ for the diagonal $\Delta\colon X\to X\times X$.
\end{lemma}
\ifextended
\begin{proof}
  The left-hand side is defined exactly when $\prn{p:u\IsDef}\times\prn{q:u\IsDef}\times k\prn{u_p,u_q}\IsDef$. Since $u\IsDef$ is a proposition, $p=q$, so this is the support $\prn{p:u\IsDef}\times k\prn{u_p,u_p}\IsDef$ of the right-hand side, and both sides take the value $k\prn{u_p,u_p}_r$ there.
\end{proof}
\fi

\begin{lemma}[Unit and naturality]\label{lem:theta-unit}
  We have $\boldsymbol{\vartheta}_{K,X}\circ\prod_{k}\eta_{X_k}=\eta_{\prod_k X_k}$. Moreover $\boldsymbol{\vartheta}_{\prn{K,-}}$ is natural: $\boldsymbol{\vartheta}_{K,Y}\circ\prod_k\PMC\prn{h_k}=\PMC\prn[\big]{\prod_k h_k}\circ\boldsymbol{\vartheta}_{K,X}$ for $h_k\colon X_k\to Y_k$.
\end{lemma}

\begin{lemma}\label{lem:theta-comp}
  We have $\boldsymbol{\vartheta}_{\mathbf{1},X}=1_X$.
  Given $Y\colon\widetilde{J}\to\SET$, we have
  $
  \boldsymbol{\vartheta}_{\widetilde{J},Y}
  =\boldsymbol{\vartheta}_{K,\prn{\prod_j Y}}\circ\textstyle\prod_{k:K}\boldsymbol{\vartheta}_{J_k,Y_k}
  $, or equivalently for each $u$ we have
  $\Scope~k:K~\In~\prn{\Scope~j:J_k~\In~u\prn{k,j}}=\Scope~\prn{k,j}:\widetilde{J}~\In~u\prn{k,j}$.
\end{lemma}

\ifextended
\begin{proof}
  Both sides send $u$ to a partial element of $\prod_{\prn{k,j}}X_{k,j}$ whose support is $\forall k\mathpunct{.}\forall j\mathpunct{.}u_{k,j}\IsDef$, equivalently $\forall\prn{k,j}:\widetilde{J}\mathpunct{.}u_{k,j}\IsDef$ by Lemma~\ref{lem:compact-sum}, and whose value at a witness $p$ is $\lambda\prn{k,j}\mathpunct{.}\prn{u_{k,j}}_{p\prn{k,j}}$; these are computed identically on the two sides.
\end{proof}
\fi

We will at times have use for the following lemma.

\begin{lemma}[Interchange over an inhabited scope]\label{lem:inhabited-scope}
  Let $K$ be a compact and \emph{inhabited} type. For any type $X$ and family of types $Y\colon K \to \SET$, given a partial element $u:\PMC\prn{X}$ and a function $c\colon X\to\prn{k:K}\to\PMC\prn{Y\prn{k}}$, we have
  $
  \Scope~k:K~\In~\prn{x\gets u;\ c\prn{x}\prn{k}}
  =
  x\gets u;\ \Scope~k:K~\In~c\prn{x}\prn{k}
  $
\end{lemma}
\ifextended
\begin{proof}
  We compute the support of the left-hand side from left to right:

  \iblock{
    \mhang{\prn[\big]{\Scope~k:K~\In~\prn{x\gets u;\ c\prn{x}\prn{k}}}\IsDef}{
      \mrow{
        = \forall k:K\mathpunct{.}\prn{p:u\IsDef}\times c\prn{u_p}\prn{k}\IsDef
      }
      \commentrow{// because $K$ is inhabited}
      \mrow{
        = \prn{p:u\IsDef}\times\forall k:K\mathpunct{.} c\prn{u_p}\prn{k}\IsDef
      }
      \mrow{
        = \prn[\big]{x\gets u;\ \Scope~k:K~\In~c\prn{x}\prn{k}}\IsDef
      }
    }
  }

  On this common support, a complete witness consists of a pair $\prn{p, q}$ with $p: u\IsDef$ and $q(k) : c\prn{u_p}\prn{k}\IsDef$; both sides take the value $\lambda k\mathpunct{.} c\prn{u_p}\prn{k}_{q\prn{k}}$ on this support.
\end{proof}
\fi

\begin{lemma}[Interchange under strengthening]\label{lem:strengthened-scope}
  Let $K$ be compact and suppose that $\prn{\forall k:K\mathpunct{.}\varphi}\to\varphi$ for every $\varphi\in\OProp$. (This holds when $K$ is inhabited, and for $K=\Tm\prn{\alpha}$ whenever the judgemental structure satisfies strengthening, \S~\ref{sec:modular-elaborators}.) Then for any type $X$, family of types $Y\colon K\to\SET$, partial element $u:\PMC\prn{X}$ and function $c\colon X\to\prn{k:K}\to\PMC\prn{Y\prn{k}}$ we have
  \[
  \Scope~k:K~\In~\prn{x\gets u;\ c\prn{x}\prn{k}}
  =
  x\gets u;\ \Scope~k:K~\In~c\prn{x}\prn{k}\text{.}
  \]
  In particular, taking $c\prn{x}\prn{k}\coloneq\Return~x$, we have \[ \Scope~k:K~\In~u=\prn{x\gets u;\ \Return~\prn{\lambda k\mathpunct{.}x}}\text{.}\] 
\end{lemma}
\ifextended
\begin{proof}
  As for Lemma~\ref{lem:inhabited-scope}: inhabitedness of $K$ was used only to pass from $\forall k\mathpunct{.}\prn{p:u\IsDef}\times c\prn{u_p}\prn{k}\IsDef$ to $u\IsDef$. Here this is the hypothesis applied to the open proposition $\varphi\coloneq u\IsDef$. Since $u\IsDef$ is a proposition, the witnesses $p$ for the various $k$ agree, so the support is $\prn{p:u\IsDef}\times\forall k\mathpunct{.}c\prn{u_p}\prn{k}\IsDef$ and both sides take the value $\lambda k\mathpunct{.}c\prn{u_p}\prn{k}_{q\prn{k}}$ on it.
\end{proof}
\fi

\subsection{Algebras; multilinear and multistrict functions}\label{sec:algebras}

Since $\PMC$ is a \emph{relative} 2-monad along $\DiscIncl\colon\SET\to\POSET$, it would not make sense for an algebra to be presented by a structure map $\PMC\prn{A}\to A$: the carrier $A$ is a \emph{poset}.\footnote{Of course, the partial element classifier also naturally gives a monad on $\POSET$, but we will not use this structure.} We use instead the extension-form presentation of \citet{altenkirch-chapman-uustalu:2015}, carried over to the two-dimensional setting as defined in \citet{arkor-saville-slattery:2025}; concretely, the extension operator of an algebra is required to be monotone both as a map out of $\PMC\prn{X}$ and in its argument, the latter being the algebraic analogue of the local monotonicity of $\prn{-}^\dagger$ recorded in Construction~\ref{con:L-relative-monad}.

\begin{definition}[Algebra]\label{def:algebra}
  An \emph{algebra} is a poset $A$ together with an operation assigning to each type $X$ and each map $f\colon X\to A$ a monotone map $f^\dagger\colon\PMC\prn{X}\to A$ such that $f\preccurlyeq g$ implies $f^\dagger\preccurlyeq g^\dagger$, subject to the following equations:
  \begin{enumerate}
    \item \textbf{(unit)} $f^\dagger\circ\eta_X=f$; and
    \item \textbf{(composition)} $\prn{f^\dagger\circ g}^\dagger=f^\dagger\circ g^{\dagger}$ for every $g\colon X'\to\PMC\prn{X}$.\footnote{Naturality of the extension operator is automatic: for $h\colon X'\to X$, taking $g\coloneq\eta_X\circ h$ in the associativity law and using the unit law yields $\prn{f\circ h}^\dagger=f^\dagger\circ\PMC\prn{h}$. We therefore need not impose it separately.}
  \end{enumerate}
\end{definition}

\begin{example}[Free algebras]
  Every type $Z$ determines a \emph{free} algebra, carried by the poset $\PMC\prn{Z}$, whose extension operator is given by Kleisli extension. The algebra laws are precisely the relative 2-monad laws of Construction~\ref{con:L-relative-monad}.
\end{example}

\begin{lemma}[Products of algebras]\label{lem:products}
  For any family of algebras $\prn{A_i}_{i:I}$, the product poset $\prod_{i:I}A_i$ carries an algebra structure computed pointwise: given $f\colon X\to \prod_{i:I}A_i$, we define $f^\dagger\prn{u}\coloneq\lambda i\mathpunct{.}\prn{\pi_i\circ f}^\dagger\prn{u}$, and is the product in $\mathcal{L}\text{-}\mathbf{Alg}$.
\end{lemma}

\subsubsection{Multilinear and lax multilinear functions}

There is a standard notion of homomorphism between algebras in extension form: the \emph{linear} map.

\begin{definition}[Linear map]\label{def:linear-map}
  A monotone function $h\colon A\to B$ between algebras is called \emph{linear} when $h\circ f^\dagger=\prn{h\circ f}^\dagger$ for every $f\colon X\to A$. Algebras and linear maps form a category $\mathcal{L}\text{-}\mathbf{Alg}$.
\end{definition}

Because we do not wish to require that tensor products exist (as we shall discuss in Remark~\ref{rem:nonrepresentable}), we will have need of a more general notion of homomorphism in many variables, defined below.

\begin{definition}[Multilinear map]\label{def:multilinear}
  Let $K$ be compact, $\prn{A_k}_{k:K}$ a family of algebras, and $B$ an algebra. A monotone map $f\colon\prod_{k:K}A_k\to B$ is \emph{multilinear} when, for every family of types $\prn{X_k}_{k:K}$ and every family of maps $a_k\colon X_k\to A_k$, the following holds for each $u:\prod_k\PMC\prn{X_k}$:
  \[
  x \gets \prn{\Scope~k:K~\In~u_k};\ f\prn{\lambda k\mathpunct{.}a_k\prn{x_k}}
  =
  f\prn{ \lambda k\mathpunct{.}\ x\gets u_k;\ a_k\prn{x} }\text{.}
  \]
\end{definition}
\begin{lemma}[$1$-ary multilinearity]\label{lem:1ary}
  A $1$-ary multilinear map is exactly a linear map in the standard sense; in particular the identity map $1_A$ is multilinear.
\end{lemma}
\ifextended
\begin{proof}
  For $K=\mathbf{1}$ we have $\boldsymbol{\vartheta}_{\mathbf{1},X}=1_X$ by Lemma~\ref{lem:theta-comp}, so the condition reads $f\circ a^\dagger=\prn{f\circ a}^\dagger$ for all $a\colon X\to A$, which is the linearity law.
\end{proof}
\fi

The following generalises Lemma~\ref{lem:duplication} from the case $K=\mathbf{2}$.

\begin{lemma}[Diagonal scoping]\label{lem:diagonal-scoping}
  Let $K$ be compact and inhabited; for any $u:\PMC\prn{X}$, we have
  $
    \prn{\Scope~k:K~\In~u}=
    \PMC\prn{\Delta_K}\prn{u}
  $ where $\Delta_K\colon X\to X^K$ is the evident diagonal map.
\end{lemma}

\begin{proof}
  The left-hand support is $\forall k:K\mathpunct{.} u\IsDef$, which is equivalent to $u\IsDef$ because $K$ is inhabited. When $u\IsDef$, both sides converge to the constant function of $K$ returning the value of $u$.
\end{proof}

\begin{corollary}[Diagonal multilinearity]\label{cor:diagonal}
  Let $f\colon\prod_{k:K}A_k\to B$ be multilinear with $K$ compact and inhabited. For $u:\PMC\prn{X}$ and maps $a_k\colon X\to A_k$ we have
  $
  f\prn{\lambda k\mathpunct{.}x\gets u;\ a_k\prn{x}}
  =
  x\gets u;\ f\prn{\lambda k\mathpunct{.}a_k\prn{x}}\text{.}
  $
  That is, a multilinear map may be applied to several copies of the same partial element, which is then bound once.
\end{corollary}
\ifextended
\begin{proof}
  Apply multilinearity of $f$ to the constant family $u_k\coloneq u$.

  \iblock{
    \mhang{
      f\prn{\lambda k\mathpunct{.} x\gets u; a_k(x)}
    }{
      \commentrow{// by multilinearity of $f$}
      \ieqrow{
        x\gets \prn{\Scope~k:K~\In~u}; f\prn{a_k\prn{x_k}}
      }
      \commentrow{// by Lemma~\ref{lem:diagonal-scoping}}
      \ieqrow{
        x\gets \PMC\prn{\Delta_K}\prn{u}; f\prn{a_k\prn{x_k}}
      }
      \commentrow{// by definition of $\PMC\prn{\Delta_K}$}
      \ieqrow{
        x\gets u; f\prn{a_k\prn{x}}
      }
      \qedhere
    }
  }
\end{proof}
\fi

Multilinear maps of finite arity, then, equip the algebras with the structure of a \emph{multicategory} $\underline{\mathcal{L}}$.

\begin{remark}[Representability?]\label{rem:nonrepresentable}
  In general, we do not expect $\underline{\mathcal{L}}$ to be \emph{representable} by a monoidal category: tensor products of algebras can be constructed assuming that $\PMC$ preserve reflexive coequalisers, which amounts to a very strong condition on $\OProp$.
\end{remark}

The multilinearity equation of Definition~\ref{def:multilinear} asks for an \emph{equality}. In the two-dimensional setting we can also consider a notion of \emph{lax} multilinearity that replaces equality by inequality, but this turns out to be satisfied automatically by every monotone map, so we do not consider it further.

\subsubsection{Multistrict functions and undefinedness}

When the dominance is strict in the sense that $\OProp\ni\bot$, we can relate multilinearity to a more intuitive property: \emph{strictness}. First we note that any algebra supports an undefined element:

\begin{notation}[The $\mathbf{undefined}$ notation]
  For an algebra $A$, we shall write ``$\mathbf{undefined}_A$'' for the element $p\gets \Await~\bot; \mathbf{abort}_A\prn{p}$ of $A$ defined by extending the universal map $\mathbf{abort}_A\colon \bot\to A$.
\end{notation}

\begin{definition}[Multistrict map]
  A monotone map $f\colon \prod_{k:K}A_k\to B$ is \emph{multistrict} when for every family $a:\prod_{k:K}A_k$ such that $a_l=\mathbf{undefined}_{A_l}$ for some $l:K$, we have $f\prn{a} = \mathbf{undefined}_B$.
\end{definition}

\begin{lemma}\label{lem:multilinear-multistrict}
  Let $K$ be compact and have decidable equality (\emph{e.g.}\ a finite cardinal); then any multilinear morphism $f\colon \prod_{k:K}A_k\to B$ is multistrict.
\end{lemma}

\ifextended
\begin{proof}
  Fix $a:\prod_{k:K}A_k$ and $l:K$ such that $a_l=\mathbf{undefined}_{A_l}$.
  We consider the family of partial elements $u:\prod_{k:K}\PMC\prn{k\not=l}$ obtained by setting $u_k \coloneq \Await~\prn{k\not= l}$, which we may define only by our assumption that $K$ has decidable equality and is thus Hausdorff.
  By multilinearity of $f$, we have
  $
  f\prn{\lambda k\mathpunct{.} \_\gets u_k; a_k} =
  \_\gets \prn{\Scope~k:K~\In~u_k}; f\prn{a}
  $
  and thus
  \[
  f\prn{\lambda k\mathpunct{.} \Await~\prn{k\not=l}; a_k}
  =
  \_\gets\Await~\prn{\forall k:K\mathpunct{.} k\not=l};
  f\prn{a}
  \text{.}
  \]

  On the left-hand side, we have
  $
  \Await~\prn{k\not= l};a_k = a_k
  $
  by cases on whether $k=l$:
  \begin{enumerate}
    \item If $k=l$, then we have $\Await~\prn{k\not= l};a_k = \mathbf{undefined}_A = a_k$ with the last step following by assumption.
    \item If $k\neq l$ then $\Await~\prn{k\not= l};a_k = a_k$ is immediate.
  \end{enumerate}

  On the right-hand side, $\Await~\prn{\forall k:K\mathpunct{.}k\not=l}$ is undefined so we deduce
  $
  f\prn{a} = \mathbf{undefined}_B
  $.
\end{proof}

We do not expect the converse to Lemma~\ref{lem:multilinear-multistrict} to hold.

\fi %

\section{Elaborators \`a la Carte, over a dominance}\label{sec:modular-elaborators}

We now describe a semantic notion of elaborator in the partiality relative 2-monad (\S\ref{sec:partiality-monad}) over an arbitrary dominance $\OProp$. Rather than fixing in advance a source language and target language, we will develop the two modularly in lockstep.

\begin{definition}[Elaborable judgemental structure]
  A judgemental structure is called \emph{elaborable} when $\Tp$ is discrete and for each $\alpha:\Tp$, the component $\Tm\prn{\alpha}$ is both discrete and compact.
\end{definition}

\begin{definition}[Strengthening]\label{def:strengthening}
  A judgemental structure is said to satisfy \emph{strengthening} when for each $\alpha:\Tp$ and open proposition $\varphi\in \OProp$, we have $\prn{\forall x:\Tm\prn{\alpha}\mathpunct{.}\varphi}\to\varphi$.
\end{definition}

\begin{remark}[What is strengthening for?]
  The universe of open propositions is relatively unspecified: it has to be a dominance, and the elaborability condition says it has to be closed under term and type equality and universal quantification over terms of a fixed type. In a semantics that is built from the syntax of Martin-L\"of type theory, the strengthening property of Definition~\ref{def:strengthening} then amounts to precisely the admissible strengthening law of judgemental equality. Of course, not all type theories satisfy strengthening, so for the sake of extensibility we have carefully delineated which results rely on it.
\end{remark}

We specify the semantic domains for elaboration over a judgemental structure:
\[
\TypScript \coloneq \PMC\prn{\Tp}
\qquad
\ChkScript \coloneq \prn{\alpha:\Tp}\to\PMC\prn{\Tm\prn{\alpha}}
\qquad
\SynScript \coloneq \PMC\prn[\big]{\widetilde{\Tm}}
\qquad
\HypScript \coloneq \widetilde{\Tm}
\]

\subsection{Structural combinators}

Everything that follows takes place over an elaborable judgemental structure.

\paragraph{Variables and substitution}
  Hypotheses are included into synthesisable expressions using the monad unit:

  \iblock{\small
    \mrow{\ElabVar : \HypScript\to\SynScript}
    \mrow{\ElabVar~x \coloneq \Return~x}
  }

  Conversely, a synthesisable expression $E:\SynScript$ can be ``substituted'' for a hypothesis in a binder $k\colon \HypScript\to A$ valued in an algebra $A$ using the algebra's extension operation:
  $
    x\gets E; k\prn{x}
  $.

\paragraph{Conversion and annotation}

  The conversion rule of Martin-L\"of type theory corresponds to the following combinator:

  \iblock{\small
    \mrow{
      \ElabConv : \SynScript\to\ChkScript
    }
    \mhang{
      \ElabConv~E~\tau \coloneq \Do
    }{
      \mrow{
        \prn{\sigma,s}\gets E
      }
      \commentrow{// The following uses the assumption that $\Tp$ is discrete.}
      \mrow{
        \_\gets\Await~\sigma =_{\Tp} \tau
      }
      \commentrow{// In this scope, we have $\sigma = \tau : \Tp$ by equality reflection in the metalanguage.}
      \mrow{\Return~s}
    }
  }

  Conversely, a checkable expression may be annotated to be used in a synthesis-mode position.

  \iblock{\small
    \mrow{\ElabThe : \TypScript \to\ChkScript\to\SynScript}
    \mrow{
      \ElabThe~A~M\coloneq
      \sigma\gets A; s\gets M\prn{\sigma};
      \Return~\prn{\sigma,s}
    }
  }

\subsection{Elaborating dependent function types}

\begin{definition}
  A dependent function type structure over an elaborable judgemental structure is said to be \emph{elaborable} when for each $\sigma:\Tp$ the preimage $\TpPi^{-1}\prn{\sigma}\subseteq \Fam\prn{\Tp}$ is open-subsingleton.
\end{definition}

We now assume an elaborable dependent function type structure. The formation rule for dependent function types is coded as follows:

\iblock{
  \small
  \mrow{
    \ElabPi : \TypScript\to\prn{\HypScript\to\TypScript}\to\TypScript
  }
  \mhang{
    \ElabPi~A~B \coloneq \Do
  }{
    \mrow{\Var{dom}\gets A}
    \commentrow{// The following uses the assumption that $\Tm\prn{\Var{dom}}$ is compact.}
    \mrow{\Var{fam}\gets \Scope~x:\Tm\prn{\Var{dom}}~\In~B\prn{\Var{dom},x}}
    \mrow{
      \Return~\TpPi\prn{\Var{dom},\Var{fam}}
    }
  }
}

Next, we implement the introduction rule. In the case of dependent functions, it is possible to envision two different directionalities for $\lambda$-abstractions: checking mode without annotations, and synthesis-mode with a type annotation on the bound variable. We implement both below:

\smallskip
\begin{minipage}[t]{.45\textwidth}\small
\iblock{
  \mrow{
    \ElabLam :
      \prn{\HypScript\to\ChkScript}
      \to \ChkScript
  }
  \mhang{
    \ElabLam~M~\Var{typ} \coloneq \Do
  }{
    \commentrow{// The following uses the assumption that the preimage $\TpPi^{-1}\prn{\Var{typ}}$ is open-subsingleton.}
    \mrow{
      \prn{\Var{dom},\Var{fam}}\gets \Await~\TpPi^{-1}\prn{\Var{typ}}
    }
    \commentrow{// The following uses the assumption that $\Tm\prn{\Var{dom}}$ is compact.}
    \mrow{
      \Var{body}\gets \Scope~x:\Tm\prn{\Var{dom}}~\In~\allowbreak M\prn{\Var{dom},x}\prn{\Var{fam}\prn{x}}
    }
    \mrow{
      \Return~\TmLam\prn{\Var{dom},\Var{fam},\Var{body}}
    }
  }
}
\end{minipage}\quad
\begin{minipage}[t]{.45\textwidth}\small
\iblock{
  \mrow{
    \ElabLam' : \TypScript\to \prn{\HypScript\to\SynScript}\to\SynScript
  }
  \mhang{
    \ElabLam'~A~E \coloneq\Do
  }{
    \mrow{\Var{dom}\gets A}
    \commentrow{// The following uses the assumption that $\Tm\prn{\Var{dom}}$ is compact.}
    \mrow{\Var{body}\gets \Scope~x:\Tm\prn{\Var{dom}}~\In~E\prn{\Var{dom},x}}
    \mhang{
      \Return
    }{
      \mrow{
        \prn[\big]{
          \TpPi\prn{
            \Var{dom},
            \pi_1\circ\Var{body}
          },
          \allowbreak
          \TmLam\prn{
            \Var{dom},
            \pi_1\circ\Var{body},
            \pi_2\circ\Var{body}
          }
        }
      }
    }
  }
}
\end{minipage}
\smallskip

Function applications are implemented in synthesis mode as follows:

\iblock{\small
  \mrow{
    \ElabApp : \SynScript\to\ChkScript\to\SynScript
  }
  \mhang{
    \ElabApp~E~M \coloneq \Do
  }{
    \mrow{\prn{\Var{typ},\Var{fun}}\gets E}
    \commentrow{// The following uses the assumption that the preimage $\TpPi^{-1}\prn{\Var{typ}}$ is open-subsingleton.}
    \mrow{
      \prn{\Var{dom},\Var{fam}}\gets \Await~\TpPi^{-1}\prn{\Var{typ}}
    }
    \mrow{
      \Var{arg}\gets M\prn{\Var{dom}}
    }
    \mrow{
      \Return~\prn{
        \Var{fam}\prn{\Var{arg}},
        \TmApp\prn{\Var{dom},\Var{fam},\Var{fun},\Var{arg}}
      }
    }
  }
}

\subsection{Elaborating dependent pair types}
\begin{definition}
  A dependent pair type structure over an elaborable judgemental structure is said to be \emph{elaborable} when for each $\sigma:\Tp$ the preimage $\TpSg^{-1}\prn{\sigma}\subseteq\Fam\prn{\Tp}$ is open-subsingleton.
\end{definition}

We can now define appropriate elaborators over an elaborable dependent pair type structure.

\smallskip
\begin{minipage}[t]{.45\textwidth}\small
\iblock{
  \mrow{
    \ElabSg : \TypScript\to\prn{\HypScript\to\TypScript}\to\TypScript
  }
  \mhang{
    \ElabSg~A~B \coloneq \Do
  }{
    \mrow{\Var{base}\gets A}
    \commentrow{// The following uses the assumption that $\Tm\prn{\Var{base}}$ is compact.}
    \mrow{\Var{fam}\gets \Scope~x:\Tm\prn{\Var{base}}~\In~B\prn{\Var{base},x}}
    \mrow{
      \Return~\TpSg\prn{\Var{base},\Var{fam}}
    }
  }
  \row
  \mrow{
    \ElabPair : \ChkScript\to\ChkScript\to\ChkScript
  }
  \mhang{
    \ElabPair~M~N~\Var{typ} \coloneq \Do
  }{
    \commentrow{// The following uses the assumption that the preimage $\TpSg^{-1}\prn{\Var{typ}}$ is open-subsingleton.}
    \mrow{
      \prn{\Var{base},\Var{fam}}\gets \Await~\TpSg^{-1}\prn{\Var{typ}}
    }
    \mrow{
      \Var{fst} \gets M\prn{\Var{base}}
    }
    \mrow{
      \Var{snd} \gets N\prn{\Var{fam}\prn{\Var{fst}}}
    }
    \mrow{
      \Return~\TmPair\prn{
        \Var{base},
        \Var{fam},
        \Var{fst},
        \Var{snd}
      }
    }
  }
}
\end{minipage}\quad
\begin{minipage}[t]{.45\textwidth}\small
\iblock{
  \mrow{
    \ElabFst : \SynScript\to\SynScript
  }
  \mhang{
    \ElabFst~E \coloneq \Do
  }{
    \mrow{
      \prn{\Var{typ},\Var{pair}} \gets E
    }
    \commentrow{// The following uses the assumption that the preimage $\TpSg^{-1}\prn{\Var{typ}}$ is open-subsingleton.}
    \mrow{
      \prn{\Var{base},\Var{fam}}\gets \Await~\TpSg^{-1}\prn{\Var{typ}}
    }
    \mrow{
      \Return~\prn{
        \Var{base},
        \TmFst\prn{\Var{base},\Var{fam},\Var{pair}}
      }
    }
  }
  \row
  \mrow{
    \ElabSnd : \SynScript\to\SynScript
  }
  \mhang{
    \ElabSnd~E \coloneq \Do
  }{
    \mrow{
      \prn{\Var{typ},\Var{pair}} \gets E
    }
    \commentrow{// The following uses the assumption that the preimage $\TpSg^{-1}\prn{\Var{typ}}$ is open-subsingleton.}
    \mrow{
      \prn{\Var{base},\Var{fam}}\gets \Await~\TpSg^{-1}\prn{\Var{typ}}
    }
    \mhang{\Return}{
      \mrow{
        \prn{
          \Var{fam}\prn{\TmFst\prn{\Var{base},\Var{fam},\Var{pair}}},\allowbreak
          \TmSnd\prn{\Var{base},\Var{fam},\Var{pair}}
        }
      }
    }
  }
}
\end{minipage}
\smallskip

\subsection{Elaborating intensional identity types}

\begin{definition}
  An intensional identity type structure over an elaborable judgemental structure is said to be \emph{elaborable} when for each $\sigma:\Tp$ the preimage $\TpId^{-1}\prn{\sigma}\subseteq \prn{\alpha:\Tp}\times \Tm\prn{\alpha}\times \Tm\prn{\alpha}$ is open-subsingleton.
\end{definition}

The usual formation rule for identity types takes a type together with two elements of that type. Alternatively, we can leave the type out if the left-hand side is in synthesis mode. Either rule is easily implemented in our system.

\smallskip
\begin{minipage}[t]{.45\textwidth}\small
\iblock{
  \mrow{
    \ElabId : \TypScript\to\ChkScript\to\ChkScript\to\TypScript
  }
  \mhang{
    \ElabId~A~M~N \coloneq \Do
  }{
    \mrow{
      \Var{typ} \gets A
    }
    \mrow{
      \Var{lhs} \gets M\prn{\Var{typ}}
    }
    \mrow{
      \Var{rhs} \gets N\prn{\Var{typ}}
    }
    \mrow{
      \Return~\TpId\prn{\Var{typ},\Var{lhs},\Var{rhs}}
    }
  }
}
\end{minipage}\quad
\begin{minipage}[t]{.45\textwidth}\small
\iblock{
  \mrow{
    \ElabId_l : \SynScript\to\ChkScript\to\TypScript
  }
  \mhang{
    \ElabId_l~E~M\coloneq \Do
  }{
    \mrow{
      \prn{\Var{typ},\Var{lhs}}\gets E
    }
    \mrow{
      \Var{rhs}\gets M\prn{\Var{typ}}
    }
    \mrow{
      \Return~\TpId\prn{\Var{typ},\Var{lhs},\Var{rhs}}
    }
  }
}
\end{minipage}
\smallskip

We now come to a simple example of equational reasoning on the surface language:

\begin{lemma}\label{lem:id-in-terms-of-idl}
  We have $\ElabId~A~M~N = \ElabId_l~\prn{\ElabThe~A~M}~N$.
\end{lemma}
\ifextended
\begin{proof}
  Unfolding the definitions, we reason using the associativity and unit laws of the monad:
  
  \iblock{
    \mhang{
      \ElabId_l~\prn{\ElabThe~A~M}~N
    }{
      \ieqhang{
        \prn{\Var{typ},\Var{lhs}}\gets \ElabThe~A~M
      }{
        \mrow{\Var{rhs}\gets N\prn{\Var{typ}}}
        \mrow{\Return~\TpId\prn{\Var{typ},\Var{lhs},\Var{rhs}}}
      }
      \ieqhang{
        \Var{typ}\gets A
      }{
        \mrow{\Var{lhs}\gets M\prn{\Var{typ}}}
        \mrow{\Var{rhs}\gets N\prn{\Var{typ}}}
        \mrow{\Return~\TpId\prn{\Var{typ},\Var{lhs},\Var{rhs}}}
      }
      \ieqrow{\ElabId~A~M~N}
      \qedhere
    }
  }
\end{proof}
\fi

Assuming the identity type structure is elaborable, the $\TmRefl$ constructor is trivial to implement in checking and synthesis mode variants.

\smallskip
\begin{minipage}[t]{0.45\textwidth}\small
\iblock{
  \mrow{
    \ElabRefl : \ChkScript
  }
  \mhang{
    \ElabRefl~\Var{idtyp} \coloneq \Do
  }{
    \commentrow{// The following uses our assumption that the preimage $\TpId^{-1}\prn{\sigma}$ is open-subsingleton.}
    \mrow{
      \prn{\Var{typ},\Var{lhs},\Var{rhs}}\gets \Await~\TpId^{-1}\prn{\Var{idtyp}}
    }
    \commentrow{// The following uses our assumption that $\Tm\prn{\Var{typ}}$ is discrete.}
    \mrow{
      \Await~\prn{\Var{lhs}=_{\Tm\prn{\Var{typ}}}\Var{rhs}}
    }
    \mrow{
      \Return~\TmRefl\prn{\Var{typ},\Var{lhs}}
    }
  }
}
\end{minipage}\quad
\begin{minipage}[t]{0.45\textwidth}\small
\iblock{
  \mrow{
    \ElabRefl' : \SynScript\to\SynScript
  }
  \mhang{
    \ElabRefl'~E \coloneq \Do
  }{
    \mrow{
      \prn{\Var{typ},\Var{tm}} \gets E
    }
    \mrow{
      \Return~\prn{\TpId\prn{\Var{typ},\Var{tm},\Var{tm}},\TmRefl\prn{\Var{typ},\Var{tm}}}
    }
  }
}
\end{minipage}
\smallskip

\begin{lemma}\label{lem:refl'-in-terms-of-refl}
  We have $\ElabRefl'~\prn{\ElabThe~A~M} = \ElabThe~\prn{\ElabId~A~M~M}~\ElabRefl$.
\end{lemma}
\ifextended
\begin{proof}
  Unfolding the definitions and using associativity, the left-hand side is
  \[
    \Var{typ}\gets A;\ \Var{tm}\gets M\prn{\Var{typ}};\ \Return~\prn{\TpId\prn{\Var{typ},\Var{tm},\Var{tm}},\TmRefl\prn{\Var{typ},\Var{tm}}}\text{.}
  \]
  On the right-hand side, $\ElabId~A~M~M$ returns $\TpId\prn{\Var{typ},\Var{lhs},\Var{rhs}}$ after binding $\Var{typ}\gets A$ and then $\Var{lhs},\Var{rhs}\gets M\prn{\Var{typ}}$ separately. Since $\TpId^{-1}\prn{\TpId\prn{\Var{typ},\Var{lhs},\Var{rhs}}}$ is the singleton on $\prn{\Var{typ},\Var{lhs},\Var{rhs}}$, the combinator $\ElabRefl$ at this type is $\Await~\Var{lhs}=\Var{rhs};\ \Return~\TmRefl\prn{\Var{typ},\Var{lhs}}$. Hence the right-hand side is
  \[
    \Var{typ}\gets A;\ \Var{lhs}\gets M\prn{\Var{typ}};\ \Var{rhs}\gets M\prn{\Var{typ}};\ \Await~\Var{lhs}=\Var{rhs};\ \Return~\prn{\TpId\prn{\Var{typ},\Var{lhs},\Var{rhs}},\TmRefl\prn{\Var{typ},\Var{lhs}}}\text{.}
  \]
  By Lemma~\ref{lem:duplication} applied to $M\prn{\Var{typ}}$, the two bindings may be identified, $\Var{rhs}$ becoming $\Var{lhs}$; the assertion $\Await~\Var{lhs}=\Var{lhs}$ is trivial, and we obtain the left-hand side.
\end{proof}
\fi

Martin-L\"of's identification eliminator requires an induction motive and a base case.

\iblock{\small
  \mrow{
    \ElabIdElim :
    \SynScript\to
    \prn{
      \HypScript\times\HypScript\times\HypScript\to\TypScript
    }
    \to\prn{
      \HypScript\to\ChkScript
    }
    \to\SynScript
  }
  \mhang{
    \ElabIdElim~E~C~M \coloneq \Do
  }{
    \mrow{
      \prn{\Var{idtyp},\Var{target}} \gets E
    }
    \commentrow{// The following uses our assumption that the preimage $\TpId^{-1}\prn{\sigma}$ is open-subsingleton.}
    \mrow{
      \prn{\Var{typ},\Var{lhs},\Var{rhs}}\gets \Await~\TpId^{-1}\prn{\Var{idtyp}}
    }
    \commentrow{// The following uses the assumption that each $\Tm\prn{\alpha}$ is compact.}
    \mhang{
      \Var{motive} \gets
    }{
      \mrow{\Scope~x:\Tm\prn{\Var{typ}}~\In}
      \mrow{\Scope~y:\Tm\prn{\Var{typ}}~\In}
      \mrow{\Scope~p:\Tm\prn{\TpId\prn{\Var{typ},x,y}}~\In}
      \mrow{C\prn[\big]{\prn{\Var{typ},x},\prn{\Var{typ},y},\prn{\TpId\prn{\Var{typ},x,y},p}}}
    }
    \mrow{
      \Var{baseCase} \gets
      \Scope~x:\Tm\prn{\Var{typ}}~\In~M\prn{x}\prn{\Var{motive}\prn{x,x,\TmRefl\prn{\Var{typ},x}}}
    }
    \mrow{
      \Return~\prn{\Var{motive}\prn{\Var{lhs},\Var{rhs},\Var{target}},\TmJ\prn{\Var{typ},\Var{lhs},\Var{rhs},\Var{target},\Var{motive},\Var{baseCase}}}
    }
  }
}

\subsection{Elaborating the answer type}

Assuming an answer type structure over an elaborable judgemental structure, we will not need any further elaborability conditions.

\smallskip
\begin{minipage}[t]{0.45\textwidth}
\iblock{\small
  \mrow{\ElabAns : \TypScript}
  \mrow{\ElabAns \coloneq \Return~\TpAns}
}
\end{minipage}
\begin{minipage}[t]{0.45\textwidth}
\iblock{\small
  \mrow{\ElabYes, \ElabNo: \ChkScript}
  \mrow{\ElabYes~\Var{typ} \coloneq \Await~\prn{\Var{typ}=_{\Tp} \TpAns}; \Return~\TmYes}
  \mrow{\ElabNo~\Var{typ} \coloneq \Await~\prn{\Var{typ}=_{\Tp} \TpAns}; \Return~\TmNo}
}
\end{minipage}
\smallskip

\subsection{Case study: elaborating a complex type}\label{sec:case-study}

We can now illustrate a bidirectional elaboration script and its evaluation to a term. We will consider in particular the type of the \emph{symmetry} function that inverts an identification using Martin-L\"of's identification eliminator. We will show by simple equational reasoning how to execute the elaboration of the following script:

\iblock{\small
  \mrow{
    \ScriptIdent{symmetry.type} : \TypScript
  }
  \mhang{
    \ScriptIdent{symmetry.type} \coloneq
  }{
    \commentrow{//
      $\prn{x : \ElabAns}\,\prn{y:\ElabAns}\,\prn{p:x= y}
      \to
      y=x$
    }
    \mrow{
      \ElabPi~\ElabAns~\prn{
        \lambda x\mathpunct{.}
        \ElabPi~\ElabAns~\prn{
          \lambda y\mathpunct{.}
          \ElabPi~\prn{\ElabId_l~\prn{\ElabVar~x}~\prn{\ElabConv~\prn{\ElabVar~y}}}~\prn{
            \lambda p\mathpunct{.}
            \ElabId_l~\prn{\ElabVar~y}~\prn{\ElabConv~\prn{\ElabVar~x}}
          }
        }
      }
    }
  }
}

\begin{computation}\label{computation:symmetry.type}
  We can evaluate $\ScriptIdent{symmetry.type}$ as a partial element of $\Tp$ as follows.
  First we unfold the definition of $\ElabPi$:

  \iblock{\small
    \mrow{\alpha\gets \ElabAns}
    \mhang{
      \beta\gets \Scope~x:\Tm\prn{\alpha}~\In~\Do
    }{
      \mrow{\alpha'\gets\ElabAns}
      \mhang{\beta'\gets\Scope~y:\Tm\prn{\alpha'}~\In~\Do}{
        \mrow{\alpha''\gets \ElabId_l~\prn{\ElabVar~\prn{\alpha,x}}~\prn{\ElabConv~\prn{\ElabVar~\prn{\alpha',y}}}}
        \mrow{\beta''\gets\Scope~p:\Tm\prn{\alpha''}~\In~\ElabId_l~\prn{\ElabVar~\prn{\alpha',y}}~\prn{\ElabConv~\prn{\ElabVar~\prn{\alpha,x}}}}
        \mrow{\Return~\TpPi\prn{\alpha'',\beta''}}
      }
      \mrow{\Return~\TpPi\prn{\alpha',\beta'}}
    }
    \mrow{\Return~\TpPi\prn{\alpha,\beta}}
  }

  Next we unfold the definitions of $\ElabAns$ and $\ElabVar$:

  \iblock{\small
    \mhang{
      \beta\gets \Scope~x:\Tm\prn{\TpAns}~\In~\Do
    }{
      \mhang{\beta'\gets\Scope~y:\Tm\prn{\TpAns}~\In~\Do}{
        \mrow{\alpha''\gets \ElabId_l~\prn{\Return~\prn{\TpAns,x}}~\prn{\ElabConv~\prn{\Return~\prn{\TpAns,y}}}}
        \mrow{\beta''\gets\Scope~p:\Tm\prn{\alpha''}~\In~\ElabId_l~\prn{\Return~\prn{\TpAns,y}}~\prn{\ElabConv~\prn{\Return~\prn{\TpAns,x}}}}
        \mrow{\Return~\TpPi\prn{\alpha'',\beta''}}
      }
      \mrow{\Return~\TpPi\prn{\TpAns,\beta'}}
    }
    \mrow{\Return~\TpPi\prn{\TpAns,\beta}}
  }

  \smallskip
  \begin{minipage}[t]{0.45\textwidth}
  Unrolling the definition of $\ElabId_l$:

  \iblock{\small
    \mhang{
      \beta\gets \Scope~x:\Tm\prn{\TpAns}~\In~\Do
    }{
      \mhang{\beta'\gets\Scope~y:\Tm\prn{\TpAns}~\In~\Do}{
        \mrow{
          y'\gets \ElabConv~\prn{\Return~\prn{\TpAns,y}}~\TpAns
        }
        \mhang{\beta''\gets\Scope~p:\Tm\prn{\TpId\prn{\TpAns,x,y'}}~\In~\Do}{
          \mrow{
            x'\gets \ElabConv~\prn{\Return~\prn{\TpAns,x}}~\TpAns
          }
          \mrow{
            \Return~\TpId\prn{\TpAns,y,x'}
          }
        }
        \mrow{\Return~\TpPi\prn{\TpId\prn{\TpAns,x,y'},\beta''}}
      }
      \mrow{\Return~\TpPi\prn{\TpAns,\beta'}}
    }
    \mrow{\Return~\TpPi\prn{\TpAns,\beta}}
  }
  \end{minipage}\quad
  \begin{minipage}[t]{0.45\textwidth}
  Unrolling the definition of $\ElabConv$:

  \iblock{\small
    \mhang{
      \beta\gets \Scope~x:\Tm\prn{\TpAns}~\In~\Do
    }{
      \mhang{\beta'\gets\Scope~y:\Tm\prn{\TpAns}~\In~\Do}{
        \mrow{\Await~\prn{\TpAns=_{\Tp}\TpAns}}
        \mhang{\beta''\gets\Scope~p:\Tm\prn{\TpId\prn{\TpAns,x,y}}~\In~\Do}{
          \mrow{\Await~\prn{\TpAns=_{\Tp}\TpAns}}
          \mrow{
            \Return~\TpId\prn{\TpAns,y,x}
          }
        }
        \mrow{\Return~\TpPi\prn{\TpId\prn{\TpAns,x,y},\beta''}}
      }
      \mrow{\Return~\TpPi\prn{\TpAns,\beta'}}
    }
    \mrow{\Return~\TpPi\prn{\TpAns,\beta}}
  }
  \end{minipage}
  \medskip

  The conditions of both $\Await$-statements are true, so we have:

  \iblock{\small
    \mhang{
      \beta\gets \Scope~x:\Tm\prn{\TpAns}~\In~\Do
    }{
      \mhang{\beta'\gets\Scope~y:\Tm\prn{\TpAns}~\In~\Do}{
        \mrow{\beta''\gets\Scope~p:\Tm\prn{\TpId\prn{\TpAns,x,y}}~\In~\Return~\TpId\prn{\TpAns,y,x}}
        \mrow{\Return~\TpPi\prn{\TpId\prn{\TpAns,x,y},\beta''}}
      }
      \mrow{\Return~\TpPi\prn{\TpAns,\beta'}}
    }
    \mrow{\Return~\TpPi\prn{\TpAns,\beta}}
  }

  Finally, we rewrite using the unit law of the $\Scope$-expression (Lemma~\ref{lem:theta-unit}) and the monad laws:

  \iblock{
    \mrow{
      \Return~\TpPi\prn{\TpAns, \lambda x\mathpunct{.}\TpPi\prn{\TpAns,\lambda y\mathpunct{.}\TpPi\prn{\TpId\prn{\TpAns,x,y},\lambda p\mathpunct{.}\TpId\prn{\TpAns,y,x}}}}
    }
  }
\end{computation}

\begin{remark}
  Because all of our reasoning is \emph{equational} and all functions are (automatically) congruent for equality, the individual steps of Computation~\ref{computation:symmetry.type} could have been carried out in any order, factored into reusable lemmas, \emph{etc.}. This is the benefit of explaining elaboration \emph{denotationally}.
\end{remark}

\subsection{Scoped elaboration in context}\label{sec:elaboration-in-context}

Much as \citet{mcbride:ni} has advocated, the \emph{context} of elaboration is not an explicit parameter to the elaboration process but rather emerges implicitly through the use of the $\Scope$ combinator. Nonetheless, an actual elaboration script (as written by a user in a proof assistant) will take place in the scope of a number of free hypotheses $x_1,\ldots,x_n:\HypScript$, and this script will be executed relative to an actual well-formed type context $x_1:A_1,\ldots,x_n:A_n$ of the same length. We show in this section that explicitly scoped elaboration scripts like this can be made sense of directly in our framework, in the presence of an elaborable judgemental structure.

\subsubsection{Telescopes as formal type contexts}

We define a length-indexed inductive type of telescopes to represent formal type contexts.

\iblock{\small
  \mhang{
    \mathbf{data}~\Tele \colon \mathbb{N}\to\SET~\mathbf{where}
  }{
    \mrow{\epsilon \colon \Tele~0}
    \mrow{
      \prn{\lhd}\colon \prn{\alpha:\Tp}\to \prn{
      \Tm\prn{\alpha}\to \Tele_n
      }
      \to
      \Tele_{n+1}
    }
  }
}

Then we define for each telescope the type of \emph{environments} as an iterated dependent pair type.

\iblock{\small
  \mrow{
    \Env_n\colon \Tele_n\to\SET
  }
  \mrow{
    \Env_0\prn{\epsilon} \coloneq \mathbf{1}
  }
  \mrow{
    \Env_{n+1}\prn{\alpha\lhd \Psi} \coloneq \prn{x:\Tm\prn{\alpha}}\times\Env_n\prn{\Psi\prn{x}}
  }
}

Each $\Env_n\prn{\psi}$ is compact.
We will write $\widetilde{\Env}_n$ for $\prn{\Psi:\Tele_n}\times \Env_n\prn{\Psi}$ and $\widetilde{\Env}$ for $\prn{n:\mathbb{N}}\times \widetilde{\Env}_n$.

\begin{construction}[Environment zipping]
  Indexed in telescopes $\Psi:\Tele_n$ we have a family of ``zipping'' functions

  \iblock{\small
    \mrow{\Zip_\Psi \colon \Env_n\prn{\Psi}\to \HypScript^n}
    \mrow{\Zip_{\epsilon}\prn{*} \coloneq *}
    \mrow{\Zip_{\alpha\lhd\Psi}\prn{x,\psi} \coloneq \prn[\big]{\prn{\alpha,x}, \Zip_{\Psi\prn{x}}\prn{\psi}}}
  }

  \noindent
  that turn dependently typed environments into lists of typed terms.
\end{construction}

\begin{construction}[Domains for scoped elaboration]
Let $\Psi:\Tele_n$ be a telescope; then we define:
\begin{align*}
  \brk{\Psi}\TypScript &\coloneq \PMC\prn{\Env_n\prn{\Psi}\to\Tp}
  \\
  \brk{\Psi}\SynScript &\coloneq \PMC\prn[\big]{\Env_n\prn{\Psi}\to\widetilde{\Tm}}
  \\
  \brk{\Psi}\ChkScript &\coloneq
  \prn{\alpha:\Env_n\prn{\Psi}\to \Tp}
  \to
  \PMC\prn[\big]{\prn{\psi:\Env_n\prn{\Psi}}\to \Tm\prn{\alpha\prn{\psi}}}
\end{align*}
\end{construction}

\begin{construction}[Executing scoped elaboration scripts]
  We define the following combinators for executing scoped elaboration scripts at $\Psi$:
  \begin{mathpar}
    \ebrule{
      \hypo{\Psi:\Tele_n}
      \hypo{A\colon\mathsf{hyp}^n\to\TypScript}
      \infer2{
        \brk{\Psi\vdash A} :\brk{\Psi}\TypScript
      }
    }
    \and
    \ebrule{
      \hypo{\Psi:\Tele_n}
      \hypo{E\colon\mathsf{hyp}^n\to\SynScript}
      \infer2{
        \brk{\Psi\vdash E} : \brk{\Psi}\SynScript
      }
    }
    \and
    \ebrule{
      \hypo{\Psi:\Tele_n}
      \hypo{M\colon\mathsf{hyp}^n\to\ChkScript}
      \infer2{
        \brk{\Psi\vdash M} : \brk{\Psi}\ChkScript
      }
    }
  \end{mathpar}

  These are defined as follows, using the compactness of $\Env_n\prn{\Psi}$:

  \iblock{\small
    \mrow{
      \brk{\Psi\vdash A} \coloneq \Scope~\psi:\Env_n\prn{\Psi}~\In~A\prn{\Zip_\Psi\prn{\psi}}
    }
    \mrow{
      \brk{\Psi \vdash E} \coloneq \Scope~\psi:\Env_n\prn{\Psi}~\In~E\prn{\Zip_\Psi\prn{\psi}}
    }
    \mrow{
      \brk{\Psi\vdash M}~\alpha \coloneq
      \Scope~\psi:\Env_n\prn{\Psi}~\In~
      M\prn{\Zip_\Psi\prn{\psi}}\prn{\alpha\prn{\psi}}
    }
  }
\end{construction}

\section{Denotational reasoning about surface syntax}\label{sec:reasoning}

We are now in a position to begin reasoning about elaboration scripts using the semantics that we have given. Although we cannot exhaust the topic in a single paper, we present two possible areas of denotational exploitation: multilinearity as a sanity condition for good elaboration combinators, and an inequational version of subject reduction for untyped elaboration scripts.

\subsection{Multilinearity of elaboration combinators}\label{sec:multilinearity-of-elaboration}

In \S~\ref{sec:algebras}, we introduced the notion of \emph{algebra} and \emph{multilinear map} for the relative 2-monad $\PMC\colon \SET\to\POSET$. We recall that the main semantic domains for elaboration are actually algebras:
\begin{enumerate}
  \item $\TypScript = \PMC\prn{\Tp}$ and $\SynScript = \PMC\prn[\big]{\widetilde{\Tm}}$ are both \emph{free} algebras.
  \item $\ChkScript = \prod_{\alpha:\Tp}\PMC\prn{\Tm\prn{\alpha}}$ is a \emph{product} of free algebras.
  \item The domains for binders, like $\HypScript\to \TypScript$ and $\HypScript\to\ChkScript$ are powers $\HypScript\pitchfork\TypScript = \prod_{\_:\HypScript}\TypScript$ and $\HypScript\pitchfork\ChkScript = \prod_{\_:\HypScript}\ChkScript$ of algebras respectively.
\end{enumerate}

Hence we may consider whether a given combinator built from the above is \emph{multilinear}.

\begin{theorem}[Multilinearity of elaboration combinators]\label{thm:multilinearity-of-combinators}
  The following elaboration combinators from \S~\ref{sec:modular-elaborators} are multilinear in all their arguments and hence track the following multimorphisms:
  \par\vspace{-1em}
  \begin{minipage}[t]{0.45\textwidth}\small
    \begin{align*}
      \underline{\ElabVar} &: \brk{}\multimap \HypScript\pitchfork\SynScript\\
      \underline{\ElabConv} &: \brk{\SynScript}\multimap\ChkScript\\
      \underline{\ElabThe} &: \brk{\TypScript,\ChkScript}\multimap \SynScript\\
      \underline{\ElabApp} &:\brk{\SynScript,\ChkScript}\multimap\SynScript\\
      \underline{\ElabPair} &:\brk{\ChkScript,\ChkScript}\multimap\ChkScript\\
      \underline{\ElabFst},\underline{\ElabSnd} &: \brk{\SynScript}\multimap\SynScript
    \end{align*}
  \end{minipage}
  \begin{minipage}[t]{0.45\textwidth}\small
    \begin{align*}
      \underline{\ElabId}&:\brk{\TypScript,\ChkScript,\ChkScript}\multimap\TypScript\\
      \underline{\ElabId}_l &:\brk{\SynScript,\ChkScript}\multimap\TypScript\\
      \underline{\ElabRefl} &: \brk{}\multimap\ChkScript\\
      \underline{\ElabRefl}'&:\brk{\SynScript}\multimap\SynScript\\
      \underline{\ElabAns} &:\brk{}\multimap\TypScript\\
      \underline{\ElabYes},\underline{\ElabNo} &:\brk{}\multimap\ChkScript
    \end{align*}
  \end{minipage}
  \vspace{-.2em}
  {\small
    \[ \underline{\ElabIdElim}:\brk{\SynScript,\HypScript^3\pitchfork\TypScript,\HypScript\pitchfork\ChkScript}\multimap\SynScript\]}

  When the judgemental structure satisfies strengthening, the remaining elaboration combinators involving binding are also multilinear:
  \par\vspace{-1em}
  \begin{minipage}[t]{0.45\textwidth}\small
    \begin{align*}
      \underline{\ElabPi} &:\brk{\TypScript,\HypScript\pitchfork\TypScript}\multimap \TypScript\\
      \underline{\ElabSg}&:\brk{\TypScript,\HypScript\pitchfork\TypScript}\multimap\TypScript
    \end{align*}
  \end{minipage}
  \begin{minipage}[t]{0.45\textwidth}\small
    \begin{align*}
      \underline{\ElabLam} &:\brk{\HypScript\pitchfork\ChkScript}\multimap\ChkScript\\
      \underline{\ElabLam}' &:\brk{\TypScript,\HypScript\pitchfork\SynScript}\multimap\SynScript\\
    \end{align*}
  \end{minipage}
\end{theorem}

\begin{proof}
  Every case follows by equational reasoning (the
  calculations are carried out in the extended version of this paper); for $\underline{\ElabIdElim}$ we use Lemma~\ref{lem:inhabited-scope}.
\end{proof}

Multilinearity in turn implies \emph{multistrictness}, which we believe is an important sanity condition for elaboration rules.

\begin{corollary}[Multistrictness of elaboration combinators]
  If $\OProp$ is strict and the judgemental structure satisfies strengthening, every elaboration combinator from \S~\ref{sec:modular-elaborators} is multistrict. (The concrete dominance of \S\ref{sec:ldec} is strict, Proposition~\ref{prop:oprop-strict}.)
\end{corollary}

In other words, if any ``subterm'' of an elaboration script built from these combinators is undefined, the entire script is undefined. This is an important sanity condition for surface languages: in those terms, it captures the idea that a well-typed term will necessarily have no ill-typed subterms.

\ifextended
\begin{proof}
  By Theorem~\ref{thm:multilinearity-of-combinators} via Lemma~\ref{lem:multilinear-multistrict}.
\end{proof}
\fi

\subsection{Shadows of subject reduction}\label{sec:subject-reduction}

An important property of certain declarative presentations of type theory is \emph{subject reduction} with respect to an (untyped) rewriting system: if $\Gamma\vdash M:A$ holds and $M\mapsto M'$, then $\Gamma\vdash M':A$ must hold. Of course, ``full spectrum'' dependent type theory is not typically presented by an untyped rewriting system, as rewriting methods have thus far never scaled to a satisfactory handling of the $\eta$-laws for dependent pair types. For that reason, discussion of ``subject reduction'' in the context of modern type theoretic systems has always had the air of science fiction about it; nonetheless, we believe that our denotational semantics of surface syntax as bidirectional elaboration combinators may provide some insight on the question.

For example, one might conjecture that a surface-level ``$\beta$-redex'' would be equal to its contractum, following the reduction rule suggested by \citet{mcbride:ni}:
\[
  \ElabApp~\prn{
    \ElabThe~\prn{
      \ElabPi~A~B
    }~\prn{
      \ElabLam~M
    }
  }~{N}
  \stackrel{?}{=}
  x\gets \ElabThe~{A}~{N};
  \ElabThe~B\prn{x}~M\prn{x}
\]
However, if we unravel the definitions of both sides, we will see that they are not equal without additional assumptions. In particular, the left and right hand sides unravel as follows:

\smallskip
\begin{minipage}[t]{.45\textwidth}\small
\iblock{
  \mhang{
    \ElabApp~\prn{
        \ElabThe~\prn{
          \ElabPi~A~B
        }~\prn{
          \ElabLam~M
        }
      }~{N} = \Do
  }{
    \mrow{\Var{dom}\gets A}
    \mrow{\Var{fam}\gets\Scope~x:\Tm\prn{\Var{dom}}~\In~B\prn{\Var{dom},x}}
    \mrow{\Var{body}\gets\Scope~x:\Tm\prn{\Var{dom}}~\In~\allowbreak M\prn{\Var{dom},x}\prn{\Var{fam}\prn{x}}}
    \mrow{\Var{arg}\gets N\prn{\Var{dom}}}
    \mrow{\Return~\prn{\Var{fam}\prn{\Var{arg}},\Var{body}\prn{\Var{arg}}}}
  }
}
\end{minipage}\quad
\begin{minipage}[t]{.45\textwidth}\small
\iblock{
  \mhang{
    x\gets \ElabThe~A~N;
      \ElabThe~\prn{B\prn{x}}~\prn{M\prn{x}}
    = \Do
  }{
    \mrow{\Var{dom}\gets A}
    \mrow{\Var{arg}\gets N\prn{\Var{dom}}}
    \mrow{\Var{fibre} \gets B\prn{\Var{dom},\Var{arg}}}
    \mrow{\Var{value} \gets M\prn{\Var{dom},\Var{arg}}\prn{\Var{fibre}}}
    \mrow{\Return~\prn{\Var{fibre},\Var{value}}}
  }
}
\end{minipage}
\smallskip

As can be seen by inspection, the difference between the two lies \emph{solely} in their definedness as partial elements. The left-hand side is defined only when $B$ denotes a well-typed family of types and when $M$ denotes a well-typed section of this family of types, whereas the right-hand side requires only that specific substitution instances of $B$ and $M$ be well-typed. The former is stronger than the latter, which suggests that subject reductions like this may correspond to \emph{inequational} reasoning in the partiality monad. Indeed:

\begin{theorem}[``Subject reduction'']\label{thm:subject-reduction}
  The following inequalities corresponding to untyped $\beta$- and annotation-removal reductions hold, where all metavariables are quantified universally:
  \begin{sizeddisplay}{\small}
  \begin{align*}
    \ElabThe~A~\prn{\ElabConv~E}
    &\preccurlyeq_{\SynScript}
    E
    \\
    \ElabApp~\prn{
      \ElabThe~\prn{
        \ElabPi~A~B
      }~\prn{
        \ElabLam~M
      }
    }~{N}
    &\preccurlyeq_{\SynScript}
    x \gets \ElabThe~A~N;
    \ElabThe~{B\prn{x}}~{M\prn{x}}
    \\
    \ElabFst~\prn{
      \ElabThe~\prn{
        \ElabSg~A~B
      }~\prn{
        \ElabPair~M~N
      }
    }
    &\preccurlyeq_{\SynScript}
    \ElabThe~A~M
    \\
    \ElabSnd~\prn{
      \ElabThe~\prn{
        \ElabSg~A~B
      }~\prn{
        \ElabPair~M~N
      }
    }
    &\preccurlyeq_{\SynScript}
    \ElabThe~\prn{
      x \gets \ElabThe~A~M;
      B\prn{x}
    }~N
    \\
    \ElabIdElim~\prn{
      \ElabThe~\prn{\ElabId~A~N~N}~\ElabRefl
    }~B~M
    &\preccurlyeq_{\SynScript}
    x\gets\ElabThe~A~N;
    p\gets\ElabThe~\prn{\ElabId~A~N~N}~\ElabRefl;
    \ElabThe~B\prn{x,x,p}~M\prn{x}
  \end{align*}
  \end{sizeddisplay}
  \vspace{-1em}
\end{theorem}

\ifextended
\begin{proof}
  See Appendix~\ref{appendix:subject-reduction}.
\end{proof}
\fi

For $\eta$-conversion, we must be a little more careful. We do have an untyped $\eta$-contraction law for functions, which holds as an equation (and hence also as an inequation):

\begin{theorem}[$\eta$-contraction for functions]\label{thm:untyped-eta-reduction}
  Assume the judgemental structure satisfies strengthening.
  Then for all $E:\SynScript$ and $\Var{funtype}:\Tp$ we have the equation
  \[
    \ElabLam~\prn{
      \lambda x\mathpunct{.}
      \ElabConv~\prn{
        \ElabApp~E~\prn{\ElabConv~\prn{\ElabVar~x}}
      }
    }~\Var{funtype}
    =
    \prn{\_\gets\Await~\TpPi^{-1}\prn{\Var{funtype}};\ \ElabConv~E~\Var{funtype}}
  \]
  from which we may deduce the following \emph{inequation}:
  \[
    \ElabLam~\prn{
      \lambda x\mathpunct{.}
      \ElabConv~\prn{
        \ElabApp~E~\prn{\ElabConv~\prn{\ElabVar~x}}
      }
    }
    \preccurlyeq_{\ChkScript}
    \ElabConv~E
  \]
\end{theorem}

\ifextended
\begin{proof}
  See Appendix~\ref{appendix:untyped-eta-reduction}.
\end{proof}
\fi

The right-hand side of the equation above is $\ElabConv~E~\Var{funtype}$ with an additional definedness condition, namely that $\Var{funtype}$ is a dependent function type; this condition is vacuous whenever $\Var{funtype}$ is produced by $\ElabPi$, which is what makes the typed version below an equation. We will also use the following fact about annotations.

\begin{lemma}[Re-annotation]\label{lem:re-annotation}
  For all $\sigma:\TypScript$ and $M:\ChkScript$ we have $\ElabThe~\sigma~\prn{\ElabConv~\prn{\ElabThe~\sigma~M}}=\ElabThe~\sigma~M$.
\end{lemma}

\ifextended
\begin{proof}
  See Appendix~\ref{appendix:re-annotation}.
\end{proof}
\fi

For dependent pairs, the situation is considerably more subtle and we will need more annotations.

\begin{theorem}[$\eta$-equality for typed pairs]\label{thm:typed-pair-eta}
  For all $A:\TypScript$, $B\colon \HypScript\to\TypScript$, and $M:\ChkScript$ the following equation holds:
  \begin{align*}
    &\ElabThe~\prn{\ElabSg~A~B}~\prn{
      \ElabPair~\prn{
        \ElabConv~\prn{
          \ElabFst~\prn{\ElabThe~\prn{\ElabSg~A~B}~M}
        }
      }~\prn{
        \ElabConv~\prn{
          \ElabSnd~\prn{\ElabThe~\prn{\ElabSg~A~B}~M}
        }
      }
    }
    \\
    &\quad
    =
    \ElabThe~\prn{\ElabSg~A~B}~M
    \text{.}
  \end{align*}
\end{theorem}

\ifextended
\begin{proof}
  See Appendix~\ref{appendix:typed-pair-eta}.
\end{proof}
\fi

A similar equation holds for typed functions.

\begin{theorem}[$\eta$-equality for typed functions]\label{thm:typed-function-eta}
  Assume that the judgemental structure satisfies strengthening.
  Then for all $A:\TypScript$, $B\colon \HypScript\to\TypScript$, and $M:\ChkScript$ the following equation holds:
  \begin{align*}
    &\ElabThe~\prn{\ElabPi~A~B}~\prn{
      \ElabLam~\prn{
        \lambda x\mathpunct{.}
        \ElabConv~\prn{
          \ElabApp~\prn{\ElabThe~\prn{\ElabPi~A~B}~M}~\prn{
            \ElabConv~\prn{
              \ElabVar~x
            }
          }
        }
      }
    }
    \\
    &\quad=
    \ElabThe~\prn{\ElabPi~A~B}~M
    \text{.}
  \end{align*}
\end{theorem}

\begin{proof}
  See Appendix~\ref{appendix:typed-function-eta}.
\end{proof}

\begin{remark}
  Because all the elaboration combinators are monotone, the ``subject reduction'' inequations of Theorems~\ref{thm:subject-reduction}~and~\ref{thm:untyped-eta-reduction} extend to an inequational congruence, making inequational reasoning \emph{compositional}.
\end{remark}

\section{Correct by construction elaboration via initiality}

In \S~\ref{sec:modular-elaborators} we have shown how to specify and construct elaboration scripts in the partiality monad of a dominance satisfying certain properties. In this section, we will describe a \emph{specific} topos and dominance in which the constructions of \S~\ref{sec:modular-elaborators} may be executed so as to lead to a correct by construction elaboration algorithm. The extraction of this algorithm is carried out abstractly, using the \emph{initiality property} of a second-order generalised algebraic theory of untyped elaboration scripts.

\subsection{Implicit computational content via a base topos}

We are careful to develop everything constructively (in the sense of validity in an elementary topos with a natural numbers object), so that the extracted elaboration function can be understood to have an algorithmic content. This is a kind of \emph{implicit} effectiveness, which a classical mathematician might prefer to make explicit by carrying out the constructions below in a realisability topos~\cite{hyland:1982}. We remain agnostic on this point, and simply work relative to an arbitrary base topos $\BaseTopos$ that shall play the role of the ambient category of sets.

\begin{assumption}[The base topos]
  Everything that follows should be understood as taking place relative to a suitable elementary topos $\BaseTopos$ with a natural numbers object; hence, if we refer to a category $\mathbb{C}$, we mean an internal category in $\BaseTopos$, and if we refer to presheaves on $\mathbb{C}$ we mean \emph{internal} $\BaseTopos$-valued presheaves on $\mathbb{C}$. The interested reader can learn the details of relative topos theory from a standard reference such as those of Johnstone~\cite{johnstone:topos:1977,johnstone:2002}, but we stress that it is enough to simply embrace constructive reasoning and work informally.
\end{assumption}

\subsection{Natural models of second-order generalised algebraic theories}\label{sec:sogats}

A \emph{natural model}~\citep{awodey:2018:natural-models} of Martin-L\"of type theory with an answer type is given by a category of contexts $\mathbb{C}$ equipped with a terminal object $\mathbf{1}_{\mathbb{C}}$ together with the following structure in $\Psh{\mathbb{C}}$:\footnote{Awodey requires $\mathbb{C}$ to be small; we follow \citet{uemura:2021:thesis} in relaxing this requirement so that standard models, like the sets model, may be constructed. The category of presheaves on a large category may not be cartesian closed, but nonetheless the function spaces needed for defining natural models will be present: when the domain of a function space is representable, the function space is computed by change of context.}
\begin{enumerate}
  \item a \emph{judgemental structure} $(\Tp,\Tm)$ such that $\Tm:\Psh{\mathbb{C}}/\Tp$ is representable in the sense that each of its fibres over a representable is representable;
  \item dependent function type, dependent pair type, identity type, and answer type structure.
\end{enumerate}

In other words, a natural model is given by an \emph{elementary} model (\S~\ref{sec:elementary-model}) in a chosen category of presheaves, with the additional condition that the base category be closed under contextual operations (empty context and context extension). The natural models arrange themselves into a (2,1)-category with a bi-initial object, which \citet{uemura:2021:thesis} shows can be built from raw syntax.\footnote{Here we will use only the existence of the bi-initial model and will not depend on any specific details of its construction! Indeed, there can be many different constructions of the bi-initial model, each of which provides a different kind of insight.}

The notion of \emph{natural model} described above is a special case of a more general notion that applies to \emph{second-order generalised algebraic theories}~\citep{uemura:2021:thesis}, of which Martin-L\"of type theory is a specific example. %
We will use SOGATs twice: first to view Martin-L\"of type theory algebraically, and second to circumscribe a theory of untyped elaboration scripts from whose initiality property we shall extract an elaboration algorithm.

\subsection{A topos and dominance for effective elaboration}

Our goal is to find a topos over $\BaseTopos$ equipped with a dominance $\OProp$ and an elementary model of Martin-L\"of type theory that is \emph{elaborable} with respect to $\OProp$ in the sense of \S~\ref{sec:modular-elaborators}.

\subsubsection{Presheaves on contexts and substitutions}\label{sec:bi-initial-model}

Now let $\InitML$ be the bi-initial natural model of Martin-L\"of type theory with an answer type, so that we have the structure of an elementary model (\S~\ref{sec:elementary-model}) in $\Psh{\CX{\InitML}}$. As $\CX{\InitML}$ is the category of contexts and substitutions of our type theory, a presheaf on $\CX{\InitML}$ is a context-indexed family of sets equipped with a substitution action.

Our first goal is to construct a dominance in $\Psh{\CX{\InitML}}$ that makes the given elementary model of type theory \emph{elaborable} in the sense of \S~\ref{sec:modular-elaborators}.
If we chose the maximal dominance $\PROP$, then the ensuing interpretation of elaboration scripts would shed no light on the \emph{effective computability} of elaboration. Instead, we shall follow \citet[\S~2.3.5]{bocquet:2026} by employing a subuniverse of \emph{levelwise decidable} propositions, which may be defined for any presheaf topos as we do in \S~\ref{sec:ldec}; we must also account for strengthening, which we incorporate in \S~\ref{sec:strengthening}.

\subsubsection{The dominance of levelwise decidable propositions}\label{sec:ldec}
The subobject classifier for presheaves is defined in terms of \emph{sieves}~\cite{maclane-moerdijk:1992}: a proposition in context $\Gamma$ is given by a \emph{sieve} on $\Gamma$, \emph{i.e.}\ a set of arrows into $\Gamma$ that is closed under precomposition (\emph{i.e.}\ substitution). The idea of \citet{bocquet:2026} is to replace sieves with \emph{levelwise decidable} sieves, \emph{i.e.}\ sets of arrows closed under precomposition for which the membership problem is decidable in the base topos $\BaseTopos$.
Hofmann and Streicher's universe lifting construction~\citep{hofmann-streicher:1997} provides a more structural characterisation, as suggested by \citet[\S~5]{awodey:2024:universes}:
\begin{enumerate}
  \item The subobject classifier of $\Psh{\CX{\InitML}}$ is the Hofmann--Streicher lifting of the subobject classifier $\mathbf{1}\hookrightarrow\PROP_{\BaseTopos}$ from the base topos.
  \item The \emph{levelwise decidable} subobject classifier of $\Psh{\CX{\InitML}}$ is the Hofmann--Streicher lifting of the \emph{decidable} subobject classifier $\mathbf{1}\hookrightarrow\mathbf{2}$ from the base topos.
\end{enumerate}

\begin{proposition}[{\citet[Lemma~2.3.24]{bocquet:2026}}]\label{prop:oprop-dominance}
  The levelwise decidable subobject classifier of $\Psh{\CX{\InitML}}$ forms a dominance such that every representable natural transformation is fibrewise compact.
\end{proposition}

\ifextended
\begin{proof}
  Without depending on the admittedly challenging machinery of \citet{bocquet:2026}, these facts can be seen by simply unravelling the interpretation of dependent sums and products in the presheaf model of type theory~\citep[\S~4.2]{hofmann:1997}, in light of which the fibrewise compactness of representable natural transformations follows from Yoneda's lemma.
\end{proof}
\fi

We shall write $\PROP_{\mathit{ldec}}\subseteq\PROP$ for the dominance of levelwise decidable propositions in $\Psh{\CX{\InitML}}$.

\subsubsection{A Lawvere--Tierney modality for strengthening}\label{sec:strengthening}

Our full results on multilinearity (Theorem~\ref{thm:multilinearity-of-combinators}) require the following \emph{strengthening} property for all open propositions.

\begin{definition}[Strengthenable propositions]
  A proposition $\varphi$ is called \emph{strengthenable} when the implication $\prn{\forall x:\Tm\prn{\alpha}\mathpunct{.}\varphi}\to \varphi$ holds for each $\alpha:\Tp$. We will write $\PROP_{\mathit{str}}\subseteq \PROP$ for the class of strengthenable propositions.
\end{definition}

Not every levelwise decidable proposition will be strengthenable; we will therefore intersect the levelwise decidable and strengthenable propositions. First, we comment that the strengthenable propositions form an \emph{oracle modality} in the sense of \citet{ahman-bauer:2026}, namely the \emph{smallest} Lawvere--Tierney modality whose modal propositions are all strengthenable. As a result, the strengthenable propositions automatically form a dominance for which \emph{all} types are compact.

\subsubsection{Open propositions and elaborability in the presheaf model}

We define $\OProp \coloneq \PROP_{\mathit{ldec}}\cap \PROP_{\mathit{str}}$ to be the dominance consisting of strengthenable levelwise decidable propositions.

\begin{theorem}\label{thm:elaborability}
  With respect to the dominance $\OProp$ of strengthenable levelwise decidable propositions,
  \begin{enumerate}
    \item the judgemental structure $\prn{\Tp_{\InitML}, \Tm_{\InitML}}$ is elaborable and satisfies strengthening, and
    \item the dependent function type, dependent pair type, and identity type structures are elaborable.
  \end{enumerate}
\end{theorem}

\begin{proof}
  That the judgemental structure satisfies strengthening is immediate.
  We have already seen in Proposition~\ref{prop:oprop-dominance} that each $\Tm_{\InitML}\prn{\alpha}$ is compact with respect to levelwise decidable propositions and (automatically) with respect to strengthenable propositions.

  That $\Tp_{\InitML}$ and each $\Tm_{\InitML}\prn{\alpha}$ are discrete with respect to $\OProp$ is established in two steps:
  \begin{enumerate}
    \item Discreteness with respect to levelwise decidable propositions is precisely the (external!) decidability of judgemental equality for types and terms~\citep{coquand:2019,sterling:2021:thesis,bocquet:2026}.
    \item Discreteness with respect to strengthenable propositions amounts (by Kripke--Joyal semantics~\citep{maclane-moerdijk:1992}) to the strengthening lemma, which follows from normalisation.
  \end{enumerate}

  That the dependent function type, dependent pair type, and identity type structures are elaborable with respect to levelwise decidable propositions amounts to a combination of the injectivity of type constructors (as formulated in the presheaf topos by \citet[Lemma~3.53]{sterling:2025:grothendieck}) with the fact that it is externally decidable whether a type is an instance of a given head constructor, which similarly follows from the canonical representation of types and terms by normal forms. Elaborability with respect to strengthenable propositions amounts to the syntactic strengthening lemma.
\end{proof}

\begin{proposition}\label{prop:oprop-strict}
  The dominance $\OProp=\PROP_{\mathit{ldec}}\cap\PROP_{\mathit{str}}$ is strict, \emph{i.e.}\ $\bot\in\OProp$.
\end{proposition}

\begin{proof}
  The proposition $\bot$ (the empty sieve) is levelwise decidable. It is also strengthenable: by Kripke--Joyal semantics, $\prn{\forall x:\Tm\prn{\alpha}\mathpunct{.}\bot}$ holds at a context $\Gamma$ only if $\bot$ holds at every $\Delta\to\Gamma$ admitting an element of $\Tm\prn{\alpha}$; taking $\Delta=\Gamma.\alpha$ and the generic variable (which exists because $\Tm$ is representable), we see that the antecedent never holds, so the implication $\prn{\forall x:\Tm\prn{\alpha}\mathpunct{.}\bot}\to\bot$ holds vacuously.
\end{proof}

\subsection{A theory of elaboration scripts}

We now abstract the elaboration combinators from \S~\ref{sec:modular-elaborators} into a second-order generalised algebraic theory (SOGAT) with four basic sorts $\HypScript$, $\SynScript$, $\ChkScript$, and $\TypScript$. Because variables range only over hypotheses, we declare $\HypScript$ to be representable and all other sorts non-representable; then we introduce operations corresponding to the basic elaboration combinators that we have defined. The entire signature is presented in Figure~\ref{fig:elab-sogat}. This theory is provisional in the sense that we might reasonably add more generators and equations in the future. Our goal is not to fix for all time a definition of what an ``elaboration script'' is, but rather to delineate a reasonable signature that can be used in practice and extended in response to practical discoveries later on.

\begin{figure}
  \iblock{\small
    \row{
      $\HypScript : \SortIdent{Sort}_{\mathit{repr}}$
      \color{gray}// A representable sort of hypotheses
    }
    \row{
      $\SynScript,\ChkScript,\TypScript : \SortIdent{Sort}$
      \color{gray}// Non-representable sorts for synthesis-mode expressions, checking-mode expressions, and type expressions.
    }
  }

  \vspace{-1em}

  \begin{multicols}{2}\small
    \iblock{
      \mrow{\ElabVar : \HypScript\to\SynScript}
      \mrow{\ElabConv : \SynScript\to\ChkScript}
      \mrow{\ElabThe:\TypScript\to\ChkScript\to\SynScript}
      \mrow{\ElabPi:\TypScript\to\prn{\HypScript\to\TypScript}\to\TypScript}
      \mrow{\ElabLam : \prn{\HypScript\to\ChkScript}\to\ChkScript}
      \mrow{\ElabLam' : \TypScript\to \prn{\HypScript\to\SynScript}\to\SynScript}
      \mrow{\ElabApp : \SynScript\to\ChkScript\to\SynScript}
      \mrow{\ElabSg:\TypScript\to\prn{\HypScript\to\TypScript}\to\TypScript}
      \mrow{\ElabPair:\ChkScript\to\ChkScript\to\ChkScript}
      \mrow{\ElabFst,\ElabSnd:\SynScript\to\SynScript}
      \mrow{\ElabId:\TypScript\to\ChkScript\to\ChkScript\to\TypScript}
      \mrow{\ElabId_l : \SynScript\to\ChkScript\to\TypScript}
      \mrow{\ElabRefl:\ChkScript}
      \mrow{\ElabRefl':\SynScript\to\SynScript}
      \mrow{\ElabIdElim:\SynScript\to\prn{\HypScript\times\HypScript\times\HypScript\to\TypScript}\to\prn{\HypScript\to\ChkScript}\to\SynScript}
      \mrow{\ElabAns:\TypScript}
      \mrow{\ElabYes,\ElabNo:\ChkScript}
    }
  \end{multicols}
  \vspace{-1em}

  \caption{The second-order generalised algebraic theory of elaboration scripts.}
  \label{fig:elab-sogat}
\end{figure}

\subsection{The presheaf model of elaboration scripts}\label{sec:presheaf-model}

Write $\InitElab$ for the bi-initial natural model of elaboration scripts.
We now define a \emph{presheaf model} $\mathcal{E}$ of elaboration scripts, defining the category of contexts $\CX{\mathcal{E}}$ to be the presheaf category $\Psh{\CX{\InitML}}$ where $\InitML$ is the bi-initial model of Martin-L\"of type theory with an answer type from \S~\ref{sec:bi-initial-model}. Then the sorts are given by their interpretation in $\Psh{\CX{\InitML}}$ from \S~\ref{sec:modular-elaborators} under the Yoneda embedding $\Psh{\CX{\InitML}} = \CX{\mathcal{E}}\hookrightarrow\Psh{\CX{\mathcal{E}}}$. The generating operations are likewise interpreted by those from \S~\ref{sec:modular-elaborators} under the Yoneda embedding.

The universal property of the bi-initial natural model $\InitElab$ of elaboration scripts determines a universal interpretation homomorphism $
  \mathcal{E}\bbrk{-} \colon \InitElab\to\mathcal{E}
$
which is unique up to isomorphism with the property of taking syntactic elaboration scripts to the corresponding constructs in the  presheaf model $\mathcal{E}$. More prosaically, the interpretation function translates elaboration scripts
\begin{align*}
  x_1 : \HypScript,\ldots,x_n:\HypScript
  &\vdash
  A\prn{x_1,\ldots,x_n}:\TypScript,\;
  E\prn{x_1,\ldots,x_n}:\SynScript,\;
  M\prn{x_1,\ldots,x_n}:\ChkScript
\end{align*}
to suitable natural transformations
\[
  \widetilde{\Tm}^n
  \xrightarrow{\mathcal{E}\bbrk{A}}
  \PMC\prn{\Tp}
  \qquad
  \widetilde{\Tm}^n
  \xrightarrow{\mathcal{E}\bbrk{E}}
  \PMC\prn[\big]{\widetilde{\Tm}}
  \qquad
  \widetilde{\Tm}^n
  \xrightarrow{\mathcal{E}\bbrk{M}}
  \prn{\alpha:\Tp}\to \PMC\prn[\big]{\Tm\prn{\alpha}}
\]
in the presheaf topos $\Psh{\CX{\InitML}}$ satisfying the defining clauses of \S~\ref{sec:modular-elaborators}.

\subsection{Executing elaboration scripts by initiality}

We will now transform the interpretation functions obtained from initiality (\S~\ref{sec:presheaf-model}) into actual external elaboration procedures that can be executed in the base topos $\BaseTopos$. To that end, we first compute more explicitly what form a section of a partial element classifier takes at the level of sets.

\begin{lemma}[Sections of the partial element classifier]\label{lem:sections-of-the-pmc}
A section $u \colon \PMC\prn{X}\prn{\Gamma}$ is given by a strengthenable levelwise decidable sieve $u\IsDef$ over $\Gamma$ together with an assignment
to each $\prn{\gamma\colon\Delta\to\Gamma}\in u\IsDef$ a section
$
  u_{\gamma}\colon X\prn{\Delta}
$ such that for all $\delta\colon\Xi\to\Delta$ we have $\delta^*u_{\gamma} = u_{\gamma\circ \delta} : X\prn{\Xi}$.
\end{lemma}

\begin{exegesis}
  Here is how to interpret Lemma~\ref{lem:sections-of-the-pmc} in which we think of the section $u$ as an elaboration problem: the sieve $u\IsDef$ says \emph{not only} whether the elaboration problem succeeds in context $\Gamma$, but whether it succeeds in further contexts $\Delta$ that can be reached by a substitution $\gamma\colon\Delta\to\Gamma$. We now consider three possibilities:
  \begin{enumerate}
  \item \textbf{Success.} If the identity substitution $1_\Gamma\colon\Gamma\to\Gamma$ lies in $u\IsDef$, then the elaboration problem is not only solved but \emph{stably} solved in the sense that it will remain solved in any further context that may be reached by a substitution.
  \item \textbf{Failure.} If $u\IsDef$ is empty, then the elaboration problem is unsolvable (a type error).
  \item \textbf{Suspension.} It might be that $1_\Gamma\not\in u\IsDef$ but nonetheless $u\IsDef$ is inhabited by some substitution $\gamma\colon\Delta\to\Gamma$. This case corresponds to a \emph{suspended} elaboration problem that might later be resumed under a more specific context.
  \end{enumerate}

  For any $\gamma\colon\Delta\to\Gamma$, including the identity substitution, the proposition $\gamma\in u\IsDef$ is decidable. Therefore, we are entitled to \emph{test} whether $\gamma\in u\IsDef$ and then, if it holds, extract the element $u_\gamma : X\prn{\Delta}$.
\end{exegesis}

\begin{lemma}[Representability of environments]\label{lem:env-representable}
  The family $\Env$ is representable over $\prn{n:\mathbb{N}}\times \Tele_n$ in the sense that to each natural transformation $\Psi\colon\mathsf{yo}\prn{\Gamma}\to \Tele_n$ in $\Pr\prn{\CX{\InitML}}$ we may assign a context $\Gamma\ldots\Psi$ with the following pullback square:
  \[
    \begin{tikzcd}
      |[elbow nw]|
      \mathsf{yo}\prn{\Gamma\ldots\Psi}
        \ar[r, densely dashed]
        \ar[d, densely dashed]
      &
      \widetilde{\Env}
        \ar[d]
      \\
      \mathsf{yo}\prn{\Gamma}
        \ar[r, "\prn{n,\Psi}"']
      &
      \prn{n:\mathbb{N}}\times\Tele_n
    \end{tikzcd}
  \]
\end{lemma}

\begin{procedure}[Elaborating a type]\label{proc:elaborating-type}
  Let $\Psi\colon\Tele_n\prn{\mathbf{1}}$ be a \emph{closed} telescope. For any syntactic type elaboration script
  $
    x_1 : \HypScript,\ldots,x_n:\HypScript
    \vdash
    A\prn{x_1,\ldots,x_n}:\TypScript
  $
  we have a corresponding partial element
  $
    \brk[\big]{\Psi\vdash \mathcal{E}\bbrk{A}} : \PMC\prn{\Env_n\prn{\Psi}\to \Tp}\prn{\mathbf{1}}\text{.}
  $
  We calculate what such a global section amounts to:
  \begin{align*}
    &\PMC\prn{\Env_n\prn{\Psi}\to \Tp}\prn{\mathbf{1}}
    \\
    &\quad\cong
    \prn{\varphi:\OProp\prn{\mathbf{1}}}
    \times
    \textstyle\int_{\Gamma\in \varphi}
    \prn{\Env_n\prn{\Psi}\to\Tp}\prn{\Gamma}
    \\
    &\quad\cong
    \prn{\varphi:\OProp\prn{\mathbf{1}}}
    \times
    \textstyle\int_{\Gamma\in\varphi}
    \Tp\prn{\Gamma\ldots\Psi}
    && \text{by Lemma~\ref{lem:env-representable} via Yoneda}
  \end{align*}

  If $\varphi = \top$, then we simply obtain an element of $\Tp\prn{\mathbf{1}\ldots\Psi}$. Conversely, if $\varphi=\bot$ then we have a type error. Otherwise, the elaboration problem is ``suspended'': there may be some other context at which it can be solved later.
\end{procedure}

Similar to Procedure~\ref{proc:elaborating-type}, we may give analogous procedures for elaborating synthesis-mode and checking-mode terms.

\ifextended
\begin{procedure}[Elaborating a term in synthesis mode]
  Let $\Psi:\Tele_n\prn{\mathbf{1}}$ be a closed telescope. For any syntactic synthesis-mode term elaboration script
  $
    x_1 : \HypScript,\ldots,x_n:\HypScript
    \vdash
    E\prn{x_1,\ldots,x_n}:\SynScript
  $
  we have a corresponding partial element
  $
    \brk[\big]{\Psi\vdash \mathcal{E}\bbrk{E}} \colon \PMC\prn[\big]{\Env_n\prn{\Psi}\to \widetilde{\Tm}}\prn{\mathbf{1}}\text{.}
  $
  We calculate what such a global section amounts to:
  \begin{align*}
    &\PMC\prn[\big]{\Env_n\prn{\Psi}\to \widetilde{\Tm}}\prn{\mathbf{1}}
    \\
    &\quad\cong
    \prn{\varphi:\OProp\prn{\mathbf{1}}}
    \times
    \textstyle\int_{\Gamma\in \varphi}
    \prn[\big]{\Env_n\prn{\Psi}\to\widetilde{\Tm}}\prn{\Gamma}
    \\
    &\quad\cong
    \prn{\varphi:\OProp\prn{\mathbf{1}}}
    \times
    \textstyle\int_{\Gamma\in\varphi}
    \widetilde{\Tm}\prn{\Gamma\ldots\Psi}
    && \text{by Lemma~\ref{lem:env-representable} via Yoneda}
    \\
    &\quad\cong
    \prn{\varphi:\OProp\prn{\mathbf{1}}}
    \times
    \textstyle\int_{\Gamma\in\varphi}
    \prn{\alpha:\Tp\prn{\Gamma\ldots\Psi}}\times
    \Tm\prn{\Gamma\ldots\Psi,\alpha}
  \end{align*}
\end{procedure}

\begin{procedure}[Elaborating a term in checking mode]
  Let $\Psi:\Tele_n\prn{\mathbf{1}}$ be a closed telescope, and let $\alpha:\Tp\prn{\mathbf{1}\ldots\Psi}$ be a type in context corresponding under Lemma~\ref{lem:env-representable} via Yoneda to a closed family of types $\hat{\alpha}\colon \prn{\Env_n\prn{\Psi}\to\Tp}\prn{\mathbf{1}}$. For any syntactic checking-mode term elaboration script
  $
    x_1 : \HypScript,\ldots,x_n:\HypScript
    \vdash
    M\prn{x_1,\ldots,x_n}:\ChkScript
  $
  we have a corresponding partial element
  $
    \brk[\big]{\Psi\vdash \mathcal{E}\bbrk{M}}_{\mathbf{1}}\prn{\hat\alpha}
     : \PMC\prn[\big]{
      \prn{\psi:\Env_n\prn{\Psi}}\to
      \Tm\prn{\hat{\alpha}\prn{\psi}}
     }
  $.
  We calculate what such a global section amounts to:
  \begin{align*}
    &\PMC\prn[\big]{
      \prn{\psi:\Env_n\prn{\Psi}}\to
      \Tm\prn{\hat{\alpha}\prn{\psi}}
    }
    \\
    &\quad\cong
    \prn{\varphi:\OProp\prn{\mathbf{1}}}
    \times
    \textstyle\int_{\Gamma\in\varphi}
    \prn[\big]{
      \prn{\psi:\Env_n\prn{\Psi}}\to \Tm\prn{\hat{\alpha}\prn{\psi}}
    }\prn{\Gamma}
    \\
    &\quad\cong
    \prn{\varphi:\OProp\prn{\mathbf{1}}}
    \times
    \textstyle\int_{\Gamma\in\varphi}
    \Tm\prn{\Gamma\ldots\Psi, \alpha_{\vert \Gamma\ldots\Psi}}
    &&\text{by Lemma~\ref{lem:env-representable} via Yoneda}
  \end{align*}
\end{procedure}
\fi %
\section{Conclusions and future work}

Bidirectional elaboration is usually thought of as a moded logic programming problem; our contribution is a relatively simple \emph{functional} semantics for the executable specification of correct-by-construction elaborators. We admit that some level of sophistication is required on the categorical side in order to substantiate the interpretation, but in return we obtain a language in which the behaviour of elaborators can be investigated using na\"ive equational and inequational reasoning as we have demonstrated in our case study (\S~\ref{sec:case-study}) and in our preliminary investigations of subject reduction as an inequational law (\S~\ref{sec:subject-reduction}) and the multilinearity of elaboration combinators (\S~\ref{sec:multilinearity-of-elaboration}).

Getting here, we have left many stones unturned. Here are a few of them:

\begin{enumerate}
  \item Many type system features are left unexplored; a cumulative universe hierarchy, perhaps in the ``fuss-free'' style suggested by \citet{sterling:blog:fuss-free}, would be the immediate next step.
  \item Real proof assistants allow the user to leave \emph{holes} in the middle of unfinished proofs; these holes correspond to metavariables that are automatically abstracted over all local variables. Incorporating holes seems to require generalising from the partial map classifier to a suitable polynomial monad~\cite{gambino-kock:2013}.
  \item Our presheaf semantics accounts for stable and \emph{dynamic} elaboration automatically, in which elaboration problems may be suspended and solved in any order. We do not yet, however, account for any of the features that would necessitate dynamic elaboration, including implicit arguments and type classes. These are both areas for future work.
  \item Our axioms are simple enough to postulate in a stock proof assistant like Agda or Rocq;\footnote{Or Lean, if only its core tactics were not so tightly coupled to the axioms of classical logic which (naturally) contradict almost any synthetic approach to anything.} although we have not yet carried this out, we believe that a key intellectual contribution of our work is the possibility of giving monadic formalisations of elaboration functions to serve as \emph{modular equational specifications} for lower-level implementations, without coupling to the inessential details of important supporting algorithms (such as normalisation).
  \item Our tentative discussion of multilinearity (\S~\ref{sec:multilinearity-of-elaboration}) suggests a deeper connection between the design of elaboration calculi and the general study of evaluation order and polarity in programming languages. A more thorough investigation in light of both call-by-push-value~\cite{levy:2003:book} and the enriched effect calculus~\cite{egger-moegelberg-simpson:2014} is expected to bear fruit.
\end{enumerate}

Finally, we stress that although we have studied a particular version of Martin-L\"of type theory in this paper, our method is applicable to any suitable (\emph{i.e.}\ second-order generalised algebraic) theory and can be tailored to the metatheorems that actually hold of that theory. 
\begin{acks}
  We are thankful to Steve Awodey for useful discussions; we are especially grateful to an anonymous reader of an earlier version of this manuscript who uncovered an error in our formulation of $\eta$-conversion, which we have since corrected.
  This work was funded by the United States Air Force Office of Scientific Research under grant FA9550-23-1-0728 (\emph{New Spaces for Denotational Semantics}; Dr Tristan Nguyen, Program Manager). Views and opinions expressed are however those of the authors only and do not necessarily reflect those of AFOSR.
\end{acks}

\bibliography{refs}

@misc{mcbride:ni,
  title = {The types who say `ni'},
  author = {Conor Titania {Mc~Bride}},
  year = 2019,
  note = {Unpublished manuscript},
}

@misc{atkey:2015,
	title = {An Algebraic Approach to Typechecking and Elaboration},
	author = {Atkey, Robert},
	year = {2015},
	notes = {Talk given at the Scottish Programming Languages Seminar},
	url = {https://bentnib.org/posts/2015-04-19-algebraic-approach-typechecking-and-elaboration.html},
}

@article{coquand:1996,
  author = {Coquand, Thierry},
  year = {1996},
  doi = {10.1016/0167-6423(95)00021-6},
  issn = {0167-6423},
  journal = {Science of Computer Programming},
  number = {1},
  pages = {167--177},
  title = {An algorithm for type-checking dependent types},
  volume = {26},
}

@article{pierce-turner:2000,
  author = {Pierce, Benjamin C. and Turner, David N.},
  year = {2000},
  journal = toplas,
  number = {1},
  pages = {1--44},
  title = {Local type inference},
  volume = {22},
}

@phdthesis{dagand:2013,
	author = {Dagand, Pierre-\'{E}variste},
	address = {Glasgow, Scotland},
	school = {University of Strathclyde},
	year = {2013},
	month = aug,
	title = {A Cosmology of Datatypes: Reusability and Dependent Types},
}

@phdthesis{uemura:2021:thesis,
  author = {Uemura, Taichi},
  address = {Amsterdam},
  publisher = {Institute for Logic, Language and Computation},
  school = {Universiteit van Amsterdam},
  url = {https://www.illc.uva.nl/cms/Research/Publications/Dissertations/DS-2021-09.text.pdf},
  year = {2021},
  title = {Abstract and Concrete Type Theories},
}

@incollection{dybjer:1996,
  author = {Dybjer, Peter},
  editor = {Berardi, Stefano and Coppo, Mario},
  address = {Berlin, Heidelberg},
  publisher = {Springer Berlin Heidelberg},
  booktitle = {Types for Proofs and Programs: International Workshop, TYPES '95 Torino, Italy, June 5--8, 1995 Selected Papers},
  year = {1996},
  isbn = {978-3-540-70722-6},
  pages = {120--134},
  title = {Internal type theory},
}

@article{awodey:2018:natural-models,
  author = {Awodey, Steve},
  publisher = {Cambridge University Press},
  year = {2018},
  doi = {10.1017/S0960129516000268},
  eprint = {1406.3219},
  eprintclass = {math.CT},
  eprinttype = {arXiv},
  journal = {Mathematical Structures in Computer Science},
  number = {2},
  pages = {241--286},
  title = {Natural models of homotopy type theory},
  volume = {28},
}

@article{moggi:1991,
  author = {Moggi, Eugenio},
  year = {1991},
  doi = {10.1016/0890-5401(91)90052-4},
  issn = {0890-5401},
  journal = i&c,
  note = {Selections from 1989 IEEE Symposium on Logic in Computer Science},
  number = {1},
  pages = {55--92},
  title = {Notions of computation and monads},
  volume = {93},
}

@phdthesis{rosolini:1986,
  author = {Rosolini, Guiseppe},
  school = {University of Oxford},
  year = {1986},
  title = {Continuity and effectiveness in topoi},
}

@inproceedings{hyland:1982,
  author = {Hyland, J. M. E.},
  editor = {Troelstra, A. S. and Dalen, D. Van},
  publisher = {North Holland Publishing Company},
  booktitle = {The {L.E.J. Brouwer} Centenary Symposium},
  year = {1982},
  pages = {165--216},
  title = {The effective topos},
}

@article{abbott-altenkirch-ghani:2005,
  author = {Abbott, Michael and Altenkirch, Thorsten and Ghani, Neil},
  year = {2005},
  doi = {10.1016/j.tcs.2005.06.002},
  issn = {0304-3975},
  journal = {Theoretical Computer Science},
  note = {Applied Semantics: Selected Topics},
  number = {1},
  pages = {3--27},
  title = {Containers: Constructing strictly positive types},
  volume = {342},
}

@misc{awodey:2025,
      title={Algebraic Type Theory, Part 1: Martin-L\"of algebras},
      author={Steve Awodey},
      year={2025},
      eprint={2505.10761},
      archivePrefix={arXiv},
      primaryClass={math.CT},
      url={https://arxiv.org/abs/2505.10761},
}

@unpublished{gratzer-sterling:2020,
  author = {Gratzer, Daniel and Sterling, Jonathan},
  year = {2020},
  eprint = {2012.10783},
  eprintclass = {cs.LO},
  eprinttype = {arXiv},
  note = {Unpublished manuscript},
  title = {Syntactic categories for dependent type theory: sketching and adequacy},
}

@misc{bocquet:2026,
  author =	{Bocquet, Rafa\"{e}l},
  title =	{{Relative induction principles for second-order generalized algebraic theories}},
  year =	{2026},
  howpublished = {PhD thesis},
  url = {https://rafaelbocquet.gitlab.io/pdfs/thesis.pdf}
}

@article{altenkirch-chapman-uustalu:2015,
  author = {Altenkirch, Thorsten and Chapman, James and Uustalu, Tarmo},
  year = {2015},
  doi = {10.2168/LMCS-11(1:3)2015},
  journal = {Logical Methods in Computer Science},
  number = {1},
  volume = {11},
  title = {Monads need not be endofunctors},
}

@article{slattery:2023,
  author = {Slattery, Andrew},
  year = {2023},
  eprint = {2304.14788},
  eprintclass = {math.CT},
  eprinttype = {arXiv},
  month = dec,
  day = {11},
  journal = {Theory and Applications of Categories},
  volume = {39},
  pages = {1018--1049},
  title = {Pseudocommutativity and Lax Idempotency for Relative Pseudomonads},
}

@phdthesis{slattery:2024:thesis,
  author = {Slattery, Andrew},
  school = {University of Leeds},
  year = {2024},
  url = {https://etheses.whiterose.ac.uk/id/eprint/35372/},
  title = {Commutativity of Relative Pseudomonads},
}

@book{maclane-moerdijk:1992,
  author = {Mac Lane, Saunders and Moerdijk, Ieke},
  address = {New York},
  publisher = {Springer},
  year = {1992},
  isbn = {0-387-97710-4},
  series = {Universitext},
  title = {Sheaves in geometry and logic: a first introduction to topos theory},
}

@article{awodey:2024:universes,
  author = {Awodey, Steve},
  year = {2024},
  doi = {10.1017/S0960129524000203},
  journal = {Mathematical Structures in Computer Science},
  number = {9},
  pages = {1--17},
  title = {On {Hofmann--Streicher} universes},
  volume = {34},
}

@unpublished{hofmann-streicher:1997,
  author = {Hofmann, Martin and Streicher, Thomas},
  url = {https://www2.mathematik.tu-darmstadt.de/~streicher/NOTES/lift.pdf},
  year = {1997},
  note = {Unpublished note},
  title = {Lifting {G}rothendieck Universes},
}

@incollection{sterling:2025:grothendieck,
author="Sterling, Jonathan",
editor="Panza, Marco
and Struppa, Daniele C.
and Szczeciniarz, Jean-Jacques",
title="Toward a Geometry for Syntax",
bookTitle="The Mathematical and Philosophical Legacy of Alexander Grothendieck",
year="2025",
publisher="Springer Nature Switzerland",
address="Cham",
pages="391--432",
isbn="978-3-031-68934-5",
doi="10.1007/978-3-031-68934-5_15",
url="https://doi.org/10.1007/978-3-031-68934-5_15"
}

@article{coquand:2019,
  author = {Coquand, Thierry},
  year = {2019},
  doi = {10.1016/j.tcs.2019.01.015},
  eprint = {1810.09367},
  eprintclass = {cs.PL},
  eprinttype = {arXiv},
  issn = {0304-3975},
  journal = tcs,
  note = {In memory of Maurice Nivat, a founding father of Theoretical Computer Science - Part I},
  pages = {184--191},
  title = {Canonicity and normalization for dependent type theory},
  volume = {777},
}

@phdthesis{sterling:2021:thesis,
  author = {Sterling, Jonathan},
  school = {Carnegie Mellon University},
  year = {2021},
  doi = {10.5281/zenodo.6990769},
  note = {Version 1.1, revised May 2022},
  number = {CMU-CS-21-142},
  title = {First Steps in Synthetic {Tait} Computability: The Objective Metatheory of Cubical Type Theory},
}

@inbook{hofmann:1997, place={Cambridge}, series={Publications of the Newton Institute}, title={Syntax and Semantics of Dependent Types}, booktitle={Semantics and Logics of Computation}, publisher={Cambridge University Press}, author={Hofmann, Martin}, editor={Pitts, Andrew M. and Dybjer, P.Editors}, year={1997}, pages={79–130}, collection={Publications of the Newton Institute}}

@article{kock:1972,
  author = {Kock, Anders},
  date = {1972-12-01},
  doi = {10.1007/BF01304852},
  id = {Kock1972},
  isbn = {1420-8938},
  journal = {Archiv der Mathematik},
  number = {1},
  pages = {113--120},
  title = {Strong functors and monoidal monads},
  volume = {23},
  year = {1972},
  }

@book{johnstone:topos:1977,
    author = {Johnstone, Peter T.},
    publisher = {Academic Press},
    year = {1977},
    title = {Topos Theory},
  }

@book{johnstone:2002,
    author = {Johnstone, Peter T.},
    publisher = {Oxford Science Publications},
    year = {2002},
    number = {43},
    series = {Oxford Logical Guides},
    title = {Sketches of an Elephant: A Topos Theory Compendium: Volumes 1 and 2},
  }

@book{levy:2003:book,
    author = {Levy, Paul Blain},
    publisher = {Kluwer, Semantic Structures in Computation, 2},
    year = {2003},
    month = jan,
    isbn = {1-4020-1730-8},
    title = {Call-by-Push-Value: A Functional/Imperative Synthesis},
  }

@article{egger-moegelberg-simpson:2014,
      author = {Egger, Jeff and Møgelberg, Rasmus Ejlers and Simpson, Alex},
      title = {The enriched effect calculus: syntax and semantics},
      journal = {Journal of Logic and Computation},
      volume = {24},
      number = {3},
      pages = {615-654},
      year = {2014},
      month = {06},
      issn = {0955-792X},
      doi = {10.1093/logcom/exs025},
  }

@article{gambino-kock:2013,
    author = {Gambino, Nicola and Kock, Joachim},
    publisher = {Cambridge University Press},
    year = {2013},
    doi = {10.1017/S0305004112000394},
    journal = {Mathematical Proceedings of the Cambridge Philosophical Society},
    number = {1},
    pages = {153--192},
    title = {Polynomial functors and polynomial monads},
    volume = {154},
  }

@InProceedings{ahman-bauer:2026,
  author =	{Ahman, Danel and Bauer, Andrej},
  title =	{{Sheaves as Oracle Computations}},
  booktitle =	{11th International Conference on Formal Structures for Computation and Deduction (FSCD 2026)},
  pages =	{1:1--1:17},
  series =	{Leibniz International Proceedings in Informatics (LIPIcs)},
  ISBN =	{978-3-95977-433-8},
  ISSN =	{1868-8969},
  year =	{2026},
  volume =	{378},
  editor =	{Pfenning, Frank},
  publisher =	{Schloss Dagstuhl -- Leibniz-Zentrum f{\"u}r Informatik},
  address =	{Dagstuhl, Germany},
  URL =		{https://drops.dagstuhl.de/entities/document/10.4230/LIPIcs.FSCD.2026.1},
  doi =		{10.4230/LIPIcs.FSCD.2026.1},
}

@article{taylor:2000,
author = {Taylor, Paul},
journal = {Theory and Applications of Categories},
pages = {284-338},
title = {Geometric and higher order logic in terms of abstract Stone duality.},
volume = {7},
number = {15},
year = {2000},
}

@online{lean4-reference-typesystem,
  title        = {The Lean Language Reference: The Type System},
  author = {Lean FRO},
  year         = {2025},
  url          = {https://lean-lang.org/doc/reference/latest/The-Type-System/},
}

@online{sterling:blog:fuss-free,
  title={Fuss-free universe hierarchies},
  author={Sterling, Jonathan},
  year = {2025},
  month = oct,
  day = {6},
  url = {https://www.jonmsterling.com/01HX/},
  note = {Blog post}
}

@online{kovacs2ltt-principles,
  author  = {Andr{\'a}s Kov{\'a}cs},
  title   = {PRINCIPLES.md},
  year    = {2024},
  url     = {https://github.com/AndrasKovacs/2ltt-impl/blob/f4e13a73c4fb14bbef7f7451974a52ea3f8d6d62/src/PRINCIPLES.md},
  urldate = {2026-07-05},
  note    = {GitHub repository \texttt{AndrasKovacs/2ltt-impl}, commit \texttt{f4e13a73c4fb14bbef7f7451974a52ea3f8d6d62}}
}

@article{hill:1996, title={Combinators for parsing expressions}, volume={6}, DOI={10.1017/S0956796800001799}, number={3}, journal={Journal of Functional Programming}, author={Hill, Steve}, year={1996}, pages={445–464}}

@article{hutton-meijer:1998, title={Monadic parsing in Haskell}, volume={8}, DOI={10.1017/S0956796898003050}, number={4}, journal={Journal of Functional Programming}, author={Hutton, Graham and Meijer, Erik}, year={1998}, pages={437–444}}

@inproceedings{fiore:2002,
  author = {Fiore, Marcelo P.},
  address = {Pittsburgh, PA, USA},
  publisher = {Association for Computing Machinery},
  booktitle = {Proceedings of the 4th ACM SIGPLAN International Conference on Principles and Practice of Declarative Programming},
  year = {2002},
  doi = {10.1145/571157.571161},
  isbn = {1-58113-528-9},
  pages = {26--37},
  series = {PPDP '02},
  title = {Semantic Analysis of Normalisation by Evaluation for Typed Lambda Calculus},
}

@inproceedings{altenkirch-hofmann-streicher:1995,
  author = {Altenkirch, Thorsten and Hofmann, Martin and Streicher, Thomas},
  editor = {Pitt, David and Rydeheard, David E. and Johnstone, Peter},
  address = {Berlin, Heidelberg},
  publisher = {Springer Berlin Heidelberg},
  booktitle = {Category Theory and Computer Science},
  year = {1995},
  doi = {10.1007/3-540-60164-3_27},
  isbn = {978-3-540-44661-3},
  pages = {182--199},
  title = {Categorical reconstruction of a reduction free normalization proof},
}

@inbook{nordstrom-petersson-smith:2001, author = {Nordstr\"{o}m, B. and Petersson, K. and Smith, J. M.}, title = {Martin-L\"{o}f's type theory}, year = {2001}, isbn = {0198537816}, publisher = {Oxford University Press, Inc.}, address = {USA}, booktitle = {Handbook of Logic in Computer Science: Volume 5: Logic and Algebraic Methods}, pages = {1–32}, numpages = {32} }

@misc{martin-lof:1986,
  author = {Martin-L\"{o}f, Per},
  year = {1986},
  note = {Notes from a lecture given in G\"{o}teborg},
  title = {{Amendment to Intuitionistic Type Theory}},
}

@article{dunfield-krishnaswami:2021,
author = {Dunfield, Jana and Krishnaswami, Neel},
title = {Bidirectional Typing},
year = {2021},
issue_date = {June 2022},
publisher = {Association for Computing Machinery},
address = {New York, NY, USA},
volume = {54},
number = {5},
issn = {0360-0300},
url = {https://doi.org/10.1145/3450952},
doi = {10.1145/3450952},
journal = {ACM Comput. Surv.},
month = may,
articleno = {98},
numpages = {38}
}

@misc{mihejevs-hedges:2025,
      title={Canonical bidirectional typechecking},
      author={Zanzi Mihejevs and Jules Hedges},
      year={2025},
      eprint={2512.07511},
      archivePrefix={arXiv},
      primaryClass={cs.PL},
      url={https://arxiv.org/abs/2512.07511},
}

@article{fiore-et-al:2016,
  author  = {Fiore, Marcelo and Gambino, Nicola and Hyland, Martin and Winskel, Glynn},
  title   = {Relative pseudomonads, Kleisli bicategories, and substitution monoidal structures},
  journal = {Selecta Mathematica - New Series, 24 (3), 2018, pp. 2791-2830},
  year    = {2016},
  doi     = {10.1007/s00029-017-0361-3}
}

@misc{arkor-saville-slattery:2025,
      title={Bicategories of algebras for relative pseudomonads},
      author={Arkor, Nathanael and Saville, Philip and Slattery, Andrew},
      year={2025},
      eprint={2501.12510},
      archivePrefix={arXiv},
      primaryClass={math.CT},
      url={https://arxiv.org/abs/2501.12510},
}

@article{felicissimo:2025,
author = {Felicissimo, Thiago},
title = {Generic Bidirectional Typing for Dependent Type Theories},
year = {2025},
issue_date = {March 2025},
publisher = {Association for Computing Machinery},
address = {New York, NY, USA},
volume = {47},
number = {1},
issn = {0164-0925},
url = {https://doi.org/10.1145/3715095},
doi = {10.1145/3715095},
journal = {ACM Trans. Program. Lang. Syst.},
month = apr,
articleno = {2},
numpages = {42}
}

@article{johnson:1975:yacc,
author="Johnson, S. C.",
title="YACC: Yet another compiler-compiler",
journal="Computing Science Technical Report",
year="1975",
volume="32",
URL="https://cir.nii.ac.jp/crid/1574231874191646464"
}

@article{pottier:2006,
title = {Towards Efficient, Typed LR Parsers},
journal = {Electronic Notes in Theoretical Computer Science},
volume = {148},
number = {2},
pages = {155-180},
year = {2006},
note = {Proceedings of the ACM-SIGPLAN Workshop on ML (ML 2005)},
issn = {1571-0661},
doi = {https://doi.org/10.1016/j.entcs.2005.11.044},
url = {https://www.sciencedirect.com/science/article/pii/S1571066106001307},
author = {François Pottier and Yann Régis-Gianas},
}

@techreport{leijen-meijer:2001,
author = {Leijen, Daan and Meijer, Erik},
title = {Parsec: Direct Style Monadic Parser Combinators for the Real World},
year = {2001},
month = {July},
url = {https://www.microsoft.com/en-us/research/publication/parsec-direct-style-monadic-parser-combinators-for-the-real-world/},
number = {UU-CS-2001-27},
note = {User Modeling 2007, 11th International Conference, UM 2007, Corfu, Greece, June 25-29, 2007},
}

@inproceedings{synthlean:cpp,
author = {Nawrocki, Wojciech and Hua, Joseph and Carneiro, Mario and Xu, Yiming and Woolfson, Spencer and Rong, Shuge and Hazratpour, Sina and Awodey, Steve},
title = {A Certifying Proof Assistant for Synthetic Mathematics in Lean},
year = {2026},
isbn = {9798400723414},
publisher = {Association for Computing Machinery},
address = {New York, NY, USA},
doi = {10.1145/3779031.3779087},
booktitle = {Proceedings of the 15th ACM SIGPLAN International Conference on Certified Programs and Proofs},
pages = {88–103},
numpages = {16},
location = {Rennes, France},
series = {CPP '26}
}

@article{bucalo-fuhrmann-simpson:2003,
title = "An equational notion of lifting monad",
author = "Anna Bucalo and Carsten Fuhrmann and Alex Simpson",
year = "2003",
doi = "10.1016/S0304-3975(01)00243-2",
language = "English",
volume = "294",
pages = "31 -- 60",
journal = "Theoretical Computer Science",
issn = "0304-3975",
publisher = "Elsevier",
number = "1--2",
}

@inproceedings{allais-et-al:2022:typos,
  author    = {Allais, Guillaume and Altenmüller, Malin and McBride, Conor
               and Nakov, Georgi and Nordvall Forsberg, Fredrik
               and Roy, Craig},
  title     = {{TypOS}: An Operating System for Typechecking Actors},
  booktitle = {28th International Conference on Types for Proofs and Programs (TYPES 2022)},
  year      = {2022},
  address   = {Nantes, France},
  month     = jun,
  url       = {https://types22.inria.fr/files/2022/06/TYPES_2022_paper_31.pdf}
}

@inproceedings{oconnor:2019,
  title = {{Deferring the details and deriving programs}},
  author = {Liam O'Connor},
  booktitle = {ICFP '19: ACM SIGPLAN International Conference on Functional Programming},
  year = 2019,
  doi = {10.1145/3331554.3342605},
}

@misc{awodey-newstead:2018,
  author = {Awodey, Steve and Newstead, Clive},
  year = {2018},
  eprint = {1802.00997},
  eprintclass = {math.CT},
  eprinttype = {arXiv},
  title = {Polynomial pseudomonads and dependent type theory},
}

@article{aberle-spivak:2025,
    title      = {Polynomial Universes in Homotopy Type Theory},
    author     = {C. B. Aberlé and David I. Spivak},
    url        = {https://entics.episciences.org/16885},
    doi        = {10.46298/entics.16885},
    journal    = {Electronic Notes in Theoretical Informatics and Computer Science},
    issn       = {2969-2431},
    volume     = {Volume 5 - Proceedings of MFPS XLI},
    eid        = 1,
    year       = {2025},
    month      = {Dec},
}
\bibliographystyle{ACM-Reference-Format}

\clearpage
\appendix

\section{The partiality relative 2-monad}

\begin{definition}[Lax multimorphism]\label{def:lax-multimap}
  Let $K$ be compact, $\prn{A_k}_{k:K}$ a family of algebras, and $B$ an algebra. A monotone map $f\colon\prod_{k:K}A_k\to B$ is a \emph{lax multimorphism} when, for every family of maps $a_k\colon X_k\to A_k$ and every $u:\prod_k\PMC\prn{X_k}$,
  \[
  x \gets \prn{\Scope~k:K~\In~u_k};\ f\prn{\lambda k\mathpunct{.}a_k\prn{x_k}}
  \preccurlyeq_B
  f\prn{\lambda k\mathpunct{.}\ x\gets u_k;\ a_k\prn{x}}\text{.}
  \]
  It is \emph{multilinear} in the sense of Definition~\ref{def:multilinear} when this inequality is an equality.
\end{definition}

Intuitively, whenever the bundled evaluation is defined, so is the pointwise one, and the two agree; equivalently, the bundled side has the smaller support. %
It transpires in the case of the partiality monad that the lax comparison holds with no hypothesis on the codomain.

\begin{lemma}[Least extension]\label{lem:least-extension}
  Let $A$ be an algebra, $f\colon X\to A$, $u:\PMC\prn{X}$, and $a:A$. Then
  \[
  f^\dagger\prn{u}\preccurlyeq a
  \Longleftrightarrow
  \forall p:u\IsDef\mathpunct{.}
  f\prn{u_p}\preccurlyeq a\text{.}
  \]
  That is, $f^\dagger\prn{u}$ is the least element of $A$ that dominates $f$ at the generic point of $u$.
\end{lemma}
\begin{proof}
  $\prn{\Rightarrow}$ Under $p:u\IsDef$ we have $u=\eta_X\prn{u_p}$, so $f\prn{u_p}=f^\dagger\prn{u}\preccurlyeq a$ by the unit law.
  $\prn{\Leftarrow}$ Put $\varphi\coloneq u\IsDef$ and $s\coloneq\lambda p\mathpunct{.}u_p\colon\varphi\to X$, so that $u=\PMC\prn{s}\prn{w}$ for the generic element $w\coloneq\prn{\varphi,\mathrm{id}_\varphi}$, and let $\bar a\colon\mathbf 1\to A$ pick out $a$. The hypothesis is exactly $f\circ s\preccurlyeq\bar a\circ{!}$, so naturality and local monotonicity of $\prn{-}^\dagger$ give
  \[
  f^\dagger\prn{u}=\prn{f\circ s}^\dagger\prn{w}\preccurlyeq\prn{\bar a\circ{!}}^\dagger\prn{w}=\bar a^\dagger\prn{\varphi,\star}\preccurlyeq\bar a^\dagger\prn[\big]{\eta_{\mathbf 1}\prn{\star}}=a\text{,}
  \]
  using $\PMC\prn{!}\prn{w}=\prn{\varphi,\star}\preccurlyeq\eta_{\mathbf 1}\prn{\star}$ for the last two steps and $\bar a^\dagger\circ\eta_{\mathbf 1}=\bar a$ at the end.
\end{proof}

\begin{proposition}[Automatic lax multilinearity]\label{prop:auto-lax}
  Let $K$ be compact, $\prn{A_k}_{k:K}$ a family of algebras, and $B$ \emph{any} algebra. Then every monotone map $f\colon\prod_{k:K}A_k\to B$ is a lax multimorphism.
\end{proposition}
\begin{proof}
  Fix maps $a_k\colon X_k\to A_k$ and $u:\prod_k\PMC\prn{X_k}$, and write $g\coloneq f\circ\prod_k a_k$ and $v\coloneq\Scope~k:K~\In~u_k$, so the two sides of the lax inequality are $g^\dagger\prn{v}$ and $b\coloneq f\prn{\lambda k\mathpunct{.}a_k^\dagger\prn{u_k}}$. By Lemma~\ref{lem:least-extension} it suffices to show $g\prn{v_p}\preccurlyeq b$ under $p:v\IsDef$. Now $v\IsDef=\forall k\mathpunct{.}u_k\IsDef$, so $p$ yields witnesses $p_k:u_k\IsDef$ with $v_p=\lambda k\mathpunct{.}\prn{u_k}_{p_k}$; then $u_k=\eta\prn{\prn{u_k}_{p_k}}$ gives $a_k^\dagger\prn{u_k}=a_k\prn{\prn{u_k}_{p_k}}$ by the unit law, so both $b$ and $g\prn{v_p}$ equal $f\prn{\lambda k\mathpunct{.}a_k\prn{\prn{u_k}_{p_k}}}$.
\end{proof}

\begin{corollary}
  Every monotone map $f\colon\prod_k A_k\to B$ carries a canonical lax multimorphism structure (Proposition~\ref{prop:auto-lax}), and $f$ is multilinear if and only if this canonical comparison is an equality. Checking multilinearity thus reduces to checking that a single inequality is an equality; as we shall see in \S~\ref{sec:multilinearity-of-elaboration}, this can fail for the binder-forming combinators unless the judgemental structure satisfies strengthening.
\end{corollary}

\begin{remark}[Strengthening is needed for $\underline{\ElabPi}$]\label{rem:strengthening-needed}
  The inequality of Definition~\ref{def:lax-multimap} can be strict for $\ElabPi$ when strengthening fails. Take $\Tp\coloneq\brc{\alpha}$ and $\Tm\prn{\alpha}\coloneq\emptyset$, with $\TpPi\prn{\alpha,!}\coloneq\alpha$ (all term-level structure being vacuous), and take $\OProp\coloneq\PROP$, which is strict. Then $\Tp$ is discrete, $\Tm\prn{\alpha}$ is discrete and compact, and $\TpPi$ is injective, but strengthening fails for $\varphi=\bot$. Let $A\coloneq\Return~\alpha:\TypScript$ and $B\colon\mathbf{1}\to\HypScript\pitchfork\TypScript$ the constant map at $\lambda h\mathpunct{.}\Return~\alpha$, and let $u\coloneq\mathbf{undefined}:\PMC\prn{\mathbf{1}}$. Then
  \[
    \ElabPi~A~\prn{x\gets u;\ B\prn{x}}
    =
    \prn{\Var{fam}\gets\Scope~a:\emptyset~\In~\mathbf{undefined};\ \Return~\TpPi\prn{\alpha,\Var{fam}}}
    =
    \Return~\alpha\text{,}
  \]
  because the support $\forall a:\emptyset\mathpunct{.}\bot$ is true, whereas $x\gets u;\ \ElabPi~A~\prn{B\prn{x}}=\mathbf{undefined}$. So the bundled side is strictly below the pointwise side, and $\underline{\ElabPi}$ is lax but not multilinear.
\end{remark}

\begin{remark}[Lax-idempotency]\label{rem:lax-idempotent}
  Lemma~\ref{lem:least-extension} is the concrete face of the \emph{lax-idempotency} of $\PMC$~\citep{slattery:2023,slattery:2024:thesis}: extensions are left Kan extensions along the unit, so that every $1$-cell between algebras acquires a canonical lax structure. Proposition~\ref{prop:auto-lax} is the multi-ary instance of the same phenomenon. %
\end{remark}

\begin{remark}[The lax multicategory]\label{rem:lax-multicat}
  The lax multimorphisms of Definition~\ref{def:lax-multimap} likewise organise into a $\POSET$-enriched multicategory $\underline{\mathcal{L}}_{\mathrm{lax}}$, of which $\underline{\mathcal{L}}$ is the wide sub-multicategory of strict multimaps. By Proposition~\ref{prop:auto-lax} the multimaps of $\underline{\mathcal{L}}_{\mathrm{lax}}$ are simply \emph{all} monotone maps between products of algebras, so their identities and composition are inherited.
\end{remark}

\begin{lemma}[Multicomposition]\label{lem:composition}
  Let $f\colon\prod_{k:K}B_k\to C$ be multilinear and, for each $k:K$, let $g_k\colon\prod_{j:J_k}A_{k,j}\to B_k$ be multilinear. Then $h\coloneq f\circ\prod_{k}g_k\colon\prod_{\prn{k,j}:\widetilde{J}}A_{k,j}\to C$ is multilinear, where $\widetilde{J}=\prn{k:K}\times J_k$.
\end{lemma}
\begin{proof}
  Let $a_{k,j}\colon X_{k,j}\to A_{k,j}$ and put $b_k\coloneq g_k\circ\prod_j a_{k,j}$. Using functoriality of products, multilinearity of each $g_k$, multilinearity of $f$, and finally Lemma~\ref{lem:theta-comp}:
  \begin{align*}
    h\circ\textstyle\prod_{k,j}a_{k,j}^\dagger
    &=f\circ\textstyle\prod_k\prn[\big]{g_k\circ\prod_j a_{k,j}^\dagger}
    \\
    &=f\circ\textstyle\prod_k\prn[\big]{b_k^\dagger\circ\boldsymbol{\vartheta}_{J_k,X_k}}\\
    &=f\circ\prn[\big]{\textstyle\prod_k b_k^\dagger}\circ\prn[\big]{\textstyle\prod_k\boldsymbol{\vartheta}_{J_k,X_k}}\\
    &=\prn[\big]{f\circ\textstyle\prod_k b_k}^\dagger\circ\boldsymbol{\vartheta}_{K,\prn{\prod_j X}}\circ\prn[\big]{\textstyle\prod_k\boldsymbol{\vartheta}_{J_k,X_k}}\\
    &=\prn[\big]{h\circ\textstyle\prod_{k,j}a_{k,j}}^\dagger\circ\boldsymbol{\vartheta}_{\widetilde{J},X}\text{.}\qedhere
  \end{align*}
\end{proof}

\begin{theorem}[Multimaps out of free algebras]\label{thm:engine}
    Let $\prn{Z_k}_{k:K}$ be types, $K$ compact, and $B$ an algebra. The assignment $g\mapsto g^\dagger \circ\boldsymbol{\vartheta}_{K,Z}$ is a bijection from maps $g\colon\prod_k Z_k\to B$ to multilinear maps $\prod_k\PMC\prn{Z_k}\to B$, with inverse $f\mapsto f\circ\prod_k\eta_{Z_k}$.
  \end{theorem}
\begin{proof}
    For the first claim, let $u_k\colon X_k\to\PMC\prn{Z_k}$; the free extension operator is $\prn{-}^{\dagger}$, so we must show $g^\dagger\circ\boldsymbol{\vartheta}_{K,Z}\circ\prod_k u_k^{\dagger}=\prn[\big]{g^\dagger\circ\boldsymbol{\vartheta}_{K,Z}\circ\prod_k u_k}^\dagger\circ\boldsymbol{\vartheta}_{K,X}$. By the algebra-associativity law the right-hand side equals $g^\dagger\circ\prn[\big]{\boldsymbol{\vartheta}_{K,Z}\circ\prod_k u_k}^{\dagger}\circ\boldsymbol{\vartheta}_{K,X}$, which coincides with the left-hand side by commutativity (Lemma~\ref{lem:scope-extension}). For the round-trips: $\prn{g^\dagger\circ\boldsymbol{\vartheta}_{K,Z}}\circ\prod_k\eta_{Z_k}=g^\dagger\circ\eta_{\prod Z}=g$ by Lemma~\ref{lem:theta-unit} and the unit law; conversely, evaluating the multilinearity condition for a multilinear $f$ at the maps $u_k\coloneq\eta_{Z_k}$, where $\eta_{Z_k}^{\dagger}=\mathrm{id}$, gives $f=\prn[\big]{f\circ\prod_k\eta_{Z_k}}^\dagger\circ\boldsymbol{\vartheta}_{K,Z}$.
  \end{proof}

\begin{corollary}\label{cor:criterion}
    A map $f\colon\prod_k\PMC\prn{Z_k}\to B$ out of free algebras is multilinear if and only if $f=\prn[\big]{f\circ\prod_k\eta_{Z_k}}^\dagger\circ\boldsymbol{\vartheta}_{K,Z}$; that is, if and only if $f$ is the canonical multilinear extension of its restriction to total inputs.
\end{corollary}

\section{Proof of Theorem~\ref{thm:multilinearity-of-combinators}, multilinearity of elaboration combinators}

For a nullary combinator $\mathbf{1}\to A$, the multilinearity condition reduces to the unit law of the target algebra. Hence $\underline{\ElabVar}$, $\underline{\ElabRefl}$, $\underline{\ElabAns}$, $\underline{\ElabYes}$, and $\underline{\ElabNo}$ are multilinear.

\begin{proof}[Case (multilinearity of $\underline{\ElabThe}$)]
  We must verify the multilinearity condition from Definition~\ref{def:multilinear} for \emph{finite} arity $K=\mathbf{2}$. To that end, we fix the following maps and partial inputs
  \[
    A\colon X_A\to \TypScript,
    M\colon X_M\to \ChkScript,
    u_A:\PMC\prn{X_A},
    u_M:\PMC\prn{X_M}
  \]
  to check that
  \[
    \ElabThe~
      \prn{x_A\gets u_A; A\prn{x_A}}~
      \prn{x_M\gets u_M; M\prn{x_M}}
    =
    x_A\gets u_A;
    x_M\gets u_M;
    \ElabThe~A\prn{x_A}~M\prn{x_M}
    \text{.}
  \]

    We derive the above by calculation from left to right:

  \iblock{
    \mhang{
      \ElabThe~
      \prn{x_A\gets u_A; A\prn{x_A}}~
      \prn{x_M\gets u_M; M\prn{x_M}}
    }{
      \commentrow{// by definition of\, $\ElabThe$}
      \mrow{
        =
        \alpha\gets \prn{x_A\gets u_A; A\prn{x_A}};
        s \gets \prn{x_M\gets u_M;M\prn{x_M}}\prn{\alpha};
        \Return~\prn{\alpha,s}
      }
      \commentrow{// by definition of\, $\ChkScript$ as a product of algebras}
      \mrow{
        =
        \alpha\gets \prn{x_A\gets u_A; A\prn{x_A}};
        s \gets \prn{x_M\gets u_M; M\prn{x_M}\prn{\alpha}};
        \Return~\prn{\alpha,s}
      }
      \commentrow{// by associativity and commutativity of\, $\PMC$}
      \mrow{
        =
        x_A\gets u_A;
        x_M\gets u_M;
        \alpha\gets A\prn{x_A};
        s\gets M\prn{x_M}\prn{\alpha};
        \Return~\prn{\alpha,s}
      }
      \commentrow{// by definition of\, $\ElabThe$}
      \mrow{
        = x_A\gets u_A;
        x_M\gets u_M;
        \ElabThe~A\prn{x_A}~M\prn{x_M}
      }
      \qedhere
    }
  }
\end{proof}

\begin{proof}[Case (multilinearity of $\underline{\ElabPair}$)]
  We fix the following maps and partial inputs
  \[
  M\colon X_M\to \ChkScript,
  N\colon X_N\to \ChkScript,
  u_M:\PMC\prn{X_M},
  u_N:\PMC\prn{X_N}
  \]
  to check that
  \[
  \ElabPair~
  \prn{x_M\gets u_M; M\prn{x_M}}~
  \prn{x_N\gets u_N; N\prn{x_N}}
  =
  x_M\gets u_M;
  x_N\gets u_N;
  \ElabPair~M\prn{x_M}~N\prn{x_N}
  \text{.}
  \]

  Because $\ChkScript$ is a pointwise product of algebras, two elements are equal if they are equal at every index. We fix an arbitrary target type $\Var{typ}:\Tp$ and calculate from left to right:

  \iblock{
    \mhang{
      \ElabPair~
      \prn{x_M\gets u_M; M\prn{x_M}}~
      \prn{x_N\gets u_N; N\prn{x_N}}
      ~\Var{typ}
    }{
      \commentrow{// by definition of\, $\ElabPair$}
      \ieqhang{
        \prn{\Var{base},\Var{fam}}\gets \Await~\TpSg^{-1}\prn{\Var{typ}}
      }{
        \mrow{\Var{fst} \gets \prn[\big]{x_M\gets u_M; M\prn{x_M}}\prn{\Var{base}}}
        \mrow{\Var{snd} \gets \prn[\big]{x_N\gets u_N; N\prn{x_N}}\prn{\Var{fam}\prn{\Var{fst}}}}
        \mrow{\Return~\TmPair\prn{\Var{base},\Var{fam},\Var{fst},\Var{snd}}}
      }
      \commentrow{// by definition of\, $\ChkScript$ as a product of algebras}
      \ieqhang{\prn{\Var{base},\Var{fam}}\gets \Await~\TpSg^{-1}\prn{\Var{typ}}}{
        \mrow{\Var{fst} \gets \prn[\big]{x_M\gets u_M; M\prn{x_M}\prn{\Var{base}}}}
        \mrow{\Var{snd} \gets \prn[\big]{x_N\gets u_N; N\prn{x_N}\prn{\Var{fam}\prn{\Var{fst}}}}}
        \mrow{\Return~\TmPair\prn{\Var{base},\Var{fam},\Var{fst},\Var{snd}}}
      }
      \commentrow{// by associativity and commutativity of\, $\PMC$}
      \ieqhang{x_M\gets u_M; x_N\gets u_N}{
        \mrow{\prn{\Var{base},\Var{fam}}\gets \Await~\TpSg^{-1}\prn{\Var{typ}}}
        \mrow{\Var{fst} \gets M\prn{x_M}\prn{\Var{base}}}
        \mrow{\Var{snd} \gets N\prn{x_N}\prn{\Var{fam}\prn{\Var{fst}}}}
        \mrow{\Return~\TmPair\prn{\Var{base},\Var{fam},\Var{fst},\Var{snd}}}
      }
      \commentrow{// by definition of\, $\ElabPair$}
      \ieqrow{x_M\gets u_M; x_N\gets u_N; \ElabPair~M\prn{x_M}~N\prn{x_N}}
      \qedhere
    }
  }
\end{proof}

\begin{proof}[Case (multilinearity of $\underline{\ElabConv}$)]
  We fix a map and partial input
  \[
  E\colon X\to \SynScript,
  u:\PMC\prn{X}
  \]
  to check that
  \[
  \ElabConv~\prn{x\gets u; E\prn{x}}
  =
  x\gets u;
  \ElabConv~E\prn{x}
  \text{.}
  \]

  Because $\ChkScript$ is a pointwise product of algebras, two elements are
  equal if they are equal at every index. We fix an arbitrary target type
  $\tau:\Tp$ and calculate from left to right:

  \iblock{
    \mhang{\ElabConv~\prn{x\gets u; E\prn{x}}~\tau}{
      \commentrow{// by definition of\, $\ElabConv$}
      \mrow{=
        \prn{\sigma,s}\gets \prn{x\gets u; E\prn{x}};
        \_\gets\Await~\sigma =_{\Tp} \tau;
        \Return~s
      }
      \commentrow{// by associativity of\, $\PMC$}
      \mrow{
        = x\gets u;
        \prn{\sigma,s}\gets E\prn{x};
        \_\gets\Await~\sigma =_{\Tp} \tau;
        \Return~s
      }
      \commentrow{// by definition of\, $\ElabConv$}
      \mrow{
        = x\gets u; \ElabConv~E\prn{x}\prn{\tau}
      }
      \commentrow{// by definition of\, $\ChkScript$ as a product of algebras}
      \mrow{
        = \prn[\big]{x\gets u; \ElabConv~E\prn{x}}\prn{\tau}
      }
      \qedhere
    }
  }
\end{proof}

\begin{proof}[Case (multilinearity of $\underline{\ElabApp}$)]
  We fix the following maps and partial inputs
  \[
  E\colon X_E\to \SynScript,
  M\colon X_M\to \ChkScript,
  u_E:\PMC\prn{X_E},
  u_M:\PMC\prn{X_M}
  \]
  to check that
  \[
  \ElabApp~
  \prn{x_E\gets u_E; E\prn{x_E}}~
  \prn{x_M\gets u_M; M\prn{x_M}}
  =
  x_E\gets u_E;
  x_M\gets u_M;
  \ElabApp~E\prn{x_E}~M\prn{x_M}
  \text{.}
  \]

  We calculate from left to right:

  \iblock{
    \mhang{
      \ElabApp~
        \prn{x_E\gets u_E; E\prn{x_E}}~
        \prn{x_M\gets u_M; M\prn{x_M}}
    }{
      \commentrow{// by definition of\, $\ElabApp$}
      \ieqhang{\prn{\Var{typ},\Var{fun}}\gets \prn{x_E\gets u_E; E\prn{x_E}}}{
        \mrow{\prn{\Var{dom},\Var{fam}}\gets \Await~\TpPi^{-1}\prn{\Var{typ}}}
        \mrow{\Var{arg}\gets \prn[\big]{x_M\gets u_M; M\prn{x_M}}\prn{\Var{dom}}}
        \mrow{\Return~\prn{\Var{fam}\prn{\Var{arg}},\TmApp\prn{\Var{dom},\Var{fam},\Var{fun},\Var{arg}}}}
      }
      \commentrow{// by definition of\, $\ChkScript$ as a product of algebras}
      \ieqhang{\prn{\Var{typ},\Var{fun}}\gets \prn{x_E\gets u_E; E\prn{x_E}}}{
        \mrow{\prn{\Var{dom},\Var{fam}}\gets \Await~\TpPi^{-1}\prn{\Var{typ}}}
        \mrow{\Var{arg}\gets \prn{x_M\gets u_M; M\prn{x_M}\prn{\Var{dom}}}}
        \mrow{\Return~\prn{\Var{fam}\prn{\Var{arg}},\TmApp\prn[\big]{\Var{dom},\Var{fam},\Var{fun},\Var{arg}}}}
      }
      \commentrow{// by associativity and commutativity of\, $\PMC$}
      \ieqhang{x_E\gets u_E; x_M\gets u_M}{
        \mrow{\prn{\Var{typ},\Var{fun}}\gets E\prn{x_E}}
        \mrow{\prn{\Var{dom},\Var{fam}}\gets \Await~\TpPi^{-1}\prn{\Var{typ}}}
        \mrow{\Var{arg}\gets M\prn{x_M}\prn{\Var{dom}}}
        \mrow{\Return~\prn{\Var{fam}\prn{\Var{arg}},\TmApp\prn{\Var{dom},\Var{fam},\Var{fun},\Var{arg}}}}
      }
      \commentrow{// by definition of\, $\ElabApp$}
      \ieqrow{x_E\gets u_E; x_M\gets u_M; \ElabApp~E\prn{x_E}~M\prn{x_M}}
      \qedhere
    }
  }

\end{proof}

\begin{proof}[Case (multilinearity of $\underline{\ElabFst}$ and $\underline{\ElabSnd}$)]
  We fix
  $E\colon X\to\SynScript$ and $u:\PMC\prn{X}$ and calculate:

  \iblock{
    \mhang{\ElabFst~\prn{x\gets u; E\prn{x}}}{
      \commentrow{// by definition of\, $\ElabFst$}
      \ieqhang{\prn{\Var{typ},\Var{pair}}\gets \prn{x\gets u; E\prn{x}}}{
        \mrow{\prn{\Var{base},\Var{fam}}\gets \Await~\TpSg^{-1}\prn{\Var{typ}}}
        \mrow{\Return~\prn{\Var{base},\TmFst\prn{\Var{base},\Var{fam},\Var{pair}}}}
      }
      \commentrow{// by associativity of\, $\PMC$}
      \ieqhang{x\gets u}{
        \mrow{\prn{\Var{typ},\Var{pair}}\gets E\prn{x}}
        \mrow{\prn{\Var{base},\Var{fam}}\gets \Await~\TpSg^{-1}\prn{\Var{typ}}}
        \mrow{\Return~\prn{\Var{base},\TmFst\prn{\Var{base},\Var{fam},\Var{pair}}}}
      }
      \commentrow{// by definition of\, $\ElabFst$}
      \ieqrow{x\gets u; \ElabFst~E\prn{x}}
    }
  }

  The computation for $\ElabSnd$ is identical apart from the value
  returned.
\end{proof}

\begin{proof}[Case (multilinearity of $\underline{\ElabId}_l$)]
  We fix the following maps and partial inputs
  \[
  E\colon X_E\to \SynScript,
  M\colon X_M\to \ChkScript,
  u_E:\PMC\prn{X_E},
  u_M:\PMC\prn{X_M}
  \]
  to check that
  \[
  \ElabId_l~
  \prn{x_E\gets u_E; E\prn{x_E}}~
  \prn{x_M\gets u_M; M\prn{x_M}}
  =
  x_E\gets u_E;
  x_M\gets u_M;
  \ElabId_l~E\prn{x_E}~M\prn{x_M}
  \text{.}
  \]

  We compute from left to right:

  \iblock{
    \mhang{
      \ElabId_l~
      \prn{x_E\gets u_E; E\prn{x_E}}~
      \prn{x_M\gets u_M; M\prn{x_M}}
    }{
      \commentrow{// by definition of $\ElabId_l$}
      \ieqhang{\prn{\Var{typ},\Var{lhs}}\gets \prn[\big]{x_E\gets u_E; E\prn{x_E}}}{
        \mrow{\Var{rhs}\gets \prn[\big]{x_M\gets u_M; M\prn{x_M}}\prn{\Var{typ}}}
        \mrow{\Return~\TpId\prn{\Var{typ},\Var{lhs},\Var{rhs}}}
      }
      \commentrow{// by definition of $\ChkScript$ as a product of algebras}
      \ieqhang{\prn{\Var{typ},\Var{lhs}}\gets \prn[\big]{x_E\gets u_E; E\prn{x_E}}}{
        \mrow{\Var{rhs}\gets \prn[\big]{x_M\gets u_M; M\prn{x_M}\prn{\Var{typ}}}}
        \mrow{\Return~\TpId\prn{\Var{typ},\Var{lhs},\Var{rhs}}}
      }
      \commentrow{// by associativity and commutativity of\, $\PMC$}
      \ieqhang{x_E\gets u_E; x_M\gets u_M}{
        \mrow{\prn{\Var{typ},\Var{lhs}}\gets E\prn{x_E}}
        \mrow{\Var{rhs}\gets M\prn{x_M}\prn{\Var{typ}}}
        \mrow{\Return~\TpId\prn{\Var{typ},\Var{lhs},\Var{rhs}}}
      }
      \commentrow{// by definition of $\ElabId_l$}
      \ieqrow{x_E\gets u_E; x_M\gets u_M; \ElabId_l~E\prn{x_E}~M\prn{x_M}}
      \qedhere
    }
  }
\end{proof}

\begin{proof}[Case (multilinearity of $\underline{\ElabId}$)]
  We fix the following maps and partial inputs
  \[
  A\colon X_A\to \TypScript,
  M\colon X_M\to \ChkScript,
  N\colon X_N\to \ChkScript,
  u_A:\PMC\prn{X_A},
  u_M:\PMC\prn{X_M},
  u_N:\PMC\prn{X_N}
  \]
  to check that
  \[
  \begin{aligned}
    &\ElabId~
    \prn{x_A\gets u_A; A\prn{x_A}}~
    \prn{x_M\gets u_M; M\prn{x_M}}~
    \prn{x_N\gets u_N; N\prn{x_N}}
    \\
    &\quad=
    x_A\gets u_A;
    x_M\gets u_M;
    x_N\gets u_N;
    \ElabId~A\prn{x_A}~M\prn{x_M}~N\prn{x_N}
    \text{.}
  \end{aligned}
  \]

  We calculate from left to right:

  \iblock{
    \mhang{
      \ElabId~
      \prn{x_A\gets u_A; A\prn{x_A}}~
      \prn{x_M\gets u_M; M\prn{x_M}}~
      \prn{x_N\gets u_N; N\prn{x_N}}
    }{
      \commentrow{// by Lemma~\ref{lem:id-in-terms-of-idl}}
      \ieqrow{
        \ElabId_l~\prn[\Big]{
          \ElabThe~\prn[\big]{
            x_A\gets u_A; A\prn{x_A}
          }~\prn[\big]{x_M\gets u_M; M\prn{x_M}}
        }
        ~\prn[\big]{x_N\gets u_N; N\prn{x_N}}
      }
      \commentrow{// by multilinearity of\, $\ElabThe$}
      \ieqrow{
        \ElabId_l~\prn[\big]{
          x_A\gets u_A; x_M\gets u_M;
          \ElabThe~A\prn{x_A}~M\prn{x_M}
        }
        ~\prn[\big]{x_N\gets u_N; N\prn{x_N}}
      }
      \commentrow{// by multilinearity of\, $\ElabId_l$}
      \ieqrow{
        x_A\gets u_A; x_M\gets u_M; x_N\gets u_N; \ElabId_l~\prn[\big]{\ElabThe~A\prn{x_A}~M\prn{x_M}}~N\prn{x_N}
      }
      \commentrow{// by Lemma~\ref{lem:id-in-terms-of-idl}}
      \ieqrow{
        x_A\gets u_A;
        x_M\gets u_M;
        x_N\gets u_N;
        \ElabId~A\prn{x_A}~M\prn{x_M}~N\prn{x_N}
      }
      \qedhere
    }
  }
\end{proof}

\begin{proof}[Case (multilinearity of $\underline{\ElabRefl}'$)]
  We fix $E\colon X\to\SynScript$ and $u:\PMC\prn{X}$ and calculate:

  \iblock{
    \mhang{
      \ElabRefl'~\prn{x\gets u; E\prn{x}}
    }{
      \commentrow{// by definition of\, $\ElabRefl'$}
      \ieqrow{
        \prn{\Var{typ},\Var{tm}}\gets \prn[\big]{x\gets u; E\prn{x}};
        \Return~\prn{\TpId\prn{\Var{typ},\Var{tm},\Var{tm}},\TmRefl\prn{\Var{typ},\Var{tm}}}
      }
      \commentrow{// by associativity of\, $\PMC$}
      \ieqrow{
        x\gets u;
        \prn{\Var{typ},\Var{tm}}\gets E\prn{x};
        \Return~\prn{\TpId\prn{\Var{typ},\Var{tm},\Var{tm}},\TmRefl\prn{\Var{typ},\Var{tm}}}
      }
      \commentrow{// by definition of\, $\ElabRefl'$}
      \ieqrow{
        x\gets u; \ElabRefl'~E\prn{x}
      }
      \qedhere
    }
  }
\end{proof}

\begin{proof}[Case (multilinearity of $\underline{\ElabIdElim}$)]
  We fix the following maps and partial inputs:
  \[
  E\colon X_E\to \SynScript,
  C\colon X_C\to \HypScript^3\pitchfork\TypScript,
  M\colon X_M\to \HypScript\pitchfork\ChkScript,
  u_E:\PMC\prn{X_E},
  u_C:\PMC\prn{X_C},
  u_M:\PMC\prn{X_M}
  \]
  to check that
  \begin{align*}
    &\ElabIdElim~\prn{x_E\gets u_E; E\prn{x_E}}~\prn{x_C\gets u_C; C\prn{x_C}}~\prn{x_M\gets u_M; M\prn{x_M}}
    \\
    &\quad=
    x_E\gets u_E;\ x_C\gets u_C;\ x_M\gets u_M;\ \ElabIdElim~E\prn{x_E}~C\prn{x_C}~M\prn{x_M}\text{.}
  \end{align*}

  Write $\widetilde{J}\coloneq\prn{a:\Tm\prn{\Var{typ}}}\times\prn{b:\Tm\prn{\Var{typ}}}\times\Tm\prn{\TpId\prn{\Var{typ},a,b}}$, which is compact by two applications of Lemma~\ref{lem:compact-sum}. Then $\prn{\Var{lhs},\Var{rhs},\Var{target}}:\widetilde{J}$, and this inhabits the motive's scope, while $\Var{lhs}:\Tm\prn{\Var{typ}}$ inhabits the baseCase's scope.

  We calculate from left to right:

  \iblock{
    \mhang{
      \ElabIdElim~\prn{x_E\gets u_E; E\prn{x_E}}~\prn{x_C\gets u_C; C\prn{x_C}}~\prn{x_M\gets u_M; M\prn{x_M}}
    }{
      \commentrow{// by definition of\, $\ElabIdElim$, Lemma~\ref{lem:theta-comp} (merging the nested $\Scope$s into one over $\widetilde{J}$), and definition of product algebras for $\HypScript^3\pitchfork\TypScript$ and $\HypScript\pitchfork\ChkScript$}
      \ieqhang{\prn{\Var{idtyp},\Var{target}}\gets \prn{x_E\gets u_E; E\prn{x_E}}}{
        \mrow{\prn{\Var{typ},\Var{lhs},\Var{rhs}}\gets \Await~\TpId^{-1}\prn{\Var{idtyp}}}
        \mrow{\Var{motive}\gets \Scope~\prn{a,b,p}:\widetilde{J}~\In~x_C\gets u_C; C\prn{x_C}\prn{\prn{\Var{typ},a}, \prn{\Var{typ},b}, \prn{\TpId\prn{\Var{typ},a,b},p}}}
        \mrow{\Var{baseCase}\gets \Scope~a:\Tm\prn{\Var{typ}}~\In~x_M\gets u_M; M\prn{x_M}\prn{a}\prn{\Var{motive}\prn{a,a,\TmRefl\prn{\Var{typ},a}}}}
          \mrow{
            \Return~\prn{
              \Var{motive}\prn{\Var{lhs},\Var{rhs},\Var{target}},
              \TmJ\prn{\Var{typ},\Var{lhs},\Var{rhs},\Var{target},\Var{motive},\Var{baseCase}}
            }
          }
      }
      \commentrow{// by Lemma~\ref{lem:inhabited-scope} (using the inhabitants fixed above), and associativity and commutativity of\, $\PMC$}
      \ieqhang{x_E\gets u_E;\ x_C\gets u_C;\ x_M\gets u_M}{
        \mrow{\prn{\Var{idtyp},\Var{target}}\gets E\prn{x_E}}
        \mrow{\prn{\Var{typ},\Var{lhs},\Var{rhs}}\gets \Await~\TpId^{-1}\prn{\Var{idtyp}}}
        \mrow{\Var{motive}\gets \Scope~\prn{a,b,p}:\widetilde{J}~\In~C\prn{x_C}\prn{\prn{\Var{typ},a}, \prn{\Var{typ},b}, \prn{\TpId\prn{\Var{typ},a,b},p}}}
        \mrow{\Var{baseCase}\gets \Scope~a:\Tm\prn{\Var{typ}}~\In~M\prn{x_M}\prn{a}\prn{\Var{motive}\prn{a,a,\TmRefl\prn{\Var{typ},a}}}}
        \mrow{\Return~\prn{\Var{motive}\prn{\Var{lhs},\Var{rhs},\Var{target}},\TmJ\prn{\Var{typ},\Var{lhs},\Var{rhs},\Var{target},\Var{motive},\Var{baseCase}}}}
      }
      \commentrow{// by definition of\, $\ElabIdElim$}
      \ieqrow{
        x_E\gets u_E; x_C\gets u_C; x_M\gets u_M; \ElabIdElim~E\prn{x_E}~C\prn{x_C}~M\prn{x_M}
      }
      \qedhere
    }
  }
\end{proof}

Now suppose that the judgemental structure satisfies strengthening.

\begin{proof}[Case (multilinearity of $\underline{\ElabPi}$)]
  We fix the following maps and partial inputs:
  \[
    A\colon X_A\to \TypScript,
    B\colon X_B\to \HypScript\pitchfork\TypScript,
    u_A:\PMC\prn{X_A},
    u_B:\PMC\prn{X_B}
  \]
  to check that
  \[
    \ElabPi~\prn{x_A\gets u_A; A\prn{x_A}}~\prn{x_B\gets u_B; B\prn{x_B}} =
    x_A\gets u_A; x_B \gets u_B; \ElabPi~A\prn{x_A}~B\prn{x_B}\text{.}
  \]

  We calculate from left to right:

  \iblock{
    \mhang{
      \ElabPi~\prn{x_A\gets u_A; A\prn{x_A}}~\prn{x_B\gets u_B; B\prn{x_B}}
    }{
      \commentrow{// by definition of\, $\ElabPi$}
      \ieqhang{\Var{dom}\gets\prn[\big]{x_A\gets u_A; A\prn{x_A}}}{
        \mrow{\Var{fam}\gets\Scope~a:\Tm\prn{\Var{dom}}~\In~\prn[\big]{x_B\gets u_B; B\prn{x_B}}\prn{a}}
        \mrow{\Return~\TpPi\prn{\Var{dom},\Var{fam}}}
      }
      \commentrow{// because $\HypScript\pitchfork\TypScript$ is a pointwise product of algebras}
      \ieqhang{\Var{dom}\gets\prn[\big]{x_A\gets u_A; A\prn{x_A}}}{
        \mrow{\Var{fam}\gets\Scope~a:\Tm\prn{\Var{dom}}~\In~\prn[\big]{x_B\gets u_B; B\prn{x_B}\prn{a}}}
        \mrow{\Return~\TpPi\prn{\Var{dom},\Var{fam}}}
      }
      \commentrow{// by associativity of $\PMC$ and Lemma~\ref{lem:scope-extension}}
      \ieqhang{x_A\gets u_A}{
        \mrow{\Var{dom}\gets A\prn{x_A}}
        \mrow{x_B\gets \Scope~a:\Tm\prn{\Var{dom}}~\In~u_B}
        \mrow{\Var{fam}\gets\Scope~a:\Tm\prn{\Var{dom}}~\In~B\prn{x_B\prn{a}}\prn{a}}
        \mrow{\Return~\TpPi\prn{\Var{dom},\Var{fam}}}
      }
      \commentrow{// by Lemma~\ref{lem:strengthened-scope} (with $c\prn{x}\prn{a}\coloneq\Return~x$) and commutativity of\, $\PMC$}
      \ieqhang{x_A\gets u_A; x_B\gets u_B}{
        \mrow{\Var{dom}\gets A\prn{x_A}}
        \mrow{\Var{fam}\gets\Scope~a:\Tm\prn{\Var{dom}}~\In~B\prn{x_B}\prn{a}}
        \mrow{\Return~\TpPi\prn{\Var{dom},\Var{fam}}}
      }
      \commentrow{// by definition of\, $\ElabPi$}
      \ieqrow{x_A\gets u_A; x_B \gets u_B; \ElabPi~A\prn{x_A}~B\prn{x_B}}
      \qedhere
    }
  }
\end{proof}

The case for $\underline{\ElabSg}$ is identical.

\begin{proof}[Case (multilinearity of $\underline{\ElabLam}$)]
  We fix the following maps and partial inputs:
  \[
    M\colon X_M\to \HypScript\pitchfork\ChkScript, u_M:\PMC\prn{X_M}
  \]
  to check that
  \[
    \ElabLam~\prn{x_M\gets u_M; M\prn{x_M}} =
    x_M\gets u_M; \ElabLam~M\prn{x_M}
  \]

    We calculate from left to right:

  \iblock{
    \mhang{
      \ElabLam~\prn{x_M\gets u_M; M\prn{x_M}}~\Var{typ}
    }{
      \commentrow{// by definition of\, $\ElabLam$}
      \ieqhang{\prn{\Var{dom},\Var{fam}}\gets\Await~\TpPi^{-1}\prn{\Var{typ}}}{
        \mrow{\Var{body}\gets \Scope~a:\Tm\prn{\Var{dom}}~\In~\prn[\big]{x_M\gets u_M; M\prn{x_M}}\prn{\Var{dom},a}\prn{\Var{fam}\prn{a}}}
        \mrow{\Return~\TmLam\prn{\Var{dom},\Var{fam},\Var{body}}}
      }
      \commentrow{// because $\HypScript\pitchfork\ChkScript$ is a product of products of algebras}
      \ieqhang{\prn{\Var{dom},\Var{fam}}\gets\Await~\TpPi^{-1}\prn{\Var{typ}}}{
        \mrow{\Var{body}\gets \Scope~a:\Tm\prn{\Var{dom}}~\In~\prn[\big]{x_M\gets u_M; M\prn{x_M}\prn{\Var{dom},a}\prn{\Var{fam}\prn{a}}}}
        \mrow{\Return~\TmLam\prn{\Var{dom},\Var{fam},\Var{body}}}
      }
      \commentrow{// by Lemma~\ref{lem:scope-extension}}
      \ieqhang{\prn{\Var{dom},\Var{fam}}\gets\Await~\TpPi^{-1}\prn{\Var{typ}}}{
        \mrow{x_M\gets\Scope~a:\Tm\prn{\Var{dom}}~\In~u_M}
        \mrow{\Var{body}\gets \Scope~a:\Tm\prn{\Var{dom}}~\In~M\prn{x_M\prn{a}}\prn{\Var{dom},a}\prn{\Var{fam}\prn{a}}}
        \mrow{\Return~\TmLam\prn{\Var{dom},\Var{fam},\Var{body}}}
      }
      \commentrow{// by Lemma~\ref{lem:strengthened-scope} (with $c\prn{x}\prn{a}\coloneq\Return~x$) and commutativity of\, $\PMC$}
      \ieqhang{x_M\gets u_M}{
        \mrow{\prn{\Var{dom},\Var{fam}}\gets\Await~\TpPi^{-1}\prn{\Var{typ}}}
        \mrow{\Var{body}\gets \Scope~a:\Tm\prn{\Var{dom}}~\In~M\prn{x_M}\prn{\Var{dom},a}\prn{\Var{fam}\prn{a}}}
        \mrow{\Return~\TmLam\prn{\Var{dom},\Var{fam},\Var{body}}}
      }
      \commentrow{// by definition of\, $\ElabLam$}
      \ieqrow{x_M\gets u_M; \ElabLam~M\prn{x_M}}
      \qedhere
    }
  }
\end{proof}

\begin{proof}[Case (multilinearity of $\underline{\ElabLam}'$)]
  We fix the following maps and partial inputs:
  \[
  A\colon X_A\to \TypScript,
  E\colon X_E\to \HypScript\pitchfork\SynScript,
  u_A:\PMC\prn{X_A},
  u_E:\PMC\prn{X_E}
  \]
  to check that
  \[
  \ElabLam'~\prn{x_A\gets u_A; A\prn{x_A}}~\prn{x_E\gets u_E; E\prn{x_E}}
  =
  x_A\gets u_A;
  x_E\gets u_E;
  \ElabLam'~A\prn{x_A}~E\prn{x_E}\text{.}
  \]

  We calculate from left to right:

  \iblock{
    \mhang{
      \ElabLam'~\prn{x_A\gets u_A; A\prn{x_A}}~\prn{x_E\gets u_E; E\prn{x_E}}
    }{
      \commentrow{// by definition of\, $\ElabLam'$}
      \ieqhang{\Var{dom}\gets \prn{x_A\gets u_A;A\prn{x_A}}}{
        \mrow{\Var{body}\gets \Scope~a:\Tm\prn{\Var{dom}}~\In~\prn[\big]{x_E\gets u_E;E\prn{x_E}}\prn{\Var{dom},a}}
        \mrow{
          \Return~\prn{
            \TpPi\prn{\Var{dom},\pi_1\circ\Var{body}},
            \TmLam\prn{\Var{dom},\pi_1\circ\Var{body},\pi_2\circ\Var{body}}
          }
        }
      }
      \commentrow{// by definition of $\HypScript\pitchfork\SynScript$ as a product of algebras}
      \ieqhang{\Var{dom}\gets \prn{x_A\gets u_A;A\prn{x_A}}}{
        \mrow{\Var{body}\gets \Scope~a:\Tm\prn{\Var{dom}}~\In~\prn[\big]{x_E\gets u_E;E\prn{x_E}\prn{\Var{dom},a}}}
        \mrow{\Return~\prn{\TpPi\prn{\Var{dom},\pi_1\circ\Var{body}},\TmLam\prn{\Var{dom},\pi_1\circ\Var{body},\pi_2\circ\Var{body}}}}
      }
      \commentrow{// by Lemma~\ref{lem:scope-extension}}
      \ieqhang{\Var{dom}\gets \prn{x_A\gets u_A;A\prn{x_A}}}{
        \mrow{x_E\gets \Scope~a:\Tm\prn{\Var{dom}}~\In~u_E}
        \mrow{\Var{body}\gets \Scope~a:\Tm\prn{\Var{dom}}~\In~E\prn{x_E\prn{a}}\prn{\Var{dom},a}}
        \mrow{\Return~\prn{\TpPi\prn{\Var{dom},\pi_1\circ\Var{body}},\TmLam\prn{\Var{dom},\pi_1\circ\Var{body},\pi_2\circ\Var{body}}}}
      }
      \commentrow{// by Lemma~\ref{lem:strengthened-scope}}
      \ieqhang{\Var{dom}\gets \prn{x_A\gets u_A;A\prn{x_A}}}{
        \mrow{x_E\gets u_E}
        \mrow{\Var{body}\gets \Scope~a:\Tm\prn{\Var{dom}}~\In~E\prn{x_E}\prn{\Var{dom},a}}
        \mrow{\Return~\prn{\TpPi\prn{\Var{dom},\pi_1\circ\Var{body}},\TmLam\prn{\Var{dom},\pi_1\circ\Var{body},\pi_2\circ\Var{body}}}}
      }
      \commentrow{// associativity and commutativity of\, $\PMC$}
      \ieqhang{x_A\gets u_A}{
        \mrow{x_E\gets u_E}
        \mrow{\Var{dom}\gets A\prn{x_A}}
        \mrow{\Var{body}\gets \Scope~a:\Tm\prn{\Var{dom}}~\In~E\prn{x_E}\prn{\Var{dom},a}}
        \mrow{\Return~\prn{\TpPi\prn{\Var{dom},\pi_1\circ\Var{body}},\TmLam\prn{\Var{dom},\pi_1\circ\Var{body},\pi_2\circ\Var{body}}}}
      }
      \commentrow{// by definition of\, $\ElabLam'$}
      \ieqrow{x_A\gets u_A; x_E\gets u_E; \ElabLam'~A\prn{x_A}~E\prn{x_E}}
      \qedhere
    }
  }
\end{proof}

\section{Proof of Theorem~\ref{thm:subject-reduction}: ``subject reduction''}\label{appendix:subject-reduction}

\begin{proof}[Case (annotation removal)]
  Fixing $A:\TypScript$ and $E:\SynScript$, we must check \[
  \ElabThe~A~\prn{\ElabConv~E}
  \preccurlyeq_{\SynScript}
  E
  \text{.}
  \]

  We calculate:

  \iblock{
    \mhang{
      \ElabThe~A~\prn{\ElabConv~E}
    }{
      \ieqrow{
        \alpha\gets A; \prn{\beta,b}\gets E; \Await~\alpha=\beta; \Return~\prn{\alpha,b}
      }
      \ieqrow{
        \alpha\gets A; \prn{\beta,b}\gets E; \Await~\alpha=\beta; \Return~\prn{\beta,b}
      }
      \mrow{
        \preccurlyeq_{\SynScript}
        E
      }
      \qedhere
    }
  }
\end{proof}

\begin{proof}[Case (function application)]
  Fixing $A:\TypScript$, $B\colon\HypScript\to\TypScript$, $M\colon \HypScript\to\ChkScript$, and $N:\ChkScript$, we must verify the following inequation:
  \[
    \ElabApp~\prn{
      \ElabThe~\prn{
        \ElabPi~A~B
      }~\prn{
        \ElabLam~M
      }
    }~{N}
    \preccurlyeq_{\SynScript}
    x \gets \ElabThe~A~N;
    \ElabThe~{B\prn{x}}~{M\prn{x}}
  \]

  We calculate:

  \iblock{
    \mhang{
      \ElabApp~\prn{
        \ElabThe~\prn{
          \ElabPi~A~B
        }~\prn{
          \ElabLam~M
        }
      }~{N}
    }{
      \ieqhang{
        \prn{\Var{funtyp},\Var{fun}}\gets \ElabThe~\prn{\ElabPi~A~B}~\prn{\ElabLam~M}
      }{
        \mrow{\prn{\Var{dom},\Var{fam}}\gets \Await~\TpPi^{-1}\prn{\Var{funtyp}}}
        \mrow{\Var{arg}\gets N(\Var{dom})}
        \mrow{
          \Return~\prn{\Var{fam}\prn{\Var{arg}},\TmApp\prn{\Var{dom},\Var{fam},\Var{fun},\Var{arg}}}
        }
      }
      \ieqhang{\Var{dom}\gets A}{
        \mrow{\Var{fam}\gets \Scope~x:\Tm\prn{\Var{dom}}~\In~B\prn{\Var{dom},x}}
        \mrow{\Var{fun}\gets \ElabLam~M~\TpPi\prn{\Var{dom},\Var{fam}}}
        \mrow{\Var{arg}\gets N(\Var{dom})}
        \mrow{
          \Return~\prn{\Var{fam}\prn{\Var{arg}},\TmApp\prn{\Var{dom},\Var{fam},\Var{fun},\Var{arg}}}
        }
      }
      \ieqhang{\Var{dom}\gets A}{
        \mrow{\Var{fam}\gets \Scope~x:\Tm\prn{\Var{dom}}~\In~B\prn{\Var{dom},x}}
        \mrow{\Var{body}\gets \Scope~x:\Tm\prn{\Var{dom}}~\In~M\prn{\Var{dom},x}\prn{\Var{fam}\prn{x}}}
        \mrow{\Var{arg}\gets N(\Var{dom})}
        \mrow{\Return~\prn{\Var{fam}\prn{\Var{arg}},\TmApp\prn{\Var{dom},\Var{fam},\TmLam\prn{\Var{dom},\Var{fam},\Var{body}},\Var{arg}}}}
      }
      \ieqhang{\Var{dom}\gets A}{
        \mrow{\Var{fam}\gets \Scope~x:\Tm\prn{\Var{dom}}~\In~B\prn{\Var{dom},x}}
        \mrow{\Var{body}\gets \Scope~x:\Tm\prn{\Var{dom}}~\In~M\prn{\Var{dom},x}\prn{\Var{fam}\prn{x}}}
        \mrow{\Var{arg}\gets N(\Var{dom})}
        \mrow{\Return~\prn{\Var{fam}\prn{\Var{arg}},\Var{body}\prn{\Var{arg}}}}
      }
      \ieqhang{\Var{dom}\gets A}{
        \mrow{\Var{arg}\gets N(\Var{dom})}
        \mrow{\Var{fam}\gets \Scope~x:\Tm\prn{\Var{dom}}~\In~B\prn{\Var{dom},x}}
        \mrow{\Var{body}\gets \Scope~x:\Tm\prn{\Var{dom}}~\In~M\prn{\Var{dom},x}\prn{\Var{fam}\prn{x}}}
        \mrow{\Return~\prn{\Var{fam}\prn{\Var{arg}},\Var{body}\prn{\Var{arg}}}}
      }
      \ieqhang[$\phantom{=}\mathllap{\preccurlyeq}$ ]{\Var{dom}\gets A}{
        \mrow{\Var{arg}\gets N(\Var{dom})}
        \mrow{\Var{fam}\gets \Scope~x:\Tm\prn{\Var{dom}}~\In~B\prn{\Var{dom},x}}
        \mrow{\Var{tm}\gets M\prn{\Var{dom},\Var{arg}}\prn{\Var{fam}\prn{\Var{arg}}}}
        \mrow{\Return~\prn{\Var{fam}\prn{\Var{arg}},\Var{tm}}}
      }
      \ieqhang[$\phantom{=}\mathllap{\preccurlyeq}$ ]{
        \prn{\Var{dom},\Var{arg}}\gets \ElabThe~A~N
      }{
        \mrow{\Var{fibre}\gets B\prn{\Var{dom},\Var{arg}}}
        \mrow{\Var{tm}\gets M\prn{\Var{dom},\Var{arg}}\prn{\Var{fibre}}}
        \mrow{\Return~\prn{\Var{fibre},\Var{tm}}}
      }
      \ieqhang{x\gets \ElabThe~A~N}{
        \mrow{\Var{fibre}\gets B\prn{x}}
        \mrow{\Var{tm}\gets M\prn{x}\prn{\Var{fibre}}}
        \mrow{\Return~\prn{\Var{fibre},\Var{tm}}}
      }
      \ieqrow{x\gets \ElabThe~A~N; \ElabThe~B\prn{x}~M\prn{x}}
      \qedhere
    }
  }
\end{proof}

\begin{proof}[Case (first projection)]
  Fixing $A:\TypScript$, $B:\HypScript\to\TypScript$, and $M,N:\ChkScript$, we must verify the following inequation:
  \[
    \ElabFst~\prn{
      \ElabThe~\prn{\ElabSg~A~B}~\prn{\ElabPair~M~N}
    }
    \preccurlyeq_{\SynScript}
    \ElabThe~A~M
  \]

  We calculate as follows:

  \iblock{
    \mhang{
      \ElabFst~\prn{
        \ElabThe~\prn{\ElabSg~A~B}~\prn{\ElabPair~M~N}
      }
    }{
      \ieqhang{\prn{\Var{pairtyp}, \Var{pair}} \gets \ElabThe~\prn{\ElabSg~A~B}~\prn{\ElabPair~M~N}}{
        \mrow{\prn{\Var{dom},\Var{fam}}\gets \Await~\TpSg^{-1}\prn{\Var{pairtyp}}}
        \mrow{\Return~\prn{\Var{dom}, \TmFst\prn{\Var{dom},\Var{fam},\Var{pair}}}}
      }
      \ieqhang{\Var{dom}\gets A}{
        \mrow{\Var{fam}\gets\Scope~x:\Tm\prn{\Var{dom}}~\In~B\prn{\Var{dom},x}}
        \mrow{\Var{fst}\gets M\prn{\Var{dom}}}
        \mrow{\Var{snd}\gets N\prn{\Var{fam}\prn{\Var{fst}}}}
        \mrow{\Return\prn{\Var{dom}, \TmFst\prn{\Var{dom},\Var{fam},\TmPair\prn{\Var{dom},\Var{fam},\Var{fst},\Var{snd}}}}}
      }
      \ieqhang{\Var{dom}\gets A}{
        \mrow{\Var{fam}\gets\Scope~x:\Tm\prn{\Var{dom}}~\In~B\prn{\Var{dom},x}}
        \mrow{\Var{fst}\gets M\prn{\Var{dom}}}
        \mrow{\Var{snd}\gets N\prn{\Var{fam}\prn{\Var{fst}}}}
        \mrow{\Return~\prn{\Var{dom}, \Var{fst}}}
      }
      \ieqhang[$\phantom{=}\mathllap{\preccurlyeq}$ ]{\Var{dom}\gets A}{
        \mrow{\Var{fst}\gets M\prn{\Var{dom}}}
        \mrow{\Return~\prn{\Var{dom}, \Var{fst}}}
      }
      \ieqrow{\ElabThe~A~M}
      \qedhere
    }
  }
\end{proof}

\begin{proof}[Case (second projection)]
  Fixing $A:\TypScript$, $B:\HypScript\to\TypScript$, and $M,N:\ChkScript$, we must verify the following inequation:
  \[
    \ElabSnd~\prn{
      \ElabThe~\prn{\ElabSg~A~B}~\prn{\ElabPair~M~N}
    }
    \preccurlyeq_{\SynScript}
    \ElabThe~\prn{x\gets\ElabThe~A~M;B\prn{x}}~N
  \]

  We calculate as follows:

  \iblock{
    \mhang{
      \ElabSnd~\prn{
        \ElabThe~\prn{\ElabSg~A~B}~\prn{\ElabPair~M~N}
      }
    }{
      \ieqhang{
        \prn{\Var{pairtyp}, \Var{pair}} \gets \ElabThe~\prn{\ElabSg~A~B}~\prn{\ElabPair~M~N}
      }{
        \mrow{\prn{\Var{dom},\Var{fam}}\gets \Await~\TpSg^{-1}\prn{\Var{pairtyp}}}
        \mrow{\Return~\prn{\Var{fam}\prn{\TmFst\prn{\Var{dom},\Var{fam},\Var{pair}}}, \TmSnd\prn{\Var{dom},\Var{fam},\Var{pair}}}}
      }
      \ieqhang{\Var{dom}\gets A}{
        \mrow{\Var{fam}\gets \Scope~x:\Tm\prn{\Var{dom}}~\In~B\prn{\Var{dom},x}}
        \mrow{\Var{fst}\gets M\prn{\Var{dom}}}
        \mrow{\Var{snd}\gets N\prn{\Var{fam}\prn{\Var{fst}}}}
        \mhang{\Return}{
          \iprefix[$\bigl($]{
            \mrow{
              \Var{fam}\prn{
                \TmFst\prn{\Var{dom},\Var{fam},\TmPair\prn{\Var{dom},\Var{fam},\Var{fst},\Var{snd}}}
              },
            }
          }{
            \mrow{
              \TmSnd\prn{\Var{dom},\Var{fam},\TmPair\prn{\Var{dom},\Var{fam},\Var{fst},\Var{snd}}}\bigr)
            }
          }
        }
      }
      \ieqhang{\Var{dom}\gets A}{
        \mrow{\Var{fam}\gets \Scope~x:\Tm\prn{\Var{dom}}~\In~B\prn{\Var{dom},x}}
        \mrow{\Var{fst}\gets M\prn{\Var{dom}}}
        \mrow{\Var{snd}\gets N\prn{\Var{fam}\prn{\Var{fst}}}}
        \mrow{\Return~\prn{\Var{fam}\prn{\Var{fst}}, \Var{snd}}}
      }
      \ieqhang[$\phantom{=}\mathllap{\preccurlyeq}$ ]{x\gets \ElabThe~A~M}{
        \mrow{\Var{fibre}\gets B\prn{x}}
        \mrow{\Var{snd}\gets N\prn{\Var{fibre}}}
        \mrow{\Return~\prn{\Var{fibre}, \Var{snd}}}
      }
      \ieqrow{x\gets \ElabThe~A~M; \ElabThe~B\prn{x}~N}
      \commentrow{// By Theorem~\ref{thm:multilinearity-of-combinators}.}
      \ieqrow{\ElabThe~\prn{x\gets \ElabThe~A~M; B\prn{x}}~N}
      \qedhere
    }
  }
\end{proof}

\begin{proof}[Case (identification elimination)]
  Fixing $A:\TypScript$, $N:\ChkScript$, $B\colon\HypScript^3\to\TypScript$, and $M\colon\HypScript\to\ChkScript$, we must verify the following inequation:
  \begin{align*}
    &\ElabIdElim~\prn{\ElabThe~\prn{\ElabId~A~N~N}~\ElabRefl}~B~M
    \\
    &\quad\preccurlyeq_{\SynScript}
    x\gets\ElabThe~A~N;\ p\gets\ElabThe~\prn{\ElabId~A~N~N}~\ElabRefl;\ \ElabThe~B\prn{x,x,p}~M\prn{x}
  \end{align*}

  First we unfold the argument of $\ElabIdElim$, using that $\TpId^{-1}\prn{\TpId\prn{\Var{typ},\Var{lhs},\Var{rhs}}}$ is the singleton on $\prn{\Var{typ},\Var{lhs},\Var{rhs}}$:

  \iblock{
    \mhang{
      \ElabThe~\prn{\ElabId~A~N~N}~\ElabRefl
    }{
      \ieqhang{\Var{typ}\gets A}{
        \mrow{\Var{lhs}\gets N\prn{\Var{typ}}}
        \mrow{\Var{rhs}\gets N\prn{\Var{typ}}}
        \mrow{\Await~\Var{lhs}=\Var{rhs}}
        \mrow{\Return~\prn{\TpId\prn{\Var{typ},\Var{lhs},\Var{rhs}},\TmRefl\prn{\Var{typ},\Var{lhs}}}}
      }
    }
  }

  Next we calculate the left-hand side of the inequation, writing $\Var{motive}$ and $\Var{baseCase}$ for the two $\Scope$-bound computations in the definition of $\ElabIdElim$:

  \iblock{
    \mhang{
      \ElabIdElim~\prn{\ElabThe~\prn{\ElabId~A~N~N}~\ElabRefl}~B~M
    }{
      \commentrow{// by definition of $\ElabIdElim$ and the unfolding above}
      \ieqhang{\Var{typ}\gets A}{
        \mrow{\Var{lhs}\gets N\prn{\Var{typ}}}
        \mrow{\Var{rhs}\gets N\prn{\Var{typ}}}
        \mrow{\Await~\Var{lhs}=\Var{rhs}}
        \mrow{\Var{motive}\gets\Scope~\prn{a,b,p}:\widetilde{J}~\In~B\prn[\big]{\prn{\Var{typ},a},\prn{\Var{typ},b},\prn{\TpId\prn{\Var{typ},a,b},p}}}
        \mrow{\Var{baseCase}\gets\Scope~a:\Tm\prn{\Var{typ}}~\In~M\prn{\Var{typ},a}\prn{\Var{motive}\prn{a,a,\TmRefl\prn{\Var{typ},a}}}}
        \mhang{\Return}{
          \iprefix[$\bigl($]{
            \mrow{\Var{motive}\prn{\Var{lhs},\Var{rhs},\TmRefl\prn{\Var{typ},\Var{lhs}}},}
          }{
            \mrow{
              \TmJ\prn{
                \Var{typ},
                \Var{lhs},
                \Var{rhs},
                \TmRefl\prn{\Var{typ},\Var{lhs}},\Var{motive},\Var{baseCase}
              }
              \bigr)
            }
          }
        }
      }
      \commentrow{// in the scope of $\Await~\Var{lhs}=\Var{rhs}$ we may replace $\Var{rhs}$ by $\Var{lhs}$; then use the computation rule for $\TmJ$}
      \ieqhang{\Var{typ}\gets A}{
        \mrow{\Var{lhs}\gets N\prn{\Var{typ}}}
        \mrow{\Var{rhs}\gets N\prn{\Var{typ}}}
        \mrow{\Await~\Var{lhs}=\Var{rhs}}
        \mrow{\Var{motive}\gets\cdots}
        \mrow{\Var{baseCase}\gets\cdots}
        \mrow{\Return~\prn{\Var{motive}\prn{\Var{lhs},\Var{lhs},\TmRefl\prn{\Var{typ},\Var{lhs}}},\Var{baseCase}\prn{\Var{lhs}}}}
      }
      \commentrow{// by Lemma~\ref{lem:duplication} applied to $N\prn{\Var{typ}}$, since the remainder depends on $\Var{rhs}$ only through $\Await~\Var{lhs}=\Var{rhs}$}
      \ieqhang{\Var{typ}\gets A}{
        \mrow{\Var{lhs}\gets N\prn{\Var{typ}}}
        \mrow{\Var{motive}\gets\cdots}
        \mrow{\Var{baseCase}\gets\cdots}
        \mrow{\Return~\prn{\Var{motive}\prn{\Var{lhs},\Var{lhs},\TmRefl\prn{\Var{typ},\Var{lhs}}},\Var{baseCase}\prn{\Var{lhs}}}}
      }
      \commentrow{// by instantiating the $\Scope$-bound computations at $\Var{lhs}$}
      \ieqhang[$\phantom{=}\mathllap{\preccurlyeq}$ ]{
        \Var{typ}\gets A
      }{
        \mrow{\Var{lhs}\gets N\prn{\Var{typ}}}
        \mrow{\Var{fibre}\gets B\prn[\big]{\prn{\Var{typ},\Var{lhs}},\prn{\Var{typ},\Var{lhs}},\prn{\TpId\prn{\Var{typ},\Var{lhs},\Var{lhs}},\TmRefl\prn{\Var{typ},\Var{lhs}}}}}
        \mrow{\Var{tm}\gets M\prn{\Var{typ},\Var{lhs}}\prn{\Var{fibre}}}
        \mrow{\Return~\prn{\Var{fibre},\Var{tm}}}
      }
    }
  }

  For the right-hand side, we unfold $\ElabThe$ and the unfolding of $\ElabThe~\prn{\ElabId~A~N~N}~\ElabRefl$ above, which now runs $A$ and $N$ afresh. By commutativity of $\PMC$ and Lemma~\ref{lem:duplication} (applied to $A$ and to $N\prn{\Var{typ}}$), all the copies of $\Var{typ}$ and $\Var{lhs}$ coincide and the assertion $\Await~\Var{lhs}=\Var{lhs}$ is trivial, so the right-hand side equals the last line above.
\end{proof}

\section{Proof of Theorem~\ref{thm:untyped-eta-reduction}: untyped \texorpdfstring{$\eta$}{eta}-contraction for functions}\label{appendix:untyped-eta-reduction}

Assume the judgemental structure satisfies strengthening.
Then for all $E:\SynScript$ and $\Var{funtype}:\Tp$, we wish to show the following equation (from which the inequation follows, because $\prn{\_\gets u;\ v}\preccurlyeq v$):
\[
  \ElabLam~\prn{
    \lambda x\mathpunct{.}
    \ElabConv~\prn{
      \ElabApp~E~\prn{\ElabConv~\prn{\ElabVar~x}}
    }
  }~\Var{funtype}
  =
  \prn{\_\gets\Await~\TpPi^{-1}\prn{\Var{funtype}};\ \ElabConv~E~\Var{funtype}}
  \text{.}
\]

Fixing $E:\SynScript$ and $\Var{funtype}:\Tp$, we calculate as follows:

\iblock{
  \mhang{
    \ElabLam~\prn{
      \lambda x\mathpunct{.}
      \ElabConv~\prn{
        \ElabApp~E~\prn{\ElabConv~\prn{\ElabVar~x}}
      }
    }~\Var{funtype}
  }{
    \ieqhang{\prn{\Var{dom},\Var{fam}}\gets\Await~\TpPi^{-1}\prn{\Var{funtype}}}{
      \mrow{
        \Var{body}\gets \Scope~x:\Tm\prn{\Var{dom}}~\In~
        {\ElabConv~\prn{\ElabApp~E~\prn{\ElabConv~\prn{\ElabVar~\prn{\Var{dom},x}}}}}~(\Var{fam}\prn{x})
      }
      \mrow{\Return~\TmLam\prn{\Var{dom},\Var{fam},\Var{body}}}
    }
  }
}

To calculate $\Var{body}$ above, we fix $x:\Tm\prn{\Var{dom}}$ and proceed as follows:

\iblock{
  \mhang{
    {\ElabConv~\prn{\ElabApp~E~\prn{\ElabConv~\prn{\ElabVar~\prn{\Var{dom},x}}}}}~(\Var{fam}\prn{x})
  }{
    \ieqhang{
      \prn{\Var{apptype}, \Var{app}}\gets \ElabApp~E~\prn{\ElabConv~\prn{\ElabVar~\prn{\Var{dom},x}}}
    }{
      \mrow{\Await~\Var{apptype} = \Var{fam}\prn{x}}
      \mrow{\Return~\Var{app}}
    }
    \ieqhang{
      \prn{\Var{funtype}',\Var{fun}}\gets E
    }{
      \mrow{\prn{\Var{dom}',\Var{fam}'}\gets \Await~\TpPi^{-1}\prn{\Var{funtype}'}}
      \mrow{\Await~\Var{dom}'=\Var{dom}}
      \mrow{\Await~\Var{fam}'\prn{x} = \Var{fam}\prn{x}}
      \mrow{
        \Return~{
          \TmApp\prn{
            \Var{dom}',
            \Var{fam}',
            \Var{fun},
            x
          }
        }
      }
    }
  }
}

Next we distribute the $\Scope$ over the above, making crucial use of the strengthening assumption, in the form of Lemma~\ref{lem:strengthened-scope}.

\iblock{
  \mhang{
    \Scope~x:\Tm\prn{\Var{dom}}~\In~
    {\ElabConv~\prn{\ElabApp~E~\prn{\ElabConv~\prn{\ElabVar~\prn{\Var{dom},x}}}}}~(\Var{fam}\prn{x})
  }{
    \ieqhang{
      \prn{\Var{funtype}',\Var{fun}}\gets E
    }{
      \mrow{\prn{\Var{dom}',\Var{fam}'}\gets \Await~\TpPi^{-1}\prn{\Var{funtype}'}}
      \mrow{\Await~\Var{dom}'=\Var{dom}}
      \mhang{\Scope~x:\Tm\prn{\Var{dom}}~\In}{
        \mrow{\Await~\Var{fam}'\prn{x} = \Var{fam}\prn{x}}
        \mrow{
          \Return~
          \TmApp\prn{
            \Var{dom}',
            \Var{fam}',
            \Var{fun},
            x
          }
        }
      }
    }
    \ieqhang{
      \prn{\Var{funtype}',\Var{fun}}\gets E
    }{
      \mrow{\prn{\Var{dom}',\Var{fam}'}\gets \Await~\TpPi^{-1}\prn{\Var{funtype}'}}
      \mrow{\Await~\Var{dom}'=\Var{dom}}
      \mrow{\Scope~x:\Tm\prn{\Var{dom}}~\In~ \Await~\Var{fam}'\prn{x} = \Var{fam}\prn{x}}
      \mrow{
        \Scope~x:\Tm\prn{\Var{dom}}~\In~\Return~
        \TmApp\prn{
          \Var{dom},
          \Var{fam},
          \Var{fun},
          x
        }
      }
    }
    \commentrow{// by function extensionality, and injectivity of $\TpPi$ (the preimage $\TpPi^{-1}\prn{\Var{funtype}'}$ is a singleton)}
    \ieqhang{
      \prn{\Var{funtype}',\Var{fun}}\gets E
    }{
      \mrow{\Await~\Var{funtype}' = \TpPi\prn{\Var{dom},\Var{fam}}}
      \mrow{
        \Return~\lambda x\mathpunct{.}
        \TmApp\prn{
          \Var{dom},
          \Var{fam},
          \Var{fun},
          x
        }
      }
    }
  }
}

We resume our calculation of the outer $\ElabLam$:

\iblock{
  \mhang{
    \ElabLam~\prn{
      \lambda x\mathpunct{.}
      \ElabConv~\prn{
        \ElabApp~E~\prn{\ElabConv~\prn{\ElabVar~x}}
      }
    }~\Var{funtype}
  }{
    \ieqhang{
      \prn{\Var{dom},\Var{fam}}\gets\Await~\TpPi^{-1}\prn{\Var{funtype}}
    }{
      \mrow{\Var{body}\gets \Scope~x:\Tm\prn{\Var{dom}}~\In~
      {\ElabConv~\prn{\ElabApp~E~\prn{\ElabConv~\prn{\ElabVar~\prn{\Var{dom},x}}}}}~(\Var{fam}\prn{x})}
      \mrow{\Return~\TmLam\prn{\Var{dom},\Var{fam},\Var{body}}}
    }
    \ieqhang{
      \prn{\Var{dom},\Var{fam}}\gets\Await~\TpPi^{-1}\prn{\Var{funtype}}
    }{
      \mrow{\prn{\Var{funtype}',\Var{fun}}\gets E}
      \mrow{\Await~\Var{funtype}' = \TpPi\prn{\Var{dom},\Var{fam}}}
      \mrow{\Return~\TmLam\prn{\Var{dom},\Var{fam},\lambda x\mathpunct{.}\TmApp\prn{\Var{dom},\Var{fam},\Var{fun},x}}}
    }
    \commentrow{// by the $\eta$-law for dependent functions, and $\Var{funtype}=\TpPi\prn{\Var{dom},\Var{fam}}$}
    \ieqhang{
      \prn{\Var{dom},\Var{fam}}\gets\Await~\TpPi^{-1}\prn{\Var{funtype}}
    }{
      \mrow{\prn{\Var{funtype}',\Var{fun}}\gets E}
      \mrow{\Await~\Var{funtype}' = \Var{funtype}}
      \mrow{\Return~\Var{fun}}
    }
    \ieqrow{
      \_\gets\Await~\TpPi^{-1}\prn{\Var{funtype}}; \ElabConv~E~\Var{funtype}
    }
  }
}

Since $\prn{\_\gets u;\ v}\preccurlyeq v$, it follows that
$
  \ElabLam~\prn{
    \lambda x\mathpunct{.}
    \ElabConv~\prn{
      \ElabApp~E~\prn{\ElabConv~\prn{\ElabVar~x}}
    }
  }
  \preccurlyeq_{\ChkScript}
  \ElabConv~E
$.

\section{Proof of Lemma~\ref{lem:re-annotation}: re-annotation}\label{appendix:re-annotation}

Fixing $A:\TypScript$ and $M:\ChkScript$, we calculate:

\iblock{
  \mhang{
    \ElabThe~A~\prn{\ElabConv~\prn{\ElabThe~A~M}}
  }{
    \commentrow{// by definition of $\ElabThe$ and $\ElabConv$}
    \ieqhang{\tau\gets A}{
      \mrow{\prn{\tau',s}\gets\prn[\big]{\tau''\gets A;\ s\gets M\prn{\tau''};\ \Return~\prn{\tau'',s}}}
      \mrow{\Await~\tau'=\tau}
      \mrow{\Return~\prn{\tau,s}}
    }
    \commentrow{// by associativity of $\PMC$}
    \ieqhang{\tau\gets A}{
      \mrow{\tau''\gets A}
      \mrow{s\gets M\prn{\tau''}}
      \mrow{\Await~\tau''=\tau}
      \mrow{\Return~\prn{\tau,s}}
    }
    \commentrow{// by Lemma~\ref{lem:duplication}}
    \ieqhang{\tau\gets A}{
      \mrow{s\gets M\prn{\tau}}
      \mrow{\Await~\tau=\tau}
      \mrow{\Return~\prn{\tau,s}}
    }
    \ieqrow{\ElabThe~A~M}
    \qedhere
  }
}

\section{Proof of Theorem~\ref{thm:typed-pair-eta}: \texorpdfstring{$\eta$}{eta}-equality for typed pairs}\label{appendix:typed-pair-eta}

Fixing $A:\TypScript$, $B\colon \HypScript\to\TypScript$, and $M:\ChkScript$, write $P\coloneq\ElabSg~A~B$ and consider the combinator
\[
  G\prn{P',S_1,S_2}\coloneq
  \ElabThe~P'~\prn{
    \ElabPair~\prn{\ElabConv~\prn{\ElabFst~S_1}}~\prn{\ElabConv~\prn{\ElabSnd~S_2}}
  }
  \colon\TypScript\times\SynScript\times\SynScript\to\SynScript\text{.}
\]
It is a composite of the combinators $\ElabThe$, $\ElabPair$, $\ElabConv$, $\ElabFst$ and $\ElabSnd$, each of which is multilinear unconditionally (Theorem~\ref{thm:multilinearity-of-combinators}), so $G$ is multilinear by Lemma~\ref{lem:composition}. The left-hand side of the equation is $G\prn{P,\ElabThe~P~M,\ElabThe~P~M}$.

For $\tau:\Tp$ we have $\ElabThe~\prn{\Return~\tau}~M=\prn{s\gets M\prn{\tau};\ \Return~\prn{\tau,s}}$, and $\ElabThe~P~M=\prn{\tau\gets P;\ \ElabThe~\prn{\Return~\tau}~M}$ by definition. Applying Corollary~\ref{cor:diagonal} to $G$ and the partial element $P$, and then to $G\prn{\Return~\tau,-,-}$ and the partial element $M\prn{\tau}$, we calculate:

\iblock{
  \mhang{
    \ElabThe~P~\prn{\ElabPair~\prn{\ElabConv~\prn{\ElabFst~\prn{\ElabThe~P~M}}}~\prn{\ElabConv~\prn{\ElabSnd~\prn{\ElabThe~P~M}}}}
  }{
    \commentrow{// Corollary~\ref{cor:diagonal}, with $u\coloneq P$ and $K\coloneq\mathbf{3}$}
    \ieqhang{\tau\gets P}{
      \mrow{G\prn{\Return~\tau, \ElabThe~\prn{\Return~\tau}~M, \ElabThe~\prn{\Return~\tau}~M}}
    }
    \commentrow{// Corollary~\ref{cor:diagonal}, with $u\coloneq M\prn{\tau}$ and $K\coloneq\mathbf{2}$}
    \ieqhang{\tau\gets P}{
      \mrow{s\gets M\prn{\tau}}
      \mrow{G\prn{\Return~\tau, \Return~\prn{\tau,s}, \Return~\prn{\tau,s}}}
    }
    \commentrow{// by definition of $P$ and the monad laws}
    \ieqhang{\Var{dom}\gets A}{
      \mrow{\Var{fam}\gets\Scope~x:\Tm\prn{\Var{dom}}~\In~B\prn{\Var{dom},x}}
      \mrow{s\gets M\prn{\TpSg\prn{\Var{dom},\Var{fam}}}}
      \mrow{G\prn{\Return~\TpSg\prn{\Var{dom},\Var{fam}}, \Return~\prn{\TpSg\prn{\Var{dom},\Var{fam}},s}, \Return~\prn{\TpSg\prn{\Var{dom},\Var{fam}},s}}}
    }
  }
}

It remains to compute the last line. Write $\tau\coloneq\TpSg\prn{\Var{dom},\Var{fam}}$. Unfolding $\ElabFst$ and $\ElabSnd$, and using that $\TpSg^{-1}\prn{\tau}$ is the singleton on $\prn{\Var{dom},\Var{fam}}$, we find
\begin{align*}
  \ElabFst~\prn{\Return~\prn{\tau,s}}&=\Return~\prn{\Var{dom},\TmFst\prn{\Var{dom},\Var{fam},s}}
  \\
  \ElabSnd~\prn{\Return~\prn{\tau,s}}&=\Return~\prn{\Var{fam}\prn{\TmFst\prn{\Var{dom},\Var{fam},s}},\TmSnd\prn{\Var{dom},\Var{fam},s}}\text{.}
\end{align*}
Hence the assertions arising from the two instances of $\ElabConv$ and from $\ElabPair$ are all of the form $\Await~\alpha=\alpha$, and we obtain:

\iblock{
  \mhang{
    G\prn{\Return~\tau,\ \Return~\prn{\tau,s},\ \Return~\prn{\tau,s}}
  }{
    \ieqrow{
      \Return~\prn{
        \tau,
        \TmPair\prn{
          \Var{dom},
          \Var{fam},
          \TmFst\prn{\Var{dom},\Var{fam},s},
          \TmSnd\prn{\Var{dom},\Var{fam},s}}
        }
    }
    \commentrow{// by the $\eta$-law for dependent pairs}
    \ieqrow{\Return~\prn{\tau,s}}
  }
}

\noindent
Therefore the left-hand side equals

\iblock{
  \mrow{\Var{dom}\gets A}
  \mrow{\Var{fam}\gets\Scope~x:\Tm\prn{\Var{dom}}~\In~B\prn{\Var{dom},x}}
  \mrow{\Var{s}\gets M\prn{\TpSg\prn{\Var{dom},\Var{fam}}}}
  \mrow{\Return~\prn{\TpSg\prn{\Var{dom},\Var{fam}},s}}
}

\noindent
which is $\ElabThe~\prn{\ElabSg~A~B}~M$ by definition.

\section{Proof of Theorem~\ref{thm:typed-function-eta}: \texorpdfstring{$\eta$}{eta}-equality for typed functions}\label{appendix:typed-function-eta}

Assume that the judgemental structure satisfies strengthening. Fixing $A:\TypScript$, $B\colon \HypScript\to\TypScript$, and $M:\ChkScript$, write $F\coloneq\ElabPi~A~B$ and $S\coloneq\ElabThe~F~M$. We calculate:

\iblock{
  \mhang{
    \ElabThe~~\prn{\ElabLam~\prn{\lambda x\mathpunct{.}\ElabConv~\prn{\ElabApp~S~\prn{\ElabConv~\prn{\ElabVar~x}}}}}
  }{
    \commentrow{// by definition of $\ElabThe$}
    \ieqhang{\tau\gets F}{
      \mrow{t\gets\ElabLam~\prn{\lambda x\mathpunct{.}\ElabConv~\prn{\ElabApp~S~\prn{\ElabConv~\prn{\ElabVar~x}}}}~\tau}
      \mrow{\Return~\prn{\tau,t}}
    }
    \commentrow{// by Theorem~\ref{thm:untyped-eta-reduction}}
    \ieqhang{\tau\gets F}{
      \mrow{\_\gets\Await~\TpPi^{-1}\prn{\tau}}
      \mrow{t\gets\ElabConv~S~\tau}
      \mrow{\Return~\prn{\tau,t}}
    }
    \commentrow{// because $F$ only produces dependent function types $\tau=\TpPi\prn{\Var{dom},\Var{fam}}$, for which $\Await~\TpPi^{-1}\prn{\tau}$ is trivially defined}
    \ieqhang{\tau\gets F}{
      \mrow{t\gets\ElabConv~S~\tau}
      \mrow{\Return~\prn{\tau,t}}
    }
    \commentrow{// by definition of $\ElabThe$}
    \ieqrow{
      \ElabThe~F~\prn{\ElabConv~\prn{\ElabThe~F~M}}
    }
    \commentrow{// by Lemma~\ref{lem:re-annotation}}
    \ieqrow{
      \ElabThe~F~M
    }
    \qedhere
  }
}

\end{document}